\documentclass[10pt]{article}
\usepackage[letterpaper,margin = 1in]{geometry}

\usepackage[none]{hyphenat}  
\usepackage{microtype}
\usepackage{bm}
\usepackage{dsfont}
\usepackage{mathrsfs}
\usepackage{amsmath}  
\usepackage{amsthm}
\usepackage{mathtools}
\usepackage{enumitem}
\usepackage{dsfont}
\usepackage[noend, boxruled, linesnumbered] {algorithm2e}
\SetAlCapSkip{1em}
\SetKwInput{Input}{Input}
\SetKwInput{Output}{Output}
\SetKwRepeat{Do}{do}{while}
\SetKwBlock{Init}{Initialization}{}

\usepackage{color}
\usepackage{float}
 
\usepackage{amssymb}
\usepackage{bm}
\usepackage{dsfont}
\usepackage{graphicx}
\graphicspath{{figures/}{./}}
\usepackage{mathrsfs}
\usepackage[symbol]{footmisc}
\usepackage[font=small]{caption}
\usepackage{titlesec}
\titleformat{\section}{\large\bfseries}{\thesection}{1em}{}
\titleformat{\subsection}{\normalsize\bfseries}{\thesubsection}{1em}{}
\renewenvironment{abstract}
 {\par\noindent\textbf{\abstractname}\ \ignorespaces}
 {\par\medskip}

\usepackage{multirow}
\usepackage{adjustbox}

\usepackage{subcaption}
\usepackage{siunitx}
\usepackage{subfiles}

\usepackage[authoryear,round]{natbib}
\usepackage{hyperref}

\makeatletter
\renewcommand{\maketitle}{\bgroup\setlength{\parindent}{0pt}
\begin{flushleft}
  \Large \textbf{\@title}
    \vspace{11pt} \\
  \normalsize \@author
\end{flushleft}\egroup
}
\makeatother

\titlespacing*{\section}{0pt}{15pt}{10pt}
\titlespacing*{\subsection}{0pt}{15pt}{10pt}
\newcommand{\aref}[1]{Appendix~\ref{#1}}
\newcommand{\eref}[1]{Eq.~(\ref{#1})}
\newcommand{\fref}[1]{Figure~\ref{#1}}

\newcommand{\sref}[1]{Section~\ref{#1}}

\newcommand{\tref}[1]{Table~\ref{#1}}

\newcommand{\thmref}[1]{Theorem~\ref{#1}}
\newcommand{\lemmaref}[1]{Lemma~\ref{#1}}
\newcommand{\asref}[1]{Assumption~\ref{#1}}
\newcommand{\propref}[1]{Proposition~\ref{#1}}
\newcommand{\corref}[1]{Corollary~\ref{#1}}
\newcommand{\gene}[1]{\emph{#1}}

\DeclareMathOperator{\LerchPhi}{LerchPhi}

\DeclareMathOperator{\child}{ch}
\DeclareMathOperator{\cdf}{\text{cdf}}

\DeclareMathOperator{\pmf}{\text{pmf}}

\DeclareMathOperator{\Anc}{Anc}
\DeclareMathOperator{\Desc}{Desc}
\DeclareMathOperator{\rmd}{\mathrm{d}}

\newcommand{\E}{\mathbb{E}}

\newcommand{\R}{\mathbb{R}}

\newcommand{\ba}{\mathbf a}
\newcommand{\bb}{\mathbf b}
\newcommand{\bq}{\mathbf q}
\newcommand{\bn}{\mathbf n}

\newcommand{\bz}{\mathbf Z}
\newcommand{\bzs}{\mathbf z}
\newcommand{\bs}{\mathbf s}
\newcommand{\be}{\mathbf e}
\newcommand{\bu}{\mathbf u}
\newcommand{\bv}{\mathbf v}
\newcommand{\bw}{\mathbf w}

\newcommand{\bg}{\boldsymbol\gamma}

\usepackage{amsthm}

\numberwithin{equation}{section}

\theoremstyle{plain} 
\newtheorem{theorem}{Theorem}[section]
\newtheorem{lemma}[theorem]{Lemma}
\newtheorem{proposition}[theorem]{Proposition}
\newtheorem{corollary}[theorem]{Corollary}

\theoremstyle{definition} 

\newtheorem{assumption}[theorem]{Assumption}

\theoremstyle{remark} 

\titleformat{\paragraph}[runin]
  {\normalsize\bfseries}{\theparagraph}{1em}{}[.]
\titlespacing*{\paragraph}{0pt}{1.5em}{1em}

\begin{document}

\title{{Numerical approximations of population size distributions for multi-type branching processes}}

\author{Xiang Ge Luo$^{1,2}$, Jack Kuipers$^{1,2}$, Niko Beerenwinkel$^{1,2,\ast}$ \\
\small $^{1}$Department of Biosystems Science and Engineering, ETH Zurich, Klingelbergstrasse 48, 4056 Basel, Switzerland \\
\small $^{2}$SIB Swiss Institute of Bioinformatics, Switzerland \\
\small ORCID: Xiang Ge Luo, \href{https://orcid.org/0000-0003-2298-4066}{0000-0003-2298-4066}; Jack Kuipers, \href{https://orcid.org/0000-0001-5357-2705}{0000-0001-5357-2705}; Niko Beerenwinkel, \href{https://orcid.org/0000-0002-0573-6119}{0000-0002-0573-6119} \\
\small $^{\ast}$Correspondence: niko.beerenwinkel@bsse.ethz.ch
}

\date{}

\maketitle

\begin{abstract}
\normalsize
Continuous-time multi-type branching processes are fundamental models for expanding and migrating populations with cancer evolution being a prototypical example. 
Inferring model parameters, like mutation and growth rates, from time-series count data requires efficient computation of population size distributions. 
Existing methods are mainly based on large-time or large-number asymptotics, which rely on either restricted initial conditions or simplified interactions between cell types. 
Here, we introduce two numerical approximations of population size distributions for multi-type branching processes on directed graphs with arbitrary initialization. 
The first approach combines a saddle-point approximation with numerical integration of the probability generating function.
We characterize admissibility and establish conditions for saddle-point existence and uniqueness. 
For directed acyclic graphs, the second approach provides a large-time small-mutation-rate alternative based on closed-form approximate Laplace transforms and efficient numerical inversion. 
We benchmark the accuracy and speed of both solutions in simulations, showing substantial improvement over the state-of-the-art large-number approximation and orders of magnitude speedup over Gillespie's stochastic simulation algorithm at matching accuracy. 
We apply our methods to analyze the relapse dynamics of an acute myeloid leukemia patient, where rapid parameter scans over a six-type patient-specific mutation tree quantify how unobserved remission burden and treatment-altered fitness can explain relapse. 
Our methods provide computational building blocks for future likelihood-based inference in cancer evolution and other expanding populations.
\end{abstract}



\section{Introduction}
\label{sec:intro}

Continuous-time multi-type branching processes provide a first-principles framework for modeling cancer evolution \citep{athreya1972branching,durrett2015branching}. 
Cells are assumed to be independent agents that randomly replicate, die, and change types by accumulating mutations. 
These mutations can be neutral, advantageous, or deleterious, so that different combinations of mutations confer different growth rates. 
The population is then described by cell counts across types. 
In cancer studies, a central task is to infer evolutionary parameters from time-series cell-count data to understand tumorigenesis and predict treatment response \citep{foo2013dynamics, wu2025statistical}. 
To do so requires efficient computation of the likelihood, which is determined by the population size distributions at the sampling times. 
However, exact solutions are available only in simple cases. 
One-type processes can be described with simple distributions \citep{athreya1972branching,durrett2015branching}. 
For two-type processes, the probability generating function is available in closed form but involves combinations of hypergeometric functions, which requires numerical evaluation \citep{antal2011exact,nguyen2023stochastic}.
For processes with more than two interacting types, tractable closed-form population-size distributions are generally unavailable.

As a result, computing the probabilities for inference often relies on asymptotics. 
For types arranged in a chain with the first one growing supercritically, a foundational class of approximations has been developed, assuming the limit of large times, and sometimes also small mutation rates \citep{durrett2010evolutionary,durrett2010evolution,cheek2018mutation,nicholson2023sequential}. 
The size distribution of each type admits an asymptotic decomposition into a time-dependent deterministic factor and a time-independent random variable, which substantially simplifies probability calculations. 
These asymptotics have been applied to study intratumor heterogeneity \citep{durrett2011intratumor}, cancer initiation \citep{zhang2023waiting}, drug resistance \citep{nicholson2019competing}, and metastasis \citep{avanzini2019cancer}. 
This framework has been extended to ($i$) linear chains that begin from one or more neutral types, followed by a sequence of mutant types with zero initial cells and non-decreasing growth rates \citep{zhang2024accumulation}; and ($ii$) tree-structured models, where the root represents a static wild-type compartment, and the mutant subclones have zero initial cells and arbitrary growth rates \citep{Luo2025Bayesian}. 
The latter enables inference of fitness landscapes from single-cell tumor mutation trees at diagnosis times across patients.

The large-time approaches above have not been generalized to settings with arbitrary mutation directions or initial cell configurations. 
In such cases, one typically appeals to approximations based on the law of large numbers and the central limit theorem (CLT), assuming sufficiently large initial populations. 
These Gaussian surrogates use only the mean and covariance matrix of the system, and are widely used to model drug resistance in cell-line experiments \citep{foo2013dynamics,foo2013cancer,gunnarsson2020understanding,li2023comparison,gunnarsson2023statistical,leder2024parameter,wu2025statistical}. 
However, Gaussian approximations are accurate mainly near the mean and can incur large relative errors in skewed tail regions, particularly when the initial populations are small.

Another direction is to evaluate the probability generating function, since it encodes the population size distribution.
Numerical integration based on the method of characteristics has been used to calculate the probability that no target cells are present \citep{quinn1989hazard,paterson2026fast}. These approaches have not covered the full distribution.
Methods based on Fourier inversion of the generating function can compute probabilities of different population sizes \citep{xu2015compressed,xu2015likelihood,stutz2022computational,awasthi2023fast}.
However, they typically require processing arrays whose size grows with the largest count of interest, making them impractical when population sizes are very large.
Saddle-point approximations can instead evaluate probabilities at selected population sizes \cite[\textit{e.g.},][]{kolassa2006series, Butler2007Saddlepoint}.
Yet, existing approaches have focused on one- or two-type processes \citep{daniels1982birth,renshaw2000applying,lang2020predicting,davison2021birthdeath}, moment calculations \citep{hyrien2010saddlepoint}, or discrete-time networks \citep{degunst2021population}.

In this work, we introduce two numerical solutions to compute population-size distributions for continuous-time multi-type branching processes on directed graphs with arbitrary initialization. 
The first combines a saddle-point approximation with forward integration of the generating function along characteristic curves, providing an efficient alternative to stochastic simulation. 
The second method resides in the large-time and small-mutation-rate regime and generalizes existing solutions to be applicable to directed acyclic graphs. 
We derive closed-form approximate Laplace transforms, which can be combined with any inverse Laplace transform tool to obtain efficient approximations. 
In simulations, we benchmark the accuracy and speed of both solutions. We then apply the saddle-point method to a longitudinal AML relapse example. 
Together, these methods enable likelihood-based inference in the broad settings of expanding and migrating populations across cancer, development, and evolution.


\section{Background}

\subsection{Model definition}
\label{sec:model}

We consider a continuous-time, multi-type branching process
\begin{equation}
\bz(t) = \left(Z_1(t), \dots, Z_d(t)\right), \quad t \geq 0,
\label{eq:process}
\end{equation}
on the state space $\mathbb{N}^d$ where $d \geq 1$ is the number of types \citep{athreya1972branching, durrett2015branching}. 
Each random variable $Z_i(t)$ denotes the number of cells (or particles) of type $i \in V := \{1,\dots,d\}$ at time $t$.
The dependency structure among types is encoded by a directed graph $\mathcal G=(V,\,E)$, where $V$ denotes types and $E$ the edges corresponding to mutational transitions. 
Each cell evolves independently according to the following mechanisms:

\begin{enumerate}
  \item \textbf{Mutation:}  
  For every edge $(i \to j) \in E$, a cell of type $i$ may produce a new cell of type $j$ with rate $\mu_{ij} > 0$.
  \[
  Z_i  \xrightarrow{ \mu_{ij} }  Z_i + Z_j
  \]
    
  \item \textbf{Replication:}  
  A cell of type $i$ may duplicate with rate $\alpha_i > 0$.  
  \[
  Z_i  \xrightarrow{ \alpha_i }  Z_i + Z_i
  \]
  
  \item \textbf{Death:}  
  A cell of type $i$ may die and be removed from the system with rate $\beta_i > 0$.
  \[
  Z_i  \xrightarrow{ \beta_i }  \emptyset
  \]
  
  \item \textbf{Immigration:}  
  New cells of type $i$ may enter the system from an external source with rate $\nu_i \geq 0$.
  \[
  \emptyset  \xrightarrow{ \nu_i }  Z_i
  \]

\end{enumerate}
\fref{fig:MTBP_overview} provides a schematic representation of a multi-type branching process with the four possible transition mechanisms.

\begin{figure}[!t]
    \centering
    \includegraphics[width=0.6\textwidth]{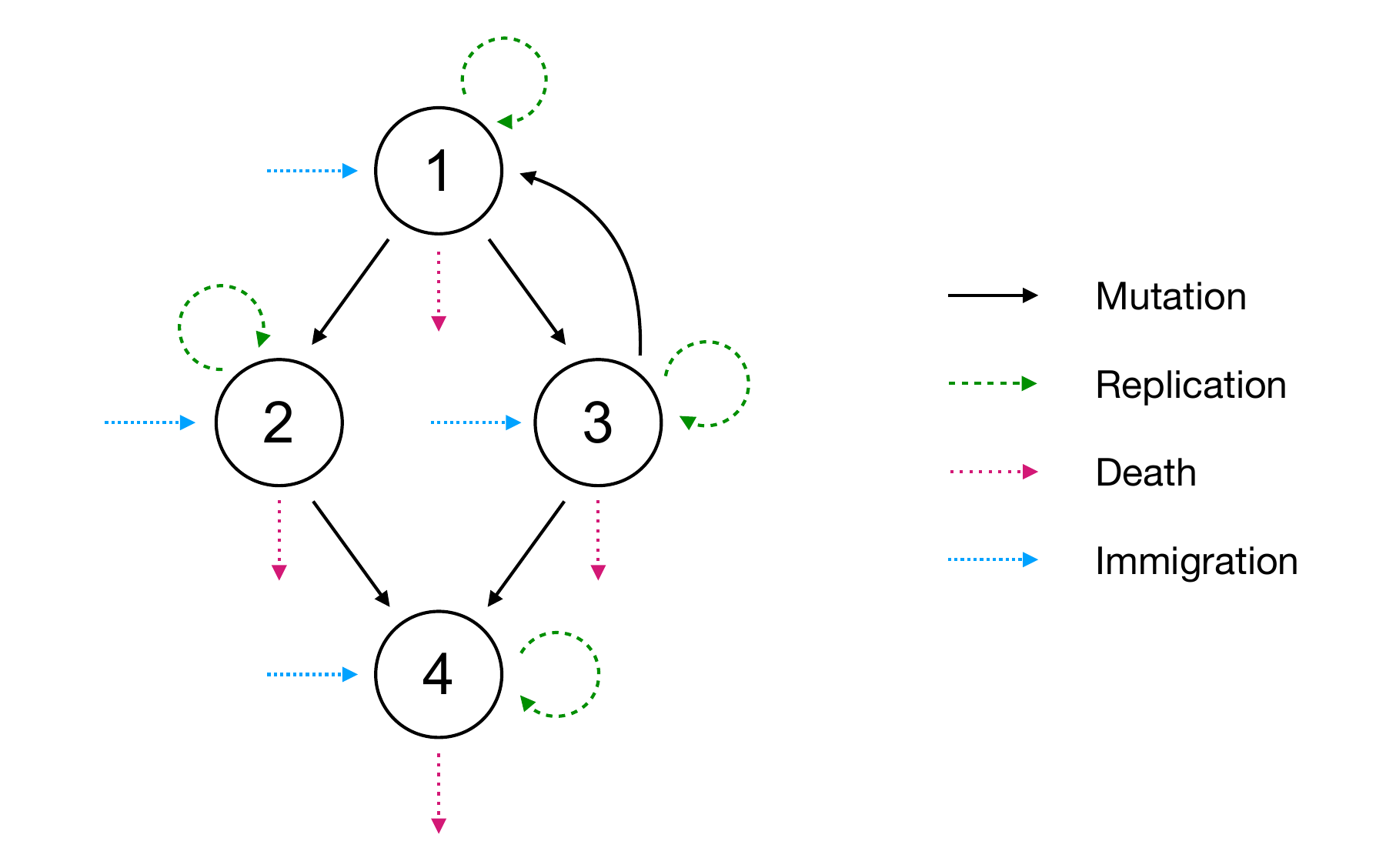}
    \caption{Schematic representation of a multi-type branching process. The figure illustrates the four possible transition mechanisms: Mutation (black solid arrow), Replication (green dashed arrow), Death (pink dotted arrow), and Immigration (blue hatched arrow)}
    \label{fig:MTBP_overview}
\end{figure}

We assume the rate parameters 
\begin{equation}
\label{eq:params}
    \theta := (\theta_i)_{i \in V}, \qquad \theta_i = \{\nu_i,\, \alpha_i,\, \beta_i,\, \bm{\mu}_{i}\}, \qquad \bm{\mu}_{i} = \{\mu_{ij}\}_{(i,j) \in E}
\end{equation}
are constant over time. We are interested in computing the population size distributions of $\bz(t)$ at time $t$, including the joint probability mass function ($\pmf$), 
\[
    p(t, \bn) := P ( \bz(t)=\bn ), \qquad \bn \in \mathbb{N}^d,
\]
the marginal probability mass function,
\[
    p(t, n_j) := P(Z_j(t) = n_j), \qquad n_j \geq 0,
\]
and the marginal cumulative distribution function ($\cdf$), 
\[
    F(t, n_j) := P(Z_j(t) \le n_j), \qquad n_j \geq 0.
\]
Let $\be_i$ denote the $i$-th unit vector. Then, the Kolmogorov forward (or chemical master) equation of the process defined by the four mechanisms above is
\[
\begin{split}
\frac{\rmd}{\rmd t} p(t, \bn)
= \sum_{i\in V} \Bigg[ & \Bigg(
    \nu_i + \alpha_i (n_i-1) + \sum_{j:(j,i)\in E}\mu_{ji}  n_j\Bigg) p(t, \bn - \be_i)\\[1ex]
  &+ \beta_i (n_i+1) p(t, \bn + \be_i)
 -  \left(
    \nu_i +
    \alpha_i n_i
  + \beta_i n_i
  + \sum_{j:(i,j)\in E}\mu_{ij} n_i
\right) p(t, \bn) \Bigg],
\end{split}
\]
with an arbitrary initial condition
$
\bz(0)=\bzs_0 \in\mathbb{N}^d.
$
Solving this equation analytically is generally prohibitive.

\subsection{Characteristic representation of the solution}

The probability generating function of $\bz(t)$ encoding the population size distributions is given by
\[
\Psi(t,\bq)=\E\left[\prod_{i \in V} q_i^{Z_i(t)}\right]
=\sum_{\bn \in \mathbb{N}^d} \bq^{\bn} p(t,\bn), \qquad
\bq^{\bn} = \prod_{i \in V} q_i^{n_i}.
\]
Following \citet{quinn1989hazard} and \citet{paterson2026fast}, for $\bq\in[0,1]^d$, the generating function satisfies the partial differential equation
\[
\dfrac{\partial\Psi}{\partial t} = \sum_{i\in V}X_i(\bq) \dfrac{\partial\Psi}{\partial q_i} +Y(\bq) \Psi,
\]
with
\begin{align}
X_i(\bq) &= \alpha_i q_i(q_i-1) + \beta_i(1 - q_i)+\sum_{j:(i,j) \in E}\mu_{ij} q_i (q_j-1), \label{eq:X}\\
Y(\bq)&=\sum_{i\in V}\nu_i(q_i-1). \label{eq:Y}
\end{align}
The method of characteristics then gives the exact solution
\begin{equation}
\label{eq:exact_solution}
\Psi(t,\bq)=\bg(t;\bq)^{\bzs_0}\exp\left(w(t;\bq)\right),
\end{equation}
where 
\begin{equation}
\label{eq:w}
    \bg(t;\bq)^{\bzs_0} = \prod_{i \in V} \gamma_i(t;\bq)^{z_{0,i}}, \qquad w(t;\bq):=\int_0^t Y\bigl(\bg(\sigma;\bq)\bigr)\rmd\sigma,
\end{equation}
and the characteristics
$\bg(\sigma;\bq)=\bigl(\gamma_i(\sigma;\bq)\bigr)_{i\in V}$
solve the autonomous ODE system
\begin{equation}
\label{eq:char}
    \forall i \in V, \qquad \dfrac{\rmd\gamma_i(\sigma;\bq)}{\rmd \sigma} = X_i\bigl(\bg(\sigma;\bq)\bigr), \qquad \gamma_i(0;\bq)=q_i.
\end{equation}
In other words, $\Psi$ can be evaluated by integrating this ODE system, but extracting the population probabilities $p(t,\bn)$ requires inverting $\Psi$. 
For the probabilities of zero cells, this inversion is unnecessary because they are obtained directly by evaluating $\Psi$ at $q_i\in\{0,1\}$ \citep{quinn1989hazard, paterson2026fast}. 
In the next section, we therefore introduce a saddle-point approximation that uses this characteristic representation of $\Psi$ to approximate population probabilities on general directed graphs.
We show that \eref{eq:exact_solution} extends to $\bq\in[0,\infty)^d$ whenever the corresponding characteristics $\bg$ remain finite, and establish conditions for the existence and uniqueness of a saddle point.
In case of non-existence, \sref{sec:large_time} provides an alternative approximation for directed acyclic graphs in the large-time, small-mutation-rate regime.

\section{Saddle-point approximation}

We first apply the standard multivariate saddle-point approximation to $p(t, \bn)$~\cite[see \textit{e.g.},][]{kolassa2006series, Butler2007Saddlepoint}.
Consider the moment generating function of $\bz(t)$ and its domain of finiteness,
\[
M(t,\bu):=\E\left[e^{\bu^\top\bz(t)}\right],
\qquad
\mathcal M_t:=
\left\{\bu\in\mathbb R^d:M(t,\bu)<\infty\right\}.
\]
For any $\mathbf c\in\operatorname{int}\mathcal M_t$, set
$r_j=e^{c_j}$. 
By Cauchy's integral formula,
\[
p(t,\bn)=\frac{1}{(2\pi i)^d}\oint_{|q_1|=r_1}\cdots\oint_{|q_d|=r_d}
\frac{\Psi(t,\bq)}{\prod_{j=1}^d q_j^{n_j+1}}
\rmd q_1\cdots\rmd q_d.
\]
Parameterizing
\[
q_j=e^{c_j+i\omega_j},
\qquad -\pi\le\omega_j\le\pi,
\]
gives the Fourier inversion form
\begin{equation}
\label{eq:Fourier}
p(t,\bn)=\frac{1}{(2\pi)^d}
\int_{\left[-\pi,\pi\right]^d}
\Psi\left(t,e^{\mathbf c+i\boldsymbol\omega}\right)
e^{-\left(\mathbf c+i\boldsymbol\omega\right)^\top\bn}
\rmd\boldsymbol\omega.
\end{equation}
On $\operatorname{int}\mathcal M_t$, we define the cumulant generating function by
\begin{equation}
\label{eq:K_v1}
K(t,\bu):=\log M(t,\bu) = \log \Psi(t,e^{\bu}),
\end{equation}
and the first and second derivatives at $\bu$ by
\[
\mathbf m(\bu):=\nabla_{\bu}K(t,\bu),\qquad \Sigma(\bu):=\nabla_{\bu}^2 K(t,\bu).
\]
The saddle point $\hat{\bu}\in\R^d$ solves the equation
\begin{equation}
\label{eq:saddle}
\mathbf m(\hat{\bu})=\bn.
\end{equation}
Suppose that a nondegenerate $\hat{\bu}\in\operatorname{int}\mathcal M_t$ exists. Taking $\mathbf c=\hat{\bu}$ in \eref{eq:Fourier}, the exponent of the integrand near $\boldsymbol\omega=\mathbf0$ admits the expansion
\[
K\left(t,\hat{\bu}+i\boldsymbol\omega\right)
-\left(\hat{\bu}+i\boldsymbol\omega\right)^\top\bn
=
K\left(t,\hat{\bu}\right)-\hat{\bu}^\top\bn 
+i\boldsymbol\omega^\top
\left(\mathbf m\left(\hat{\bu}\right)-\bn\right)
-\frac12
\boldsymbol\omega^\top
\Sigma\left(\hat{\bu}\right)
\boldsymbol\omega
+\mathcal O\left(\left\lVert\boldsymbol\omega\right\rVert^3\right).
\]
By \eref{eq:saddle}, the linear term vanishes. Retaining the quadratic term and extending the integral to $\mathbb R^d$ gives the second-order saddle-point approximation formula
\begin{equation}
\label{eq:SP}
p(t,\bn)\approx\dfrac{\exp\left(K\bigl(t,\hat{\bu}\bigr)-\hat{\bu}^{\top} \bn\right)}{(2\pi)^{d/2}\sqrt{\det \Sigma(\hat{\bu})}}.
\end{equation}

Under standard saddle-point asymptotics, the approximation improves as the effective population size increases, which depends on both $\bzs_0$ and $\bn$ \cite[see \textit{e.g.},][]{daniels1954saddlepoint, kolassa2006series}. Higher-order corrections involve the third and fourth derivatives of $K$ at $\hat{\bu}$. Unlike Gaussian approximations derived from the central limit theorem \cite[see \textit{e.g.},][]{gunnarsson2023statistical}, which expand $K(t,\bu)$ at $\bu=0$ and use only the mean $\mathbf m(0)$ and covariance $\Sigma(0)$ under the original measure, the saddle-point method first applies an exponential tilt to match the target count by solving $\mathbf m(\hat{\bu})=\bn$ and then makes the quadratic approximation at $\hat{\bu}$. This uses the full shape of $K$ through the map $\bu\mapsto\mathbf m(\bu)$ and adapts to $\bn$ rather than to the center, reducing large relative errors in skewed tail regions.
\aref{app:saddle_illustration} illustrates this difference on a one-type process with immigration and no initial cells.

\subsection{Admissibility, existence, and uniqueness}
\label{sec:admissibility}

Computing the saddle-point approximation \eref{eq:SP} involves evaluating the generating functions and their derivatives.
Next, we study the model-specific conditions under which the characteristic representation of $\Psi$ (\eref{eq:exact_solution}) and the saddle-point approximation are well defined.
For compact notation, we write
\[
\bg_{\bu}(\sigma):=\bg(\sigma;e^{\bu}),\qquad
\gamma_{i,\bu}(\sigma):=\gamma_i(\sigma;e^{\bu}),\qquad
w_{\bu}(\sigma):=w(\sigma;e^{\bu}).
\]
For fixed $t>0$, we define the admissible set as
\[
\mathcal U_t:=\left\{\bu\in\mathbb R^d:
\bg_{\bu}\text{ exists on }[0,t]\right\},
\]
and call $\bu$ admissible if $\bu\in\mathcal U_t$. 
For $\bu\in\mathcal U_t$, \lemmaref{lemma:positive_characteristics} and \thmref{thm:characteristic_mgf} show that \eref{eq:exact_solution} remains valid for $\bq = e^{\bu}$ and that \eref{eq:K_v1} can be rewritten using the characteristics
\begin{equation}
\label{eq:K_v2}
K(t,\bu)=\bzs_0^{\top}\log\bg_{\bu}(t)+w_{\bu}(t).
\end{equation}
Moreover, \lemmaref{lemma:admissible_mgf_interior} gives
\[
\mathcal U_t\subseteq\operatorname{int}\mathcal M_t,
\]
so $K(t,\cdot)$ is smooth and its derivatives $\mathbf m(\bu)$ and $\Sigma(\bu)$ exist on $\mathcal U_t$.

\begin{assumption}
\label{assum:nondegen}
Let the active type set be
\[
\mathcal A=\{j\in V:\nu_j>0 \text{ or } z_{0,j}>0\},
\]
which contains the types that seed the process through immigration or initial cells.
We assume that $\mathcal A$ is non-empty, and that for every type $k\in V$, either $k\in\mathcal A$, or there exists a directed path in $\mathcal G$ from some type $j\in\mathcal A$ to $k$.
\end{assumption}

Note that if \asref{assum:nondegen} fails, then any type not reachable from $\mathcal A$ is identically zero for all $t\ge0$. Hence, we may restrict attention to the induced subgraph
\[
\mathcal G'=\mathcal G[V'], \qquad
V':=\{k\in V:\exists j\in\mathcal A \text{ such that there is a directed path from } j \text{ to } k\},
\]
on which \asref{assum:nondegen} holds automatically, and all subsequent results can be stated and proved with $\mathcal G'$ in place of $\mathcal G$.

\begin{assumption}
\label{assum:all-active}
Every type is active, \textit{i.e.,}~$\mathcal A=V$.
\end{assumption}

By definition, an admissible saddle for target $\bn$ exists if and only if
$\bn\in\mathbf m(\mathcal U_t)$. Under \asref{assum:nondegen}, \propref{prop:saddle_uniqueness} shows that $\Sigma(\bu)$ is positive definite for all $\bu\in\mathcal U_t$ and the map $\mathbf m:\mathcal U_t\to\R^d$ is injective. Hence, every admissible saddle is nondegenerate and unique whenever it exists.
Under \asref{assum:all-active}, \thmref{thm:saddle_point_existence} gives
$\mathbf m(\mathcal U_t)=(0,\infty)^d$, so an admissible saddle exists for
every positive target $\bn \in (0,\infty)^d$. A three-type counterexample in \sref{app:saddle_counterexample} shows that \asref{assum:nondegen} alone does not guarantee existence for all positive targets. If $n_j=0$ for some $j\in V$, then set $q_j=0$ in $\Psi(t,\bq)$ and apply saddle-point approximation in the remaining positive-target dimensions.

\subsection{Gradients, Hessians, and admissibility conditions}

The search for the saddle point requires repeated evaluations of $K(t, \bu)$, its gradient $\mathbf m(\bu)$, and Hessian $\Sigma(\bu)$, while ensuring the admissibility of $\bu$.
Our strategy is to integrate the characteristic system \eref{eq:char} together with its first- and second-order variational equations, and to derive auxiliary equations that provide necessary conditions for admissibility.

\paragraph{Gradients}
Let
\begin{equation}
\label{eq:S_xi}
    S(\sigma) := \dfrac{\partial \bg_{\bu}(\sigma)}{\partial \bu} \in\R^{d\times d}
    \qquad \text{and} \qquad
    \bm\xi(\sigma):= \dfrac{\partial w_{\bu}(\sigma)}{\partial \bu} \in\R^d.
\end{equation}
Differentiating \eref{eq:w} and \eref{eq:char} with respect to $\bu$ yields the variational equations
\begin{align}
\dfrac{\rmd S}{\rmd \sigma} &= J(\bg_{\bu}) S, & S(0)&=\mathrm{diag}(\bq), \label{eq:Svar}\\
\dfrac{\rmd \bm\xi}{\rmd \sigma}&= S^{\top}\bm\nu, & \bm\xi(0)&=\mathbf 0, \label{eq:xivar}
\end{align}
where $J(\bq)$ is the Jacobian of $\mathbf X$ with entries
\begin{equation}
\label{eq:jacobian}
(J(\bq))_{jk} := \frac{\partial X_j}{\partial q_k} =
    \begin{dcases}
        2\alpha_j q_j-\alpha_j-\beta_j+ \sum_{\ell:(j,\ell) \in E}\mu_{j\ell}(q_{\ell}-1) & \text{if } j = k, \\
        \mu_{jk} q_j & \text{if } (j,k) \in E, \\
         & \\
        0 & \text{otherwise.}
    \end{dcases}
\end{equation}
and $\bm\nu=(\nu_j)_{j \in V}$. At time $t$, the gradient $\nabla_{\bu} K(t,\bu)$ is given by
\begin{equation}
\label{eq:mean}
\mathbf m (\bu) = S^\top \mathrm{diag}(\bg_{\bu}(t)^{-1}) \bzs_0 + \bm \xi.
\end{equation}

\paragraph{Hessians}
Let
\begin{equation}
    T(\sigma) := \dfrac{\partial^2 \bg_{\bu}(\sigma)}{\partial \bu^2} \in\R^{d\times d\times d} \qquad \text{and} \qquad \Xi(\sigma) := \dfrac{\partial^2 w_{\bu}(\sigma)}{\partial \bu^2} \in\R^{d\times d}.
\end{equation}
Differentiating \eref{eq:w} and \eref{eq:char} with respect to \(\bu\) twice yields
\begin{align}
\dfrac{\rmd T_j}{\rmd \sigma} &= \sum_{a \in V}J_{ja}(\bg_{\bu}) T_a
+ S^{\top} H_j(\bg_{\bu}) S,
& T_{j}(0)&=q_j\be_j \be_j^{\top} , \label{eq:Tvar}\\
\dfrac{\rmd \Xi}{\rmd \sigma}&=\sum_{j=1}^d \nu_j T_j, & \Xi(0)&=\mathbf{0}, \label{eq:Xivar}
\end{align}
where $T_j \in \mathbb R^{d\times d}$ is the $j$-th matrix slice of the tensor $T$ with entries $(T_j)_{k{\ell}} =T_{jk{\ell}}$, and $H_j(\bq)$ is the Hessian of $X_j$ with entries
\begin{equation}
(H_j(\bq))_{k{\ell}}:= \dfrac{\partial^2 X_j}{\partial q_k \partial q_{\ell}} =
    \begin{dcases}
        2\alpha_j & \text{if } j = k = \ell,\\
        \mu_{jk} & \text{if } j = \ell \text{ and } (j,k) \in E,\\
        \mu_{j\ell} & \text{if } j = k \text{ and } (j,\ell) \in E,\\
        0 & \text{otherwise.}
    \end{dcases}
\end{equation}
At time $t$, the Hessian $\nabla_{\bu}^2K$ is given by
\begin{equation}
\label{eq:Sigma}
\Sigma(\bu)= \sum_{j \in V} \dfrac{z_{0,j}}{\gamma_{j,\bu}(t)} T_j - S^{\top} L S +\Xi, \qquad L = \mathrm{diag}\left( \dfrac{z_{0,1}}{\gamma_{1,\bu}(t)^2}, \dots, \dfrac{z_{0,d}}{\gamma_{d,\bu}(t)^2}\right).
\end{equation}

\paragraph{Admissibility conditions}
To obtain computable diagnostics for admissibility, we rewrite each component of \eref{eq:char} as a scalar Riccati equation
\begin{equation}
\label{eq:riccati}
\frac{\rmd}{\rmd\sigma}\gamma_{i,\bu}(\sigma)
=
\alpha_i\gamma_{i,\bu}(\sigma)^2
+b_i(\sigma)\gamma_{i,\bu}(\sigma)+\beta_i,
\qquad
b_i(\sigma)=-(\alpha_i+\beta_i)+\sum_{k:(i,k)\in E}\mu_{ik}\left(\gamma_{k,\bu}(\sigma)-1\right),
\end{equation}
with initial condition $\gamma_{i,\bu}(0)=q_i=e^{u_i}$. Define $y_i$ as the solution
of
\begin{equation}
\label{eq:riccati_transformation}
y_i'(\sigma)=-\alpha_i\gamma_{i,\bu}(\sigma)y_i(\sigma) \qquad \Longleftrightarrow \qquad
\gamma_{i,\bu}(\sigma)=-\frac{1}{\alpha_i}
\frac{y_i'(\sigma)}{y_i(\sigma)}.
\end{equation}
Differentiating \eref{eq:riccati_transformation} and using \eref{eq:riccati}
shows that $y_i$ satisfies
\begin{equation}
\label{eq:linear}
y_i''(\sigma)-b_i(\sigma)y_i'(\sigma)+\alpha_i\beta_i y_i(\sigma)=0,
\qquad
y_i(0)=1,\quad y_i'(0)=-\alpha_i q_i.
\end{equation}
Let $\varphi_i,\varrho_i$ be the fundamental solutions of \eref{eq:linear} with
\[
\varphi_i(0)=1,\quad \varphi_i'(0)=0,\qquad
\varrho_i(0)=0,\quad \varrho_i'(0)=1.
\]
Then, the fundamental matrix
\[
Q_i(\sigma)=
\begin{bmatrix}
\varphi_i(\sigma) & \varrho_i(\sigma)\\
\dot\varphi_i(\sigma) & \dot\varrho_i(\sigma)
\end{bmatrix}
\in\mathbb R^{2\times 2}
\]
solves
\begin{equation}
\label{eq:admissible_ODE}
\frac{\rmd}{\rmd\sigma}Q_i(\sigma)=A_i(\sigma)Q_i(\sigma),\qquad Q_i(0)=I_2,
\end{equation}
with
\begin{equation}
\label{eq:admissible_A}
A_i(\sigma)=
\begin{bmatrix}
0 & 1\\
-\alpha_i\beta_i & b_i(\sigma)
\end{bmatrix},
\end{equation}
and
\begin{equation}
\label{eq:linear_solution}
y_i(\sigma)=\varphi_i(\sigma)-\alpha_i q_i\varrho_i(\sigma).
\end{equation}

By \propref{prop:admissibility-caps}, admissibility requires
\[
y_i(\sigma)>0,
\qquad i\in V,\quad 0\le\sigma\le t.
\]
This provides a practical pathwise diagnostic: if $\bu$ is not admissible, then at least one auxiliary coordinate $y_i(\sigma)$ tends to zero as $\sigma\uparrow T$, where $T\le t$ is the first blowup time of $\bg$. 
Moreover, the terminal quantity
\[
q_{i,\max}(t;\bu):=\frac{\varphi_i(t)}{\alpha_i\varrho_i(t)}
\]
provides a convenient \textit{necessary} coordinate-wise bound:
\begin{equation}
\label{eq:coordinate_caps}
u_i<\log q_{i,\max}(t;\bu),
\qquad i\in V.
\end{equation}
If $i\in V$ is a leaf with no children, it reduces to
\[
q_{i,\max}(t)=
\begin{dcases}
\dfrac{e^{\lambda_i t}-\beta_i/\alpha_i}{e^{\lambda_i t}-1} & \text{if } \alpha_i\neq\beta_i, \\[1ex]
1+\dfrac{1}{\alpha_i t} & \text{if } \alpha_i=\beta_i,
\end{dcases}
\qquad \text{where } \lambda_i:=\alpha_i-\beta_i.
\]
Together, $y_i$ and $q_{i,\max}$ provide admissibility diagnostics and the initial step-size heuristic used by the numerical solver in \sref{sec:SP_solver}.

\subsection{Univariate marginals}

To obtain the univariate marginal $\pmf$, $p(t, n_j) = P(Z_j(t) = n_j)$, we set $q_i=1$ for all $i\neq j$ (equivalently, $u_i=0$ for $i\neq j$), which reduces the cumulant generating function to the scalar map
\[
k_j(u):= K\left(t, u \be_j\right),
\]
and apply the same saddle-point construction as for the joint $\pmf$, with the tilt restricted to the $j$-th coordinate.
Under \asref{assum:nondegen}, the univariate saddle point is unique whenever it exists, and under \asref{assum:all-active}, it exists for every positive target (\propref{prop:univariate-inheritance}).
For the marginal $\cdf$, $F(t, n_j) = P(Z_j(t)\le n_j)$, we use the continuity-corrected Lugannani--Rice approximation of \citet{lugannani1980saddle} and \citet{Butler2007Saddlepoint}, a specialized univariate saddle-point method to compute the tail probabilities. Specifically, let the saddle point $\hat{u}_j\in\mathbb{R}$ be the solution to the equation
\[
    k'_j(\hat{u}_j) = \tilde{n}_j := n_j + \frac{1}{2},
\]
where the offset $\frac{1}{2}$ is for continuity correction. With two auxiliary quantities
\[
\hat{\zeta}_j := \operatorname{sign}(\hat{u}_j)\sqrt{2\left[\hat{u}_j \tilde{n}_j - k_j(\hat{u}_j)\right]} \qquad \text{and} \qquad
\hat{\eta}_j := 2 \sinh\left(\dfrac{\hat{u}_j}{2}\right)\sqrt{k''_j(\hat{u}_j)},
\]
the approximation is given by
\begin{equation} 
\label{eq:LR}
P(Z_j(t) \le n_j) \approx
\begin{dcases}
    \Phi(\hat{\zeta}_j) + \phi(\hat{\zeta}_j)\left(\dfrac{1}{\hat{\zeta}_j} - \dfrac{1}{\hat{\eta}_j}\right) & \text{if } \tilde n_j \neq \mu_j, \\[1.5ex]
    \frac{1}{2} + \frac{k'''_j(0)}{6\sqrt{2\pi} k''_j(0)^{3/2}} & \text{if } \tilde n_j = \mu_j.
\end{dcases}
\end{equation}
Here, $\Phi$ and $\phi$ are the standard normal cumulative distribution and density functions, respectively, and $\mu_j := k_j'(0)$ is the mean of $Z_j(t)$. 
For computation, the derivatives $k_j'$ and $k_j''$ are obtained from first- and second-order directional variational equations along $\be_j$ (\aref{app:saddle_univariate}). In the case $\tilde n_j=\mu_j$, we obtain $k'''_j(0)$ by a symmetric finite difference of $k_j''(u)$ around $u=0$.

\subsection{Numerical procedure}
\label{sec:SP_solver}

We now combine the saddle-point formula, variational equations, and admissibility diagnostics in a single numerical solver.
We augment the numerical scheme of \citet{paterson2026fast} by simultaneously integrating the characteristic equations (\eref{eq:w} and \eref{eq:char}), the variational equations (\eref{eq:S_xi}--\eref{eq:Sigma}), and the auxiliary systems (\eref{eq:admissible_ODE}--\eref{eq:linear_solution}) in one forward pass. This yields $(K,\mathbf m,\Sigma)$ for saddle-point evaluation, as well as $\bq_{\max}=(q_{i,\max})_{i\in V}$ and $y_{\min}=\min_{i\in V}\min_{0\le\sigma\le t}y_i(\sigma)$ for pathwise admissibility diagnostics. 
For a target $\bn$, solving the saddle-point equation $\mathbf m(\bu)=\bn$ is equivalent to minimizing the strictly convex objective (\propref{prop:saddle_uniqueness}):
\[
F_{\bn}(\bu):=K(t,\bu)-\bn^\top\bu,
\qquad \nabla F_{\bn}=\mathbf m-\bn,
\qquad \nabla^2F_{\bn}=\Sigma.
\]
We apply a damped Newton method:
\begin{enumerate}[label=(\roman*)]
    \item Starting from $\bu=\mathbf0$, the outer solver computes a direction $\delta$ from the Cholesky factorization of
    \[
        \Sigma \delta = \bn - \mathbf m.
    \]
    The coordinate-wise bounds $\bq_{\max}$ provide an initial step-size heuristic,
    \begin{equation}
    \label{eq:s}
        \tilde{s} = \min_{i:\delta_i>0}
      \frac{\log q_{i,\max}(t;\bu)-u_i}{\delta_i},
    \qquad
    s = \min\left\{1,\epsilon\tilde{s}\right\},\quad \epsilon\in(0,1).
    \end{equation}
    A backtracking line search then halves $s$ until the trial forward pass remains admissible and gives sufficient decrease in $F_{\bn}$. If the line search stalls near $y_{\min}=0$ with a nonzero residual, the solver reports that no interior saddle was detected. 
    Convergence is declared when $\|\mathbf m-\bn\|_2$ drops below a specified tolerance, after which \eref{eq:SP} is evaluated.

    \item The inner solver uses Heun's method to perform one forward pass and returns $(K,\mathbf m,\Sigma,\bq_{\max},y_{\min})$ together with an admissibility flag.
\end{enumerate}
We provide the pseudocode in \aref{app:algorithms}.

\paragraph{Computational cost}
Let $N_{\mathrm{step}}=\lceil t/\Delta\rceil$ and $m=|E|$. A joint forward pass is dominated by the second-order variational equations, which require $\mathcal O\left((d+m)d^2\right)$ operations per Heun step and $\mathcal O(d^3)$ working memory. A forward pass, together with the Cholesky factorization of $\Sigma$ in a Newton update, costs
\[
\mathcal O\left(N_{\mathrm{step}}(d+m)d^2+d^3\right).
\]
For a univariate marginal, directional recursions reduce each forward pass to $\mathcal O\left(N_{\mathrm{step}}(d+m)\right)$ time and $\mathcal O(d)$ working memory.

\section{Large-time small-mutation-rate approximation}
\label{sec:large_time}

Throughout this section, we assume that $\mathcal G=(V,E)$ is a directed acyclic graph (DAG). At large times and small mutation rates, the population size distributions of $\bz(t)$ can be approximated by recursive Laplace transforms. We describe here how to build the approximation out of simpler components, with additional details in \aref{app:ltsm}. This approximation does not rely on \asref{assum:nondegen} or \asref{assum:all-active}.

\subsection{Elementary processes}
\label{sec:elementary_processes}

We begin with two elementary processes whose Laplace transforms are available in closed form \citep{athreya1972branching, durrett2015branching, nicholson2023sequential} (\fref{fig:LTSM}a):
\begin{enumerate}
    \item A one-type branching process with zero initial population and positive immigration:
    \[
        \left(Z^0(t)\right)_{t \geq 0}, \qquad Z^0(0) = 0 \text{ and } \alpha, \beta, \nu > 0.
    \]
    Let $\lambda = \alpha - \beta$. The Laplace transform of $Z^0(t)$ at time $t > 0$ is given by
    \begin{equation}
    \label{eq:L_0}
        \mathcal L^0(s; t) := \E \left[ e^{-sZ^0(t)}\right] = \left[\dfrac{p(t)}{1 - (1-p(t))e^{-s}}\right]^{\dfrac{\nu}{\alpha}},
    \end{equation}
    where
    \[
        p(t) = 
        \begin{dcases}
        \dfrac{1}{1 + \alpha t} & \text{if } \lambda = 0, \\[1ex]
        \dfrac{\lambda}{\alpha \exp(\lambda t) - \beta} & \text{otherwise.} 
        \end{dcases}
    \]
    
    \item A one-type branching process initiating from a single cell and having no immigration:
    \[
        \left(Z^1(t)\right)_{t \geq 0}, \qquad Z^1(0) = 1 \text{ and } \alpha, \beta> 0, \,\nu = 0.
    \]
    The Laplace transform of $Z^1(t)$ at time $t > 0$ is given by
    \begin{equation}
    \label{eq:L_1}
        \mathcal L^1(s; t) := \E \left[ e^{-sZ^1(t)}\right] = a(t) + (1 - a(t)) \dfrac{(1 - b(t))e^{-s}}{1 - b(t) e^{-s}},
    \end{equation}
    where
    \[
        a(t) = 
        \begin{dcases}
            \dfrac{\alpha t}{1 + \alpha t} & \text{ if } \lambda = 0, \\[1ex]
            \dfrac{\beta \exp(\lambda t) - \beta}{\alpha \exp(\lambda t) - \beta}& \text{ otherwise},
        \end{dcases}
         \qquad \text{and} \qquad 
         b(t) = 
        \begin{dcases}
            \dfrac{\alpha t}{1 + \alpha t} & \text{ if } \lambda = 0, \\[1ex]
            \dfrac{\alpha \exp(\lambda t) - \alpha}{\alpha \exp(\lambda t) - \beta} & \text{ otherwise}.
        \end{dcases}
    \]
\end{enumerate}
For $i$ in $V$, we write $Z_i^0, Z_i^1$ to represent the above elementary processes instantiated with the rate parameters $\theta_i = \{\nu_i, \alpha_i, \beta_i, \bm{\mu}_{i}\}$. When $\nu_i=0$, we take $Z_i^0(t)=0$ and $\mathcal L_i^0(s;t)=1$.

\subsection{Large-time behavior along a path}
\label{sec:LT_linear}

Next, we study the large-time behavior along a path, which will become the building block for the full DAG approximation.
A simple directed path $\pi=(v_0\to v_1\to\cdots\to v_L)$ is a finite sequence of $(L+1)$ distinct connected vertices in $\mathcal{G}$. We write $\pi_{a:b}=(v_a\to\cdots\to v_b)$, with $\pi_{:b}=\pi_{0:b}$, $\pi_{a:}=\pi_{a:L}$, and $\pi_{a:a}=(v_a)$. If $(v_L\to u)\in E$, then write $\pi\to u=(v_0\to\cdots\to v_L\to u)$.
A linear system is defined along a path $\pi$: 
the root is a linear combination of independent $Z^0_{v_0}$ and $Z^1_{v_0}$ processes (\fref{fig:LTSM}b,c). For $1\le\ell\le L$, type $v_\ell$ has zero initial population and is
populated by mutations from its parent $v_{\ell-1}$. Let
$\{Z^{1,m}_{v_\ell}\}_{m\in\mathbb N}$ be independent copies of the
single-cell process $Z^1_{v_\ell}$. Then,
\[
Z_{v_\ell}(t)=\sum_{m:T_m\le t}Z^{1,m}_{v_\ell}\left(t-T_m\right),
\]
where, conditional on the parent trajectory $\{Z_{v_{\ell-1}}(u)\}_{u\le t}$, the arrival times $T_m$ of new type-$v_\ell$ cells follow a nonhomogeneous Poisson process with intensity $f(u)=\mu_{v_{\ell-1}v_\ell}Z_{v_{\ell-1}}(u)$. Each arrival founds an
independent type-$v_\ell$ clone, and the sum gives the total size of these clones at time $t$. This is an equivalent representation of the mutation mechanism defined in \sref{sec:model}  \citep{athreya1972branching,durrett2015branching}.

Following \citet{nicholson2023sequential} and \citet{Luo2025Bayesian}, we define the running-max net growth rate at $v_{\ell}$ for $\ell = 0, \dots, L$ by
\begin{equation}
\label{eq:delta}
    \delta_{v_\ell}^{\pi} = \max \left\{0,  \max_{0 \le k \le \ell} \lambda_{v_k}\right\},
\end{equation}
and the number of times $\delta_{v_\ell}^{\pi}$ has been obtained up to $v_\ell$ by
\begin{equation}
\label{eq:r}
    r^{\pi}_{v_\ell} = \#\{0 \le k \le \ell:  \lambda_{v_k} = \delta^{\pi}_{v_\ell} \} + \mathds{1} \{\delta^{\pi}_{v_\ell} = 0\},
\end{equation}
where the indicator function adds one to the counter when the fittest type has non-supercritical growth ($\delta^{\pi}_{v_\ell} \le 0$).
We define the time-dependent scaling function at $v_\ell$ as
\begin{equation}
\label{eq:large_time_weight}
    w_{v_\ell}^{\pi}(t)=t^{-\left(r_{v_\ell}^{\pi}-1\right)}e^{-\delta_{v_\ell}^{\pi}t}.
\end{equation}
The superscript $\pi$ indicates that these quantities are path-dependent.

In \aref{app:convergence}, we show that for all $0 \le \ell \le L$, there exists a nonnegative random variable $\tilde Z_{v_\ell}$ such that
\[
    w_{v_\ell}^{\pi}(t)Z_{v_\ell}(t)\overset{d}{\longrightarrow}\tilde Z_{v_\ell} \qquad \text{as } t\to\infty.
\]
Therefore, we use the approximation
\begin{equation}
\label{eq:large_time_decomp}
    w_{v_\ell}^{\pi}(t)Z_{v_\ell}(t) \approx \tilde Z_{v_\ell}.
\end{equation}
The scaling function $w_{v_\ell}^{\pi}(t)$ captures the deterministic asymptotic growth of type $v_\ell$, with an exponential factor determined by the fittest type along the path and a polynomial factor determined by the number of times this maximal growth rate is attained, while $\tilde Z_{v_\ell}$ captures the remaining stochastic variation.
This extends the results of \citet{nicholson2023sequential} to roots with immigration and to critical and subcritical growth.

\begin{figure}[!t]
    \centering
    \includegraphics[width=\textwidth]{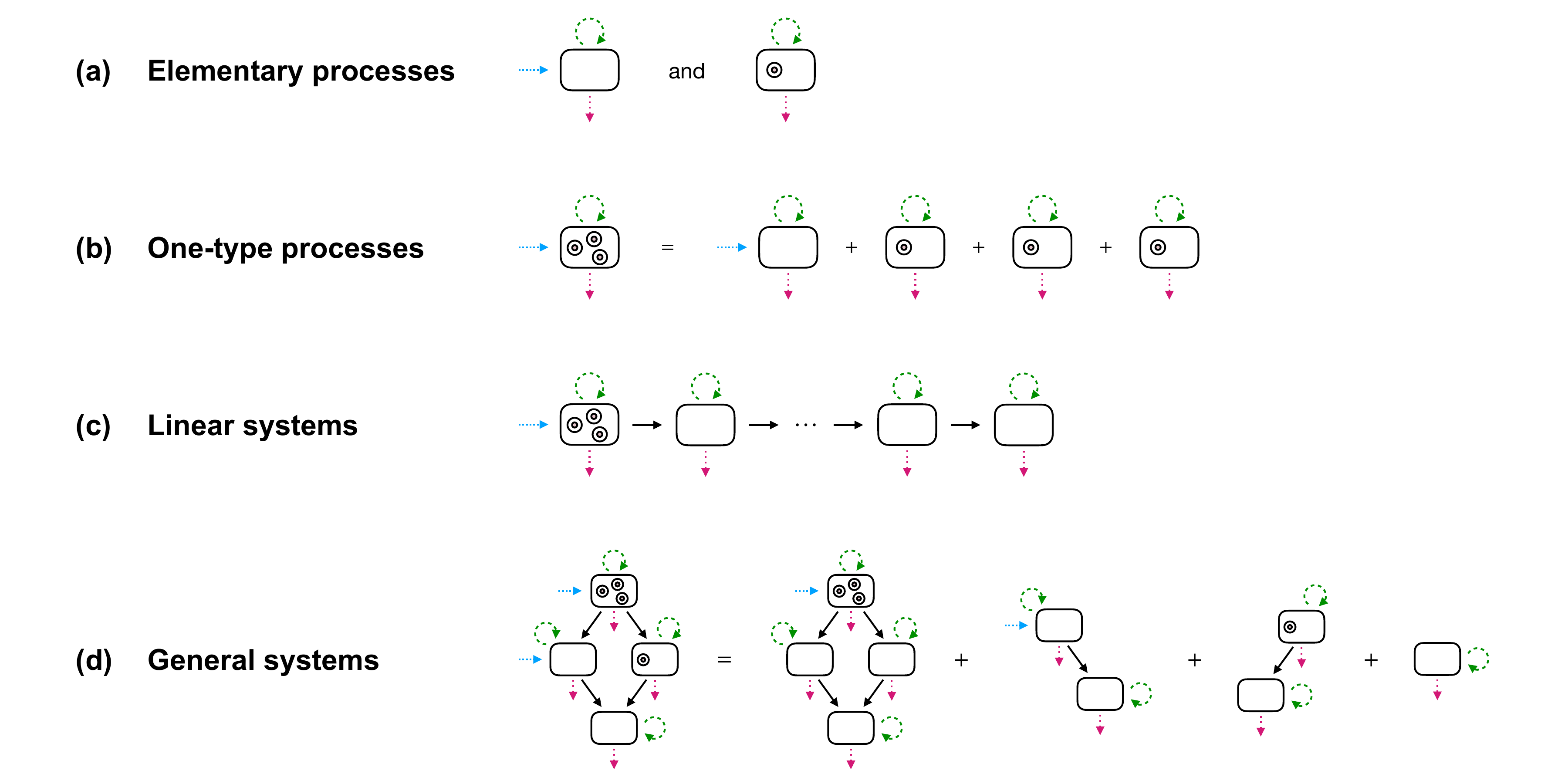}
    \caption{
        Four levels of multi-type branching processes with increasing complexity. (\textbf{a}) Elementary processes: either a one-type process with zero initial population and positive immigration ($Z^0$), or a one-type process initiating from a single cell and having no immigration ($Z^1$).
        (\textbf{b}) One-type processes: a one-type process with arbitrary initial population and positive immigration as a linear combination of independent $Z^0$ and $Z^1$ processes.
        (\textbf{c}) Linear systems: a multi-type branching process along a simple directed path, where the root is a linear combination of independent $Z^0$ and $Z^1$ processes, and each child has zero initial population and receives parental influxes.
        (\textbf{d}) General systems: a multi-type branching process on a directed acyclic graph, where each type has arbitrary initial population and immigration. The system can be decomposed into independent subsystems}
    \label{fig:LTSM}
\end{figure}

\subsection{Approximations via recursive Laplace transforms}
\label{sec:laplace}

We call the multi-type branching process $\bz(t)$ on a directed acyclic graph $\mathcal G$ a general system. 
For $i\in V$, let $\Anc(i)$ and $\Desc(i)$ be the strict ancestors and descendants, and set $\Anc^\ast(i)=\Anc(i)\cup\left\{i\right\}$ and $\Desc^\ast(i)=\Desc(i)\cup\left\{i\right\}$. For $v,u\in V$, let $\Pi_{v\to u}$ be the set of all simple directed paths from $v$ to $u$.

Our strategy is to decompose the process first by active source and then by directed path, propagate the elementary transforms recursively along each path, and finally recombine and invert the resulting transform.

\paragraph{Subsystem decomposition}
Let $\bz^{(i)}(t)$ denote the subsystem where $i$ is the only active type ($\nu_i>0$ or $z_{0,i}>0$). Then, $Z_k^{(i)}(t)=0$ for $k\notin\Desc^\ast(i)$, and $Z^{(i)}_i(t)$ is a linear combination of independent $Z_i^{(i),0}(t)$ and $\left\{Z^{(i),1,m}_i(t)\right\}_{m = 1}^{z_{0,i}}$ processes (\fref{fig:LTSM}b):
\begin{equation}
\label{eq:elementary_decomposition}
Z_i^{(i)}(t)=Z_i^{(i),0}(t)+\sum_{m=1}^{z_{0,i}}Z_i^{(i),1,m}(t),
\end{equation}
so its Laplace transform factorizes exactly using the elementary Laplace transforms (\eref{eq:L_0}, \eref{eq:L_1}),
\[
\E\left[e^{-sZ_i^{(i)}(t)}\right]=\mathcal L_i^0(s; t)\left[\mathcal L_i^1(s; t)\right]^{z_{0,i}}.
\]
The general system $\bz(t)$ can be decomposed into independent subsystems (\fref{fig:LTSM}d):
\begin{equation}
\label{eq:subsystem_decomposition}
\bz(t)=\sum_{i\in V}\bz^{(i)}(t).
\end{equation}

\paragraph{Pathwise contributions}

For each $i \in V$, we write each descendant population $Z_j^{(i)}(t)$ for $j \in \Desc(i)$ as a sum of contributions from every directed path from $i$ to $j$. Let $\mathcal{G}_i := \mathcal G [\Desc^\ast(i)]$ denote the subgraph induced by $\Desc^\ast(i)$. Since $\mathcal G$ is a DAG, $\mathcal G_i$ is also a DAG, and a topological order of $\Desc^\ast(i)$ exists. For all $j \in \Desc(i)$ and all $\pi = (i = v_0 \to v_1 \to \cdots \to v_{\ell-1} \to v_{\ell} = j) \in \Pi_{i \to j}$, we define, in topological order of $\Desc^\ast(i)$, the mutation flow from $v_{\ell - 1}$ to $j$ via the path $\pi$ as
\[
Z^{(i),\pi}_j(t) = \sum_{m: T_m \le t} Z^{(i),1,m}_j\left(t - T_m\right),
\]
where the $Z_j^{(i),1,m}$ are independent copies of $Z_j^1$, chosen separately for each path. Conditional on the parent history, $\{T_m\}_{m \in \mathbb N}$ come from a Poisson process with intensity
\[
f^{(i),\pi}_j(u) = \mu_{v_{\ell-1} j}\, Z^{(i),\pi_{:\ell-1}}_{v_{\ell-1}}(u).
\]
Then, for each variable $Z^{(i),\pi}_j(t)$, the path-level asymptotic decomposition \eref{eq:large_time_decomp} applies.
We denote $Z^{(i), \pi}_i(t) := Z^{(i)}_i(t)$ for $\pi = (i)$. It follows that $Z^{(i)}_j(t)$ is the sum of all pathwise contributions from $i$ to $j$:
\begin{equation}
\label{eq:path_decomposition}
    Z^{(i)}_j(t) = \sum_{\pi \in \Pi_{i \to j}} Z^{(i),\pi}_j(t) \qquad \text{for all }j \in \Desc^\ast(i).
\end{equation}

\paragraph{Conditional approximation}
For an edge $v_{\ell-1}\to v_\ell$ along a path $\pi$, the large-time, small-mutation-rate approximation gives the marginal recursion
\begin{equation}
\label{eq:laplace_recursion}
\E\left[\exp\left(-sZ_{v_\ell}(t)\right)\right]
\approx
\exp\left(-c^{\pi_{:\ell}}_{v_{\ell-1}v_\ell}(s)\right)
\E\left[\exp\left(-h^{\pi_{:\ell}}_{v_{\ell-1}v_\ell}(s)Z_{v_{\ell-1}}(t)\right)\right].
\end{equation}
Detailed derivations and the edge maps $c^{\pi_{:\ell}}_{v_{\ell-1}v_\ell}$ and $h^{\pi_{:\ell}}_{v_{\ell-1}v_\ell}$ are given in \aref{app:lt_path_decomp}--\ref{app:lt_closed}. 

On a DAG, several downstream branches may share the same parent and are only conditionally independent given the full parent history. Therefore, the marginal recursion \eref{eq:laplace_recursion} does not extend directly to the joint transform of the child branches. In the spirit of the approximate model of \citet{nicholson2023sequential}, we make the following ansatz to allow a recursive approximation of the joint transform. Heuristically, at large times and small mutation rates, the parent history can be represented by its endpoint, so the child branches are approximately conditionally independent given the parent endpoint.

\paragraph{Conditional endpoint ansatz}
At large times and small mutation rates, let $\pi=(v_0\to\cdots\to v_\ell)$ be a path, and let $Z_u^{\pi\to u}$ denote its contribution to a child $u$ of $v_\ell$. For $q_u\ge0$, $u\in\child(v_\ell)$, we approximate the joint conditional Laplace transform of the child branches using only the parent endpoint:
\begin{equation}
\label{eq:cond_laplace_recursion}
\E\left[
\left.
\exp\left(-\sum_{u\in\child(v_\ell)}q_uZ_u^{\pi\to u}(t)\right)
\right| Z_{v_\ell}^{\pi}(t)
\right]
\approx
\prod_{u\in\child(v_\ell)}
\exp\left(-c^{\pi\to u}_{v_\ell u}(q_u)
-h^{\pi\to u}_{v_\ell u}(q_u)Z_{v_\ell}^{\pi}(t)\right).
\end{equation}
For a single child, \eref{eq:cond_laplace_recursion} is the endpoint-conditioned counterpart of \eref{eq:laplace_recursion}.

\paragraph{Laplace transforms}

For the general system, we define the edge maps and the corresponding recursive maps $G^{[a]}_\ell$ and $C^{[a]}_\ell$ separately for the $Z^0$ ($a = 0$) and $Z^1$ ($a = 1$) root components in \eref{eq:elementary_decomposition}, with the edge maps given in \aref{app:lt_closed} and the recursive maps in \aref{app:lt_general}. Applying the conditional endpoint ansatz recursively (\aref{app:lt_general}) gives a product of an exponential factor determined by $C^{[a]}_0$ and the elementary Laplace transforms evaluated at $G^{[a]}_0$,
\begin{equation}
\label{eq:joint_LT_channels}
\E\left[e^{-\bs^\top\bz(t)}\right]
\approx
\prod_{v\in V}
\exp\left(
-C^{[0]}_0\left(\bs;(v)\right)
-z_{0,v}C^{[1]}_0\left(\bs;(v)\right)
\right)
\mathcal L_v^0\left(G^{[0]}_0\left(\bs;(v)\right);t\right)
\left[
\mathcal L_v^1\left(G^{[1]}_0\left(\bs;(v)\right);t\right)
\right]^{z_{0,v}} .
\end{equation}
For each $i\in V$, setting $\bs=s\be_i$ gives the marginal approximation.
We note that the branching process model in FiTree~\citep{Luo2025Bayesian} is a tree-structured special case of the DAG model considered here with zero initial populations and positive immigration from a static wild-type source. 
Thus, only the $Z^0$ component in \eref{eq:joint_LT_channels} contributes to FiTree's Laplace transform.

\paragraph{Numerical inverse Laplace transforms}
We recover the probabilities by numerical inversion of the recursive approximation in \eref{eq:joint_LT_channels} using the concentrated matrix exponential (CME) method of \citet{horvath2020numerical} for its monotonicity and numerical stability. With continuity-corrected targets $\tilde n_j=n_j+1/2$ and $\tilde\bn=(\tilde n_1,\dots,\tilde n_d)^\top$, the Laplace transforms on the right below are evaluated using \eref{eq:joint_LT_channels} and its marginal form:
\begin{align}
    & F(t,n_j)=\mathcal L^{-1}\left\{\frac{1}{s}\E\left[e^{-sZ_j(t)}\right]\right\}(\tilde n_j),\\
    & p(t,n_j)=\mathcal L^{-1}\left\{\frac{1-e^{-s}}{s}\E\left[e^{-sZ_j(t)}\right]\right\}(\tilde n_j),\label{eq:marginal_pmf_ILT}\\
    & p(t,\bn)=\mathcal L^{-1}\left\{\prod_{j=1}^d\frac{1-e^{-s_j}}{s_j}\E\left[e^{-\bs^\top\bz(t)}\right]\right\}(\tilde\bn).\label{eq:joint_pmf_ILT}
\end{align}
The factor $\frac{1-e^{-s}}{s}$ and the extension to $d$-dimensional hyperrectangles are described in \aref{app:ILT-rectangular}.

\paragraph{Computational cost}
Let $N_{\mathrm{ILT}}$ be the total number of CME transform nodes evaluated per dimension, including conjugate nodes. Let $d_i^{\mathrm{anc}}:=|\Anc^\ast(i)|$, and define
\[
N_{\mathrm{path}}^i:=
\sum_{v\in\Anc^\ast(i)}
\sum_{u\in\Anc^\ast(i)\cap\Desc(v)}
\left|\Pi_{v\to u}\right|,
\qquad
N_{\mathrm{path}}^{\mathrm{joint}}:=
\sum_{v\in V}
\sum_{u\in\Desc(v)}
\left|\Pi_{v\to u}\right|.
\]
These quantities count all nontrivial directed paths involved in the recursive edge-map calculations. Thus, one marginal inversion at node $i$ and one joint inversion cost
\[
\mathcal O\left(N_{\mathrm{ILT}}\left(N_{\mathrm{path}}^i+d_i^{\mathrm{anc}}\right)\right),
\qquad
\mathcal O\left(N_{\mathrm{ILT}}^d\left(N_{\mathrm{path}}^{\mathrm{joint}}+d\right)\right),
\]
respectively.


\section{Simulations}

We conduct simulation experiments to evaluate the accuracy and computational cost of our saddle-point approximation (SP) and large-time small-mutation-rate approximation (LTSM). 
We include the state-of-the-art central limit theorem approach (CLT) \citep{gunnarsson2023statistical} as a baseline and use Gillespie's exact stochastic simulation algorithm (SSA) \citep{gillespie1977exact} as the ground truth reference. 
In all experiments, SP uses a residual tolerance of $10^{-9}$ and a time step of $10^{-3}$. The CME order used for LTSM is $20$~\citep{horvath2020numerical}.
We implemented both approximations in C++17, available at \url{https://github.com/cbg-ethz/MTBP_Numerics}.
The repository also provides unit tests and workflows that reproduce the simulations and AML analysis in \sref{sec:application}.

\subsection{Diagnostic example}

\begin{figure}[!b]
    \centering
    \includegraphics[width=0.45\textwidth]{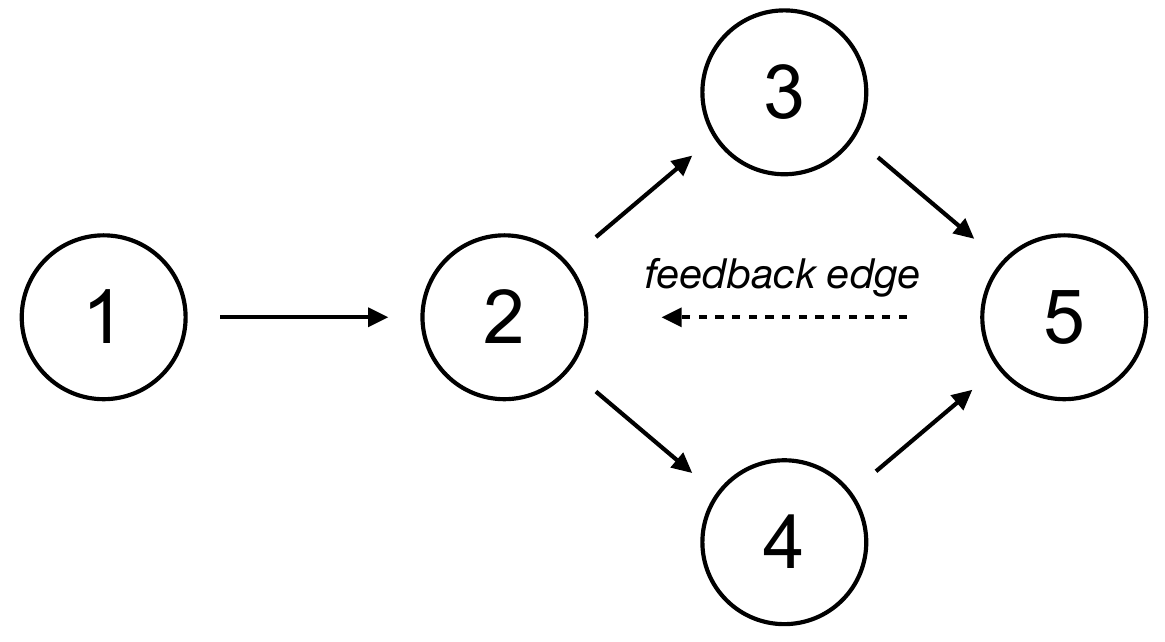}
    \caption{The extended diamond graph used for the diagnostic simulation experiments. The cyclic variant adds a feedback edge from type~5 to type~2}
    \label{fig:MTBP_simulation_setup}
\end{figure}

The diagnostic model is a five-type branching process on the extended diamond graph 
illustrated in \fref{fig:MTBP_simulation_setup}.
This structure contains divergent edges, convergent edges, and directed paths with more than one edge, providing a minimal yet representative example for biological interactions observed in other studies \citep{Beerenwinkel2009MarkovMutations, luo2023joint}. A cyclic variant adds the feedback edge $5 \to 2$, for which only SP and CLT apply, as LTSM requires a DAG. We study three root regimes: supercritical ($\lambda_1 > 0$), critical ($\lambda_1 = 0$), and subcritical ($\lambda_1 < 0$). The specific choice of the parameter values for all scenarios is described in \aref{sec:sim_params}. For each scenario, we run SSA with $M = 10^6$ independent trajectories and compare marginal $\cdf$'s for all five types at time points $t = 1, 2, \ldots, 10$. 

In Figures~\ref{fig:cdf_supercritical_main}, \ref{fig:cdf_supercritical_appendix}, \ref{fig:cdf_critical_appendix}, and \ref{fig:cdf_subcritical_appendix}, we see that both SP and LTSM closely match the SSA reference curves, while CLT shows much higher errors for almost all configurations except Type 1 in the subcritical case. This is because CLT relies strongly on the assumption of large initial population sizes. LTSM shows slightly larger errors at earlier times, also consistent with its asymptotic assumption (\fref{fig:error_over_time}).

\begin{figure}[!t]
    \centering
    \includegraphics[width=0.9\textwidth]{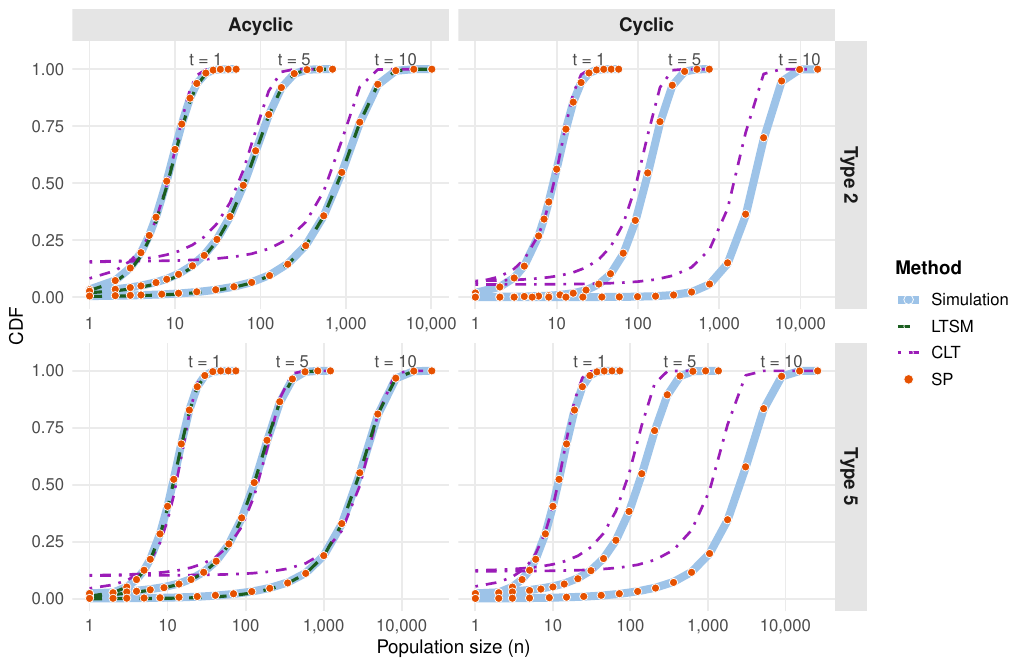}
    \caption{Marginal $\cdf$\ comparison for types~2 (intermediate; top row) and~5 (terminal; bottom row) in the supercritical scenario.
    Left panels: acyclic graph; right panels: cyclic graph counterpart, for which LTSM does not apply.
    The Gillespie reference is drawn as a thick pale blue band, and the approximations as thin marks on top of it: SP (orange markers), LTSM (green dashes), and CLT (purple dash-dots).
    All curves are evaluated on the shared grid of population counts used for the accuracy metric of \aref{sec:sim_metrics}.
    Three time points $t = 1, 5, 10$ are shown, labelled above the corresponding group of curves. The remaining time points are in \fref{fig:cdf_supercritical_appendix}}
    \label{fig:cdf_supercritical_main}
\end{figure}

\begin{figure}[!b]
    \centering
    \includegraphics[width=0.9\textwidth]{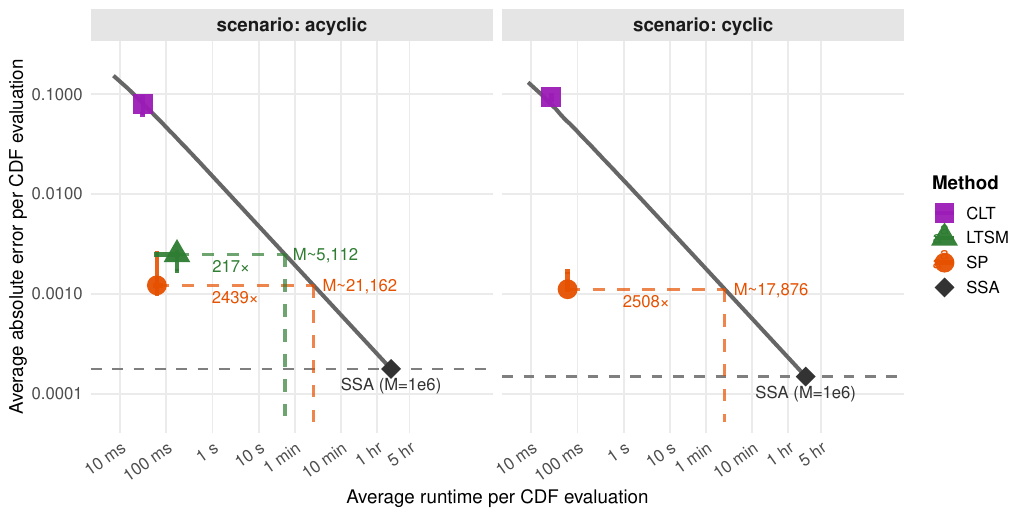}
    \caption{Accuracy--runtime trade-off for the supercritical scenario at $t = 10$.
    Each point gives the median per-type runtime (x-axis) and median per-type absolute $\cdf$\ error (y-axis), with interquartile-range bars across types.
    The solid grey curve is the projected SSA runtime--accuracy trade-off.
    The dashed horizontal line marks the SSA accuracy floor at $M = 10^6$.
    For SP and LTSM, dashed guides project each method to the SSA curve at matching accuracy. Labels report the speedup and the corresponding number of independent SSA trajectories $M$}
    \label{fig:runtime_vs_accuracy}
\end{figure}

\fref{fig:runtime_vs_accuracy} demonstrates the accuracy--runtime trade-off for the supercritical scenario at $t = 10$.
Each point summarises the median absolute $\cdf$\ error and the median wall-clock time per evaluation across all five types (see \aref{sec:sim_metrics} for detailed definitions).
SP and LTSM both achieve accurate $\cdf$\ evaluations in roughly $100$\,ms, which is orders of magnitude faster than SSA at the same accuracy level.
Reaching that accuracy with SSA alone would require a substantially larger number of trajectories and heavy parallelization.
This speed advantage is especially important for likelihood-based statistical inference, where methods such as MCMC require many probability evaluations.
CLT is slightly faster in raw runtime than SP and LTSM, but its accuracy is only comparable to SSA with very few trajectories in the parameter regimes considered here.
In most cases, SP shows superior performance to LTSM. When the target type is initialized at zero, however, LTSM can be more accurate than SP (\fref{fig:ltsm_special_case_type5_t10_mutdiv10}). 

Moreover, we compare the bivariate $\pmf$'s at time $t = 2, 5, 10$ for the type pairs $(Z_3,\, Z_4)$, $(Z_2,\, Z_5)$, and $(Z_3,\, Z_5)$, representing sibling, ancestor-descendant, and direct parent-child dependence, respectively (Figures~\ref{fig:bivariate_23_acyclic}--\ref{fig:bivariate_24_cyclic}). The close match between the SP and LTSM approximations against the SSA reference confirms that both methods accurately capture the joint dependence structure, in addition to the marginals.

\subsection{Large random DAGs}

To assess scalability, we generate 50 independent random DAGs for each combination of $N \in \{10, 20, 50\}$ nodes and target edge intensity $D$ (sparse: 1, medium: 2, dense: 3), where the expected number of edges is $ND$. The graph construction and the parameter ranges, sampled uniformly from biologically motivated intervals, are described in \aref{sec:sim_params_phase2}. Exact Gillespie simulation with $M = 10^4$ trajectories serves as the reference (worst-case $\cdf$\ standard error $0.005$). For each DAG, we evaluate marginal $\cdf$'s at $t=10$ for five structurally representative downstream nodes: deepest sink, deepest node, node with the most ancestors, highest in-degree node, and a randomly selected node from the upper half of graph depths.

\begin{figure}[!b]
    \centering
    \includegraphics[width=\textwidth]{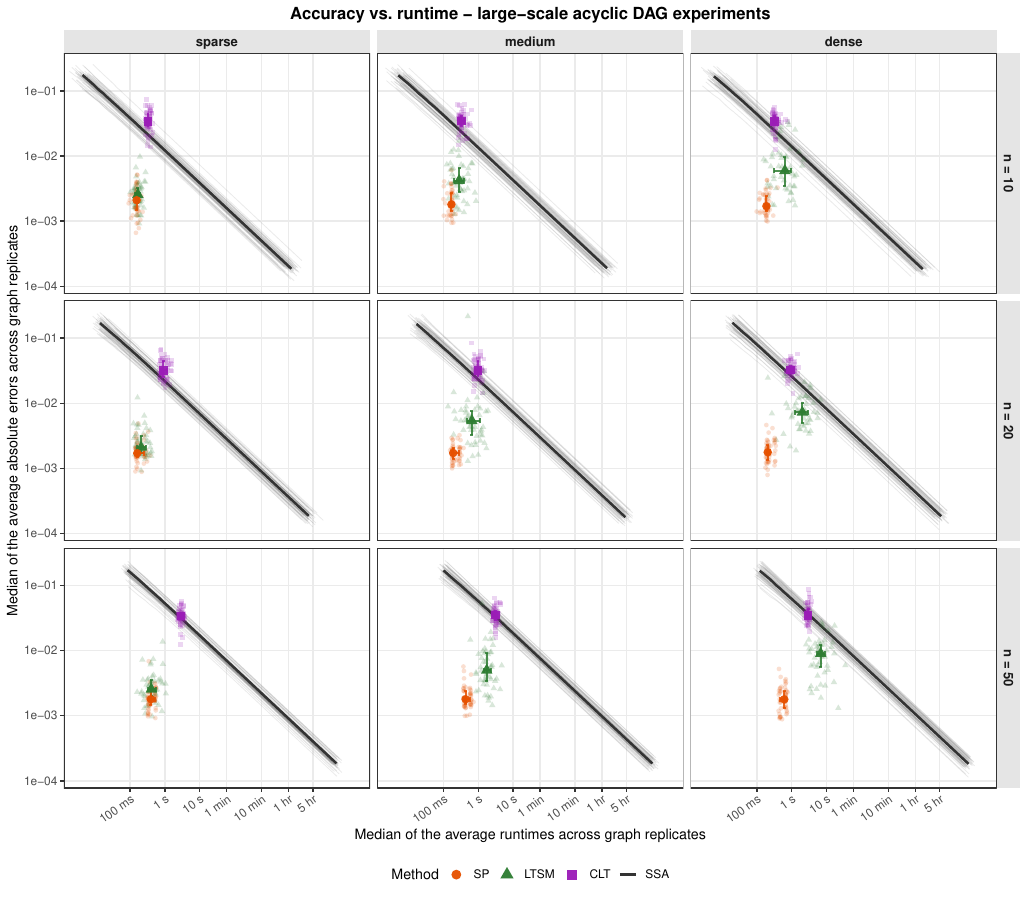}
    \caption{Accuracy versus runtime for the large-scale random DAG experiments.
    Each point represents the median for one replicate DAG and method (SP, LTSM, or CLT).
    Panels are faceted by number of nodes ($N$) and expected mean out-degree ($D$ = 1, 2, 3, corresponding to sparse, medium, dense graphs, respectively)}
    \label{fig:accuracy_vs_runtime_phase2}
\end{figure}

\fref{fig:accuracy_vs_runtime_phase2} shows the accuracy--runtime scatterplots for the large random DAG experiments.
SP consistently achieves the lowest error across all graph sizes and densities, with a median absolute $\cdf$\ error slightly above $10^{-3}$.
Its accuracy is highly stable across replicates, and runtimes scale well, remaining below one second even for $N = 50$.
LTSM performs comparably to SP on sparse graphs, but its accuracy decreases and runtime increases for denser and larger graphs.
This is because LTSM must integrate over all directed paths to a target node, so the number of terms grows with graph density.
For deeper nodes, the large-time and small-mutation-rate assumptions are also more susceptible to violation.
Nevertheless, the median error remains below $10^{-2}$ across all configurations, which is still within a useful range for inference, and the runtime remains competitive with SSA at matching accuracy.
CLT's accuracy is comparable to SSA with only a small number of trajectories and is far less accurate than SP or LTSM across all configurations.

\section{Application: acute myeloid leukemia relapse}
\label{sec:application}

Acute myeloid leukemia (AML) is a blood cancer characterized by the rapid proliferation of abnormal myeloid cells. 
Despite initial treatment responses, relapse is common and is often driven by resistant subclones, which are genetically identical sub-populations of cells that survive therapy \citep{Dagogo-Jack2018TumourTherapies}. 
Understanding the relapse dynamics of these subclones is crucial for improving treatment strategies. 
Here, we apply our saddle-point approximation to analyze the tumor evolution of patient AML-09 from a longitudinal single-cell sequencing study~\citep{Morita2020ClonalGenomics}. 
The patient was sequenced at diagnosis and treated with azacitidine (a type of chemotherapy) and sorafenib (a \gene{FLT3} inhibitor). 
Complete remission was clinically observed around $155$ days after treatment initiation, but no sequencing was performed.
The patient relapsed around $135$ days later and was sequenced a second time.
The results show that a subclone carrying the sorafenib-resistant \gene{FLT3} p.D835Y mutation~\citep{Smith2015FLT3}, which was small at diagnosis, became dominant at relapse.

\begin{figure}[!b]
    \centering
    \includegraphics[width=\textwidth]{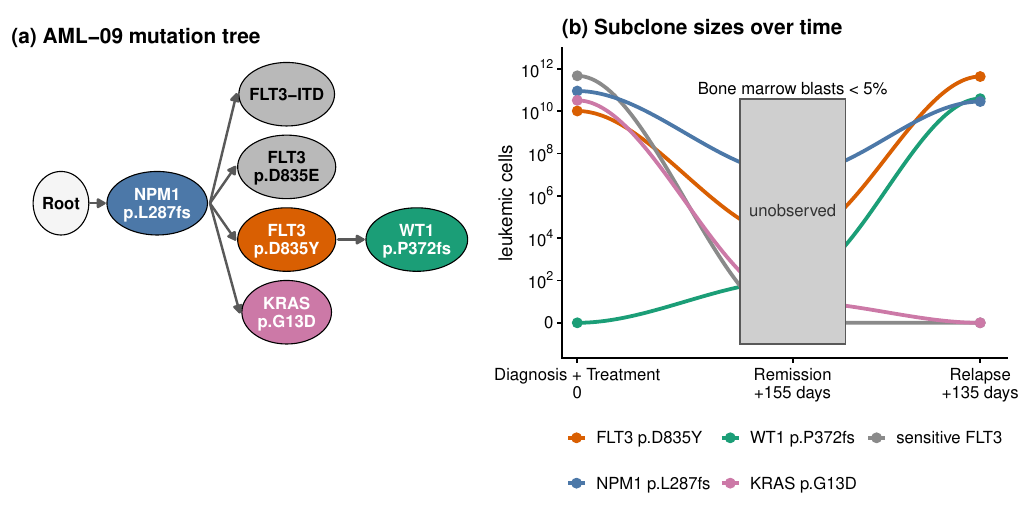}
  \caption{Patient AML-09. (\textbf{a}) Inferred single-cell mutation tree. Sorafenib-sensitive \gene{FLT3} subclones are shown in grey, resistant \gene{FLT3} p.D835Y in orange, \gene{NPM1} in blue, \gene{KRAS} in pink, and \gene{WT1} in green. (\textbf{b}) Schematic longitudinal tumor burden. Points at diagnosis and relapse are derived from clinical and single-cell data. The intervening curves are illustrative because the remission state is unobserved}
    \label{fig:aml09schematic}
\end{figure}

We want to understand how this relapse occurred. One possibility is that more cancer cells survived treatment than the clinical measurements could detect. Another possibility is that treatment gave some groups of cancer cells a growth advantage. Both may have happened. We therefore ask how many cancer cells may have remained at remission and how much faster they would need to grow to reach the numbers seen at relapse.
We map these questions to the following branching process model:
\begin{description}[leftmargin=12pt,labelindent=12pt,itemsep=2pt,parsep=0pt,topsep=3pt]
  \item[Mutation tree.] We use the mutation tree of this patient reconstructed by SCITE~\citep{jahn2016tree}, a stochastic search algorithm for inferring tumor evolutionary histories from noisy single-cell mutation profiles, as the type-transition graph $\mathcal G=(V,E)$ of a six-type branching process (\fref{fig:aml09schematic}a). The six non-root nodes form $V$, and each represents a subclone. The arrows between these nodes form $E$, and each represents the acquisition of an additional mutation. The root lies outside $V$ and provides direct immigration to the subclone carrying the \gene{NPM1} p.L287fs mutation, while the remaining subclones do not receive direct immigration.
  The population size vector $\bz(t)=(Z_v(t))_{v\in V}$ collects the number of cells in each subclone at time $t$.

  \item[Time window.] Since the process is Markovian, we initialize $t=0$ at remission and model only the interval of $t_\Delta=135$ days from remission to relapse. 

  \item[Subclone sizes.] We use clinical measurements of blasts (immature leukemic cells) and single-cell proportions to estimate subclone population sizes at diagnosis and relapse. We denote the estimated relapse size for subclone $v$ by $\hat z_{v,\mathrm{rel}}$. At remission, however, no measurements are available, so we treat the total population size $N_{\mathrm{rem}}$ as an unknown parameter. We use the clinical definition of complete remission (bone marrow blast percentage below 5\%) as an upper bound on $N_{\mathrm{rem}}$. The initial state $\bzs_0$ is then obtained by distributing $N_{\mathrm{rem}}$ across the subclones that are not sorafenib-sensitive, in their relative diagnosis proportions.

  \item[Parameters.] We fix the pre-treatment baseline parameters $\theta_v=\{\nu_v,\alpha_v,\beta_v,\bm\mu_v\}$ for $v\in V$ as in FiTree~\citep{Luo2025Bayesian}, whose tree-structured branching process is a special case of our model. To represent possible treatment effects, we vary only the growth component of this baseline and keep all rates constant from remission to relapse within each scenario.
\end{description}
Further details on model construction and parameter values are provided in \aref{sec:aml09_appendix}.

As a first model, we use a single multiplier $m_{\mathrm{shared}}$ to represent the post-treatment change in the diagnosis net growth rate of every subclone,
\[
  \lambda_{v, \mathrm{rel}} = m_{\mathrm{shared}}\lambda_{v, \mathrm{diag}}.
\]
We then scan over a grid of $(N_{\mathrm{rem}}, m_{\mathrm{shared}})$ values and compute the probability that a subclone $v$ exceeds its relapse size $\hat z_{v,\mathrm{rel}}$ using the marginal saddle-point $\cdf$ in \eref{eq:LR},
\[
  P(Z_v(t_\Delta)\ge \hat z_{v,\mathrm{rel}}
  \mid N_{\mathrm{rem}}, m_{\mathrm{shared}}),
\]
which indicates whether a parameter combination makes reaching at least the relapse count plausible.

Under this shared-effect model, the heatmaps for the \gene{NPM1} p.L287fs and \gene{FLT3} p.D835Y subclones (\fref{fig:aml09sensitivity}a and b) show a clear tradeoff: smaller remission sizes require larger growth increases to reach the relapse burden.
These two quantities are therefore not identifiable from the diagnosis and relapse measurements alone. 
If treatment did not alter the fitness landscape drastically, a larger residual population at remission would be the more plausible explanation. 
The two subclones nevertheless impose different requirements. 
Near the $5\%$ clinical remission ceiling ($\log_{10}N_{\mathrm{rem}} \approx 10$), the \gene{NPM1} p.L287fs subclone can reach its relapse burden with little growth increase.
For the sorafenib-resistant \gene{FLT3} p.D835Y subclone~\citep{Smith2015FLT3}, the transition to high exceedance probability occurs only around $m_{\mathrm{shared}}\approx 50$. 
This suggests stronger selection for the resistant subclone, potentially together with broader changes in the post-chemotherapy environment such as impaired immune surveillance. 
For both subclones, the transition from negligible to near-certain exceedance probability becomes sharp at larger remission burdens, reflecting reduced stochastic variation when more residual cells are present. 
These results therefore emphasize that the clinical remission threshold is too coarse to resolve the relapse-relevant residual disease state, and that more sensitive measurable residual disease assays~\citep{buccisano2012prognostic} could better constrain relapse-risk predictions.

\begin{figure}[!t]
    \centering
    \includegraphics[width=\textwidth]{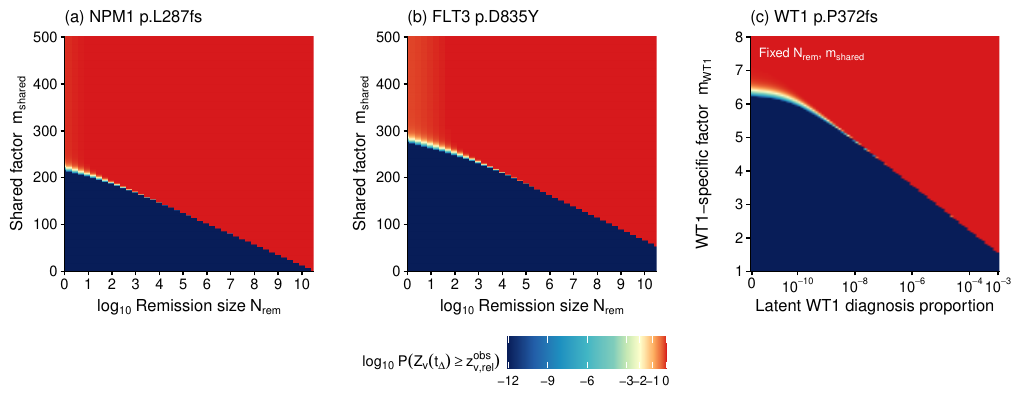}
    \caption{Relapse-count exceedance probabilities for detected subclones. (\textbf{a}) Shared-factor scan for \gene{NPM1} p.L287fs over remission size $N_{\mathrm{rem}}$ and $m_{\mathrm{shared}}$. (\textbf{b}) Shared-factor scan for \gene{FLT3} p.D835Y over remission size $N_{\mathrm{rem}}$ and $m_{\mathrm{shared}}$. (\textbf{c}) \gene{WT1} p.P372fs scan at the 5\% remission ceiling and $m_{\mathrm{shared}}\approx 53$ over latent \gene{WT1} diagnosis proportion and $m_{\mathrm{WT1}}$. Colour gives $\log_{10} P(Z_v(t_\Delta)\ge \hat z_{v,\mathrm{rel}})$, with red indicating high and blue low probability of reaching the relapse count}
    \label{fig:aml09sensitivity}
\end{figure}

The different patterns for \gene{NPM1} p.L287fs and \gene{FLT3} p.D835Y, together with the absence of the \gene{KRAS} p.G13D subclone at relapse, suggest that treatment did not induce a uniform growth change across all subclones. 
Therefore, we examine the newly emerged \gene{WT1} p.P372fs subclone separately (\fref{fig:aml09sensitivity}c). 
Since this subclone was not observed at diagnosis, we treat its diagnosis proportion as latent and bound it above by the detection limit of $0.1\%$~\citep{Morita2020ClonalGenomics}. 
We fix the total remission burden at the $5\%$ clinical upper bound and the shared growth multiplier at $m_{\mathrm{shared}}\approx 50$, at which the parent \gene{FLT3} p.D835Y subclone can reach its relapse burden.
We then introduce an additional \gene{WT1}-specific multiplier $m_{\mathrm{WT1}}$, so that
\[
\lambda_{\mathrm{WT1},\mathrm{rel}}=
m_{\mathrm{WT1}}m_{\mathrm{shared}}\lambda_{\mathrm{WT1},\mathrm{diag}}.
\]
Even when the latent diagnosis proportion is near the detection limit, the relapse burden requires an additional \gene{WT1}-specific growth increase, which becomes larger as the latent diagnosis proportion decreases (\fref{fig:aml09sensitivity}c).
This suggests that the diagnosis fitness landscape may have underestimated the positive epistasis between \gene{FLT3} and \gene{WT1}, or that the \gene{WT1} p.P372fs subclone was already present close to the detection limit at diagnosis but was not sampled.

Together, these analyses illustrate how our approximations support rapid parameter scans. They also provide a computational building block for future likelihood-based inference from longitudinal sequencing data.

\section{Discussion}

In this work, we introduced a saddle-point approximation (SP) and a large-time small-mutation-rate approximation (LTSM) to compute population size distributions for multi-type branching processes on directed graphs. 
These methods extend existing theory to more general graph structures and initial conditions. 
The saddle-point approximation is a fast numerical integrator, for which we derived practical checks for admissibility and established conditions for existence and uniqueness of the saddle point.
For directed acyclic graphs, the large-time small-mutation-rate approximation constructs closed-form approximate Laplace transforms that can be evaluated efficiently by numerical inversion.

In simulations, we benchmarked the accuracy and speed of both methods against the large-number approximation and stochastic simulation across a range of graph structures and initial conditions. 
SP provided the most stable accuracy, whereas LTSM remained competitive on sparse DAGs but became less favorable as the number of directed paths increased.
Both methods substantially outperformed the large-number approximation and provided orders of magnitude speedup over stochastic simulation at matching accuracy. 
The choice between the two methods can be guided by specific graph structures, initial conditions, and parameter regimes of interest. 
SP is generally preferred, particularly when the graph contains cycles.
For zero initial cells, however, the accuracy of SP may deteriorate, or an admissible saddle simply may not exist.
In such cases, LTSM provides a useful alternative for marginal probabilities or joint probabilities involving only a few types in sparse DAGs.

The branching process model here assumes constant rates, independent cell lineages, and exponentially distributed event times.
To generalize to other biologically relevant scenarios, such as treatment effects that change over time, cell-cell interactions, or cell-division times governed by progression through the cell cycle, alternative computational frameworks may need to be considered, such as neural surrogates for nonlinear reaction networks and differentiable stochastic simulation \citep{badolle2025interpretable,rijal2025differentiable}. 
However, the scalability of these methods to large graphs and long time horizons remains a challenge, and the approximations presented here may still provide useful building blocks for such extensions.

Despite these limitations, our methods enable efficient and accurate evaluation of population size probabilities, as shown in the AML example, laying the computational foundation for likelihood-based inference in cancer studies.
The inferred model can be used to make probabilistic forecasts, such as calculating whether a resistant or metastatic population exceeds a specified size at a future time. 
They also provide a computational tool for studying other expanding populations, such as microbial communities or ecological systems, where branching processes can be used to model growth and transitions between types.


\section*{Acknowledgements}

The authors would like to thank Prof. Dr. Koichi Takahashi for kindly providing the clinical data for patient AML-09 \citep{Morita2020ClonalGenomics}.

\section*{Declarations}

\subsection*{Funding}

This work was supported by ETH Zurich [Open ETH project SKINTEGRITY.CH (N.B.)] and the European Union's Horizon 2020 research and innovation program [grant number 951970, OLISSIPO project (N.B.)].

\subsection*{Competing interests}

The authors declare that they have no competing interests.

\subsection*{Data availability}

The single-cell sequencing data for the AML patient analyzed in this work are available from the original publication \citep{Morita2020ClonalGenomics} at NCBI BioProject ID \href{https://www.ncbi.nlm.nih.gov/Traces/study/?acc=PRJNA648656&o=acc_s%3Aa}{PRJNA648656}. For confidentiality, we provide only aggregated or approximated clinical data for this patient.

\subsection*{Code availability}

The code for the methods, simulations, and data analysis is available at \url{https://github.com/cbg-ethz/MTBP_Numerics}.

\subsection*{Author contributions}

Conceptualization: all authors; Methodology: all authors; Software: XL; Data Curation: XL; Formal Analysis: XL; Supervision: JK, NB; Visualization: XL; Writing - Original Draft Preparation: XL; Writing - Review \& Editing: all authors.

\bibliography{references}

@article{antal2011exact,
  title={Exact solution of a two-type branching process: models of tumor progression},
  author={Antal, Tibor and Krapivsky, PL},
  journal={Journal of Statistical Mechanics: Theory and Experiment},
  volume={2011},
  pages={P08018},
  year={2011},
  publisher={IOP Publishing},
  doi       = {10.1088/1742-5468/2011/08/P08018},
}

@book{athreya1972branching,
  title={Branching processes},
  author={Athreya, Krishna B and Ney, Peter E},
  year={1972},
  publisher={Springer},
  address={Berlin},
  doi       = {10.1007/978-3-642-65371-1},
}

@article{avanzini2019cancer,
  title={Cancer recurrence times from a branching process model},
  author={Avanzini, Stefano and Antal, Tibor},
  journal={PLoS Computational Biology},
  volume={15},
  optnumber={11},
  pages={e1007423},
  year={2019},
  publisher={Public Library of Science San Francisco, CA USA},
  doi       = {10.1371/journal.pcbi.1007423},
}

@inproceedings{awasthi2023fast,
  title={Fast computation of branching process transition probabilities via {ADMM}},
  author={Awasthi, Achal and Xu, Jason},
  booktitle={Proceedings of the 26th International Conference on Artificial Intelligence and Statistics},
  series={Proceedings of Machine Learning Research},
  volume={206},
  pages={2327--2347},
  year={2023},
  publisher={PMLR},
  url={https://proceedings.mlr.press/v206/awasthi23a.html}
}

@misc{badolle2025interpretable,
  title={Interpretable Neural Approximation of Stochastic Reaction Dynamics with Guaranteed Reliability},
  author={Badolle, Quentin and Theuer, Arthur and Fang, Zhou and Gupta, Ankit and Khammash, Mustafa},
  year={2025},
  howpublished={arXiv preprint arXiv:2512.06294},
  doi={10.48550/arXiv.2512.06294},
}

@article{barczy2023critical,
  author = {Barczy, M{\'a}ty{\'a}s and Bezd{\'a}ny, D{\'a}niel and Pap, Gyula},
  title = {Asymptotic behaviour of critical decomposable 2-type {Galton--Watson} processes with immigration},
  journal = {Stochastic Processes and their Applications},
  volume = {160},
  pages = {318--350},
  year = {2023},
  doi = {10.1016/j.spa.2023.03.003},
}

@article{barczy2024critical,
  author = {Barczy, M{\'a}ty{\'a}s and Bezd{\'a}ny, D{\'a}niel},
  title = {Asymptotic behavior of some strongly critical decomposable 3-type {Galton--Watson} processes with immigration},
  journal = {Advances in Applied Probability},
  year = {2026},
  pages = {1--41},
  doi = {10.1017/apr.2026.10071},
}

@Article{Beerenwinkel2009MarkovMutations,
  author   = {Beerenwinkel, N. and Sullivant, S.},
  title    = {{Markov models for accumulating mutations}},
  doi       = {10.1093/biomet/asp023},
  optissn     = {00063444},
  optnumber   = {3},
  pages    = {645--661},
  volume   = {96},
  arxivid  = {0709.2646},
  journal  = {Biometrika},
  optmonth    = {9},
  year     = {2009},
}

@book{bertsekas2009convex,
  title={Convex Optimization Theory},
  author={Bertsekas, Dimitri P.},
  year={2009},
  publisher={Athena Scientific},
  address={Belmont, MA}
}

@article{bianconi2013estimation,
  title={An estimation of the number of cells in the human body},
  author={Bianconi, Eva and Piovesan, Allison and Facchin, Federica and Beraudi, Alina and Casadei, Raffaella and Frabetti, Flavia and Vitale, Lorenza and Pelleri, Maria Chiara and Tassani, Simone and Piva, Francesco and others},
  journal={Annals of Human Biology},
  volume={40},
  pages={463--471},
  year={2013},
  publisher={Taylor \& Francis},
  doi       = {10.3109/03014460.2013.807878},
}

@article{bisel1956criteria,
  title={Criteria for the evaluation of response to treatment in acute leukemia},
  author={Bisel, Harry F},
  journal={Blood},
  volume={11},
  optnumber={7},
  pages={676--677},
  year={1956}
}

@article{buccisano2012prognostic,
  title={Prognostic and therapeutic implications of minimal residual disease detection in acute myeloid leukemia},
  author={Buccisano, Francesco and Maurillo, Luca and Del Principe, Maria Ilaria and Del Poeta, Giovanni and Sconocchia, Giuseppe and Lo-Coco, Francesco and Arcese, William and Amadori, Sergio and Venditti, Adriano},
  journal={Blood},
  volume={119},
  optnumber={2},
  pages={332--341},
  year={2012},
  publisher={American Society of Hematology Washington, DC},
  doi       = {10.1182/blood-2011-08-363291},
}

@book{Butler2007Saddlepoint,
  author    = {Ronald W. Butler},
  title     = {Saddlepoint Approximations with Applications},
  publisher = {Cambridge University Press},
  address   = {Cambridge},
  year      = {2007},
  doi       = {10.1017/CBO9780511619083},
}

@article{cheek2018mutation,
  title={Mutation frequencies in a birth--death branching process},
  author={Cheek, David and Antal, Tibor},
  journal={Annals of Applied Probability},
  volume={28},
  optnumber={6},
  pages={3922--3947},
  year={2018},
  publisher={JSTOR},
  doi       = {10.1214/18-AAP1413},
}

@Article{Dagogo-Jack2018TumourTherapies,
  author    = {Dagogo-Jack, Ibiayi and Shaw, Alice T.},
  title     = {{Tumour heterogeneity and resistance to cancer therapies}},
  doi       = {10.1038/nrclinonc.2017.166},
  optissn      = {17594782},
  optnumber    = {2},
  pages     = {81--94},
  volume    = {15},
  journal   = {Nature Reviews Clinical Oncology},
  optmonth     = {2},
  pmid      = {29115304},
  publisher = {Nature Publishing Group},
  year      = {2018},
}

@article{daniels1954saddlepoint,
  title={Saddlepoint Approximations in Statistics},
  author={Daniels, Henry E},
  journal={The Annals of Mathematical Statistics},
  volume={25},
  pages={631--650},
  year={1954},
  publisher={JSTOR},
  doi       = {10.1214/aoms/1177728652},
}

@article{daniels1982birth,
  title={The saddlepoint approximation for a general birth process},
  author={Daniels, H. E.},
  journal={Journal of Applied Probability},
  volume={19},
  optnumber={1},
  pages={20--28},
  year={1982},
  doi       = {10.1017/S0021900200028242},
}

@article{davison2021birthdeath,
  title={Parameter estimation for discretely observed linear birth-and-death processes},
  author={Davison, Anthony C. and Hautphenne, Sophie and Kraus, Andrea},
  journal={Biometrics},
  volume={77},
  optnumber={1},
  pages={186--196},
  year={2021},
  doi       = {10.1111/biom.13282},
}

@article{degunst2021population,
  title={Parameter estimation for multivariate population processes: a saddlepoint approach},
  author={de Gunst, Mathisca and Hautphenne, Sophie and Mandjes, Michel and Sollie, Birgit},
  journal={Stochastic Models},
  volume={37},
  optnumber={1},
  pages={168--196},
  year={2021},
  doi       = {10.1080/15326349.2020.1832895},
}

@article{durrett2010evolution,
  title={Evolution of resistance and progression to disease during clonal expansion of cancer},
  author={Durrett, Richard and Moseley, Stephen},
  journal={Theoretical Population Biology},
  volume={77},
  optnumber={1},
  pages={42--48},
  year={2010},
  publisher={Elsevier},
  doi       = {10.1016/j.tpb.2009.10.008},
}

@article{durrett2010evolutionary,
  title={Evolutionary dynamics of tumor progression with random fitness values},
  author={Durrett, Rick and Foo, Jasmine and Leder, Kevin and Mayberry, John and Michor, Franziska},
  journal={Theoretical Population Biology},
  volume={78},
  optnumber={1},
  pages={54--66},
  year={2010},
  publisher={Elsevier},
  doi       = {10.1016/j.tpb.2010.05.001},
}

@article{durrett2011intratumor,
  title={Intratumor heterogeneity in evolutionary models of tumor progression},
  author={Durrett, Rick and Foo, Jasmine and Leder, Kevin and Mayberry, John and Michor, Franziska},
  journal={Genetics},
  volume={188},
  optnumber={2},
  pages={461--477},
  year={2011},
  publisher={Oxford University Press},
  doi       = {10.1534/genetics.110.125724},
}

@book{durrett2015branching,
  title={Branching process models of cancer},
  author={Durrett, Richard},
  year={2015},
  publisher={Springer},
  address={Cham},
  doi       = {10.1007/978-3-319-16065-8},
}

@article{foo2013cancer,
  title={Cancer as a moving target: understanding the composition and rebound growth kinetics of recurrent tumors},
  author={Foo, Jasmine and Leder, Kevin and Mumenthaler, Shannon M},
  journal={Evolutionary Applications},
  volume={6},
  optnumber={1},
  pages={54--69},
  year={2013},
  publisher={Wiley Online Library},
  doi       = {10.1111/eva.12019},
}

@article{foo2013dynamics,
  title={Dynamics of cancer recurrence},
  author={Foo, Jasmine and Leder, Kevin},
  journal={Annals of Applied Probability},
 volume={23},
pages={1437--1468},
  year={2013},
  doi       = {10.1214/12-AAP876},
}

@article{gillespie1977exact,
  title={Exact stochastic simulation of coupled chemical reactions},
  author={Gillespie, Daniel T},
  journal={The Journal of Physical Chemistry},
  volume={81},
  optnumber={25},
  pages={2340--2361},
  year={1977},
  publisher={ACS Publications},
  doi       = {10.1021/j100540a008},
}

@article{gunnarsson2020understanding,
  title={Understanding the role of phenotypic switching in cancer drug resistance},
  author={Gunnarsson, Einar Bjarki and De, Subhajyoti and Leder, Kevin and Foo, Jasmine},
  journal={Journal of Theoretical Biology},
  volume={490},
  pages={110162},
  year={2020},
  publisher={Elsevier},
  doi       = {10.1016/j.jtbi.2020.110162},
}

@article{gunnarsson2023statistical,
  title={Statistical inference of the rates of cell proliferation and phenotypic switching in cancer},
  author={Gunnarsson, Einar Bjarki and Foo, Jasmine and Leder, Kevin},
  journal={Journal of Theoretical Biology},
  volume={568},
  pages={111497},
  year={2023},
  publisher={Elsevier},
  doi       = {10.1016/j.jtbi.2023.111497},
}

@article{horvath2020numerical,
  title={Numerical inverse Laplace transformation using concentrated matrix exponential distributions},
  author={Horv{\'a}th, G{\'a}bor and Horv{\'a}th, Ill{\'e}s and Almousa, Salah Al-Deen and Telek, Mikl{\'o}s},
  journal={Performance Evaluation},
  volume={137},
  pages={102067},
  year={2020},
  publisher={Elsevier},
  doi       = {10.1016/j.peva.2019.102067},
}

@article{hyrien2010saddlepoint,
  title={Saddlepoint approximations to the moments of multitype age-dependent branching processes, with applications},
  author={Hyrien, Olivier and Chen, Rong and Mayer-Pr{\"o}schel, Margot and Noble, Mark},
  journal={Biometrics},
  volume={66},
  optnumber={2},
  pages={567--577},
  year={2010},
  doi       = {10.1111/j.1541-0420.2009.01281.x},
}

@article{jahn2016tree,
  title={Tree inference for single-cell data},
  author={Jahn, Katharina and Kuipers, Jack and Beerenwinkel, Niko},
  journal={Genome Biology},
  volume={17},
  pages={86},
  year={2016},
  publisher={Springer},
  doi       = {10.1186/s13059-016-0936-x},
}

@book{kolassa2006series,
  title={Series Approximation Methods in Statistics},
  author={Kolassa, John E},
  edition={3},
  year={2006},
  publisher={Springer},
  address={New York},
  doi       = {10.1007/0-387-32227-2},
}

@article{lang2020predicting,
  title={Predicting colorectal cancer risk from adenoma detection via a two-type branching process model},
  optauthor={Lang, Brian M and Kuipers, Jack and Misselwitz, Benjamin and Beerenwinkel, Niko},
author={Lang, Brian M and Kuipers, Jack and others},
  journal={PLoS Computational Biology},
  volume={16},
  pages={e1007552},
  year={2020},
  publisher={Public Library of Science San Francisco, CA USA},
  doi       = {10.1371/journal.pcbi.1007552},
}

@article{leder2024parameter,
  title={Parameter estimation from single patient, single time-point sequencing data of recurrent tumors},
  author={Leder, Kevin and Sun, Ruping and Wang, Zicheng and Zhang, Xuanming},
  journal={Journal of Mathematical Biology},
  volume={89},
  pages={51},
  year={2024},
  publisher={Springer},
  doi       = {10.1007/s00285-024-02149-x},
}

@article{li2023comparison,
  title={A comparison of mutation and amplification-driven resistance mechanisms and their impacts on tumor recurrence},
  author={Li, Aaron and Kibby, Danika and Foo, Jasmine},
  journal={Journal of Mathematical Biology},
  volume={87},
  optnumber={4},
  pages={59},
  year={2023},
  publisher={Springer},
  doi       = {10.1007/s00285-023-01992-8},
}

@article{lugannani1980saddle,
  title={Saddle point approximation for the distribution of the sum of independent random variables},
  author={Lugannani, Robert and Rice, Stephen},
  journal={Advances in Applied Probability},
  volume={12},
  optnumber={2},
  pages={475--490},
  year={1980},
  publisher={Cambridge University Press},
  doi       = {10.2307/1426607},
}

@article{luo2023joint,
  title={Joint inference of exclusivity patterns and recurrent trajectories from tumor mutation trees},
  author={Luo, Xiang Ge and Kuipers, Jack and Beerenwinkel, Niko},
  journal={Nature Communications},
  volume={14},
  pages={3676},
  year={2023},
  publisher={Nature Publishing Group UK London},
  doi       = {10.1038/s41467-023-39400-w},
}

@article{Luo2025Bayesian,
    author = {Luo, Xiang Ge and Kuipers, Jack and Rupp, Kevin and Takahashi, Koichi and Beerenwinkel, Niko},
    title = {Bayesian inference of fitness landscapes via tree-structured branching processes},
    journal = {Bioinformatics},
    volume = {41},
    optnumber = {Supplement_1},
    pages = {i160-i169},
    year = {2025},
  doi       = {10.1093/bioinformatics/btaf193},
}

@article{meyn1993stabilityII,
  title={Stability of {Markovian} processes {II}: {Continuous-time} processes and sampled chains},
  author={Meyn, Sean P. and Tweedie, Richard L.},
  journal={Advances in Applied Probability},
  volume={25},
  number={3},
  pages={487--517},
  year={1993},
  doi={10.2307/1427521},
}

@article{meyn1993stabilityIII,
  title={Stability of {Markovian} processes {III}: {Foster--Lyapunov} criteria for continuous-time processes},
  author={Meyn, Sean P. and Tweedie, Richard L.},
  journal={Advances in Applied Probability},
  volume={25},
  number={3},
  pages={518--548},
  year={1993},
  doi={10.2307/1427522},
}

@article{Morita2020ClonalGenomics,
  title={Clonal evolution of acute myeloid leukemia revealed by high-throughput single-cell genomics},
  author={Morita, Kiyomi and Wang, Feng and Jahn, Katharina and Hu, Tianyuan and Tanaka, Tomoyuki and Sasaki, Yuya and Kuipers, Jack and Loghavi, Sanam and Wang, Sa A. and Yan, Yuanqing and Furudate, Ken and Matthews, Jairo and Little, Latasha and Gumbs, Curtis and Zhang, Jianhua and Song, Xingzhi and Thompson, Erika and Patel, Keyur P. and Bueso-Ramos, Carlos E. and DiNardo, Courtney D. and Ravandi, Farhad and Jabbour, Elias and Andreeff, Michael and Cortes, Jorge and Bhalla, Kapil and Garcia-Manero, Guillermo and Kantarjian, Hagop and Konopleva, Marina and Nakada, Daisuke and Navin, Nicholas and Beerenwinkel, Niko and Futreal, P. Andrew and Takahashi, Koichi},
  journal={Nature Communications},
  volume={11},
  optnumber={1},
  pages={5327},
  year={2020},
  publisher={Nature Publishing Group},
  doi       = {10.1038/s41467-020-19119-8},
}

@article{nguyen2023stochastic,
  title={Stochastic models of stem cells and their descendants under different criticality assumptions},
  author={Nguyen, Nam H. and Kimmel, Marek},
  journal={Stochastic Models},
  volume={39},
  optnumber={1},
  pages={249--264},
  year={2023},
  doi={10.1080/15326349.2022.2093374},
}

@article{nicholson2019competing,
  title={Competing evolutionary paths in growing populations with applications to multidrug resistance},
  author={Nicholson, Michael D and Antal, Tibor},
  journal={PLoS Computational Biology},
  volume={15},
  optnumber={4},
  pages={e1006866},
  year={2019},
  publisher={Public Library of Science San Francisco, CA USA},
  doi       = {10.1371/journal.pcbi.1006866},
}

@article{nicholson2023sequential,
  title={Sequential mutations in exponentially growing populations},
  author={Nicholson, Michael D and Cheek, David and Antal, Tibor},
  journal={PLoS Computational Biology},
  volume={19},
  pages={e1011289},
  year={2023},
  publisher={Public Library of Science San Francisco, CA USA},
  doi       = {10.1371/journal.pcbi.1011289},
}

@book{norris1998markov,
  title={Markov chains},
  author={Norris, James Robert},
  year={1998},
  publisher={Cambridge University Press},
  address={Cambridge},
  doi       = {10.1017/CBO9780511810633},
}

@article{paterson2026fast,
  title={A fast numerical integration scheme for clonal expansion processes on graphs},
  author={Paterson, Chay and Gao, Miaomiao and Hellier, Joshua and Luebeck, Georg and Wedge, David C and Bozic, Ivana},
  journal={PLOS Computational Biology},
  volume={22},
  optnumber={8},
  pages={e1012784},
  year={2026},
  doi       = {10.1371/journal.pcbi.1012784},
  publisher={Public Library of Science San Francisco, CA USA}
}

@book{perko2001differential,
  title={Differential Equations and Dynamical Systems},
  author={Perko, Lawrence},
  volume={7},
  edition={3},
  year={2001},
  publisher={Springer Science \& Business Media},
  address={New York},
  doi       = {10.1007/978-1-4613-0003-8},
}

@article{quinn1989hazard,
  title={Calculating the hazard function and probability of tumor for cancer risk assessment when the parameters are time-dependent},
  author={Quinn, Dennis W.},
  journal={Risk Analysis},
  volume={9},
  optnumber={3},
  pages={407--413},
  year={1989},
  doi       = {10.1111/j.1539-6924.1989.tb01006.x},
}

@article{renshaw2000applying,
  title={Applying the saddlepoint approximation to bivariate stochastic processes},
  author={Renshaw, Eric},
  journal={Mathematical Biosciences},
  volume={168},
  optnumber={1},
  pages={57--75},
  year={2000},
  publisher={Elsevier},
  doi       = {10.1016/S0025-5564(00)00037-7},
}

@article{rijal2025differentiable,
  title={A differentiable Gillespie algorithm for simulating chemical kinetics, parameter estimation, and designing synthetic biological circuits},
  author={Rijal, Krishna and Mehta, Pankaj},
  journal={eLife},
  volume={14},
  pages={RP103877},
  year={2025},
  doi       = {10.7554/eLife.103877.3},
}

@article{Smith2015FLT3,
  title={{FLT3 D835 mutations confer differential resistance to type II FLT3 inhibitors}},
  author={Smith, Catherine C and Lin, Kimberly and Stecula, Adrian and Sali, Andrej and Shah, Neil P},
  journal={Leukemia},
  volume={29},
  optnumber={12},
  pages={2390--2392},
  year={2015},
  publisher={Nature Publishing Group},
  doi       = {10.1038/leu.2015.165},
}

@article{stutz2022computational,
  title={Computational tools for assessing gene therapy under branching process models of mutation},
  author={Stutz, Timothy C. and Sinsheimer, Janet S. and Sehl, Mary E. and Xu, Jason},
  journal={Bulletin of Mathematical Biology},
  volume={84},
  optnumber={1},
  pages={15},
  year={2022},
  doi       = {10.1007/s11538-021-00969-2},
}

@article{wu2025statistical,
  title={A statistical framework for detecting therapy-induced resistance from drug screens},
  author={Wu, Chenyu and Gunnarsson, Einar Bjarki and Foo, Jasmine and Leder, Kevin},
  journal={npj Systems Biology and Applications},
  volume={11},
  optnumber={1},
  pages={88},
  year={2025},
  publisher={Nature Publishing Group UK London},
  doi       = {10.1038/s41540-025-00560-8},
}

@inproceedings{xu2015compressed,
  title={Efficient transition probability computation for continuous-time branching processes via compressed sensing},
  author={Xu, Jason and Minin, Vladimir N.},
  booktitle={Proceedings of the Thirty-First Conference on Uncertainty in Artificial Intelligence},
  pages={952--961},
  year={2015},
  publisher={AUAI Press},
  url={https://pmc.ncbi.nlm.nih.gov/articles/PMC4775097/}
}

@article{xu2015likelihood,
  title={Likelihood-based inference for discretely observed birth--death--shift processes, with applications to evolution of mobile genetic elements},
  author={Xu, Jason and Guttorp, Peter and Kato-Maeda, Midori and Minin, Vladimir N.},
  journal={Biometrics},
  volume={71},
  optnumber={4},
  pages={1009--1021},
  year={2015},
  doi={10.1111/biom.12352},
}

@article{zhang2023waiting,
  title={Waiting times in a branching process model of colorectal cancer initiation},
  author={Zhang, Ruibo and Ukogu, Obinna A and Bozic, Ivana},
  journal={Theoretical Population Biology},
  volume={151},
  pages={44--63},
  year={2023},
  publisher={Elsevier},
  doi       = {10.1016/j.tpb.2023.04.001},
}

@article{zhang2024accumulation,
  title = {Accumulation of Oncogenic Mutations During Progression from Healthy Tissue to Cancer},
  author = {Zhang, Ruichu and Bozic, Ivana},
  journal = {Bulletin of Mathematical Biology},
  volume = {86},
  pages = {142},
  year = {2024},
  doi       = {10.1007/s11538-024-01372-3},
  opturl = {https://optdoi.org/10.1007/s11538-024-01372-3}
}

\newpage

\appendix


\section{Saddle-point approximations}

This section contains the technical results and mathematical proofs for the saddle-point approximations. It is organized as follows. \sref{app:saddle_illustration} provides a one-type illustration of the saddle-point approximation compared to the one based on the central limit theorem. \sref{app:saddle_proofs} contains the proofs of the saddle-point results. \sref{app:saddle_counterexample} gives a three-type counterexample showing that \asref{assum:nondegen} is not sufficient for the existence of a saddle point. \sref{app:saddle_univariate} derives the variational equations for the univariate marginals, and \sref{app:algorithms} provides algorithmic details.

\subsection{One-type illustration}
\label{app:saddle_illustration}

Consider a one-type branching process $Z(t)$ with immigration rate $\nu>0$, birth rate $\alpha>0$, death rate $\beta\ge0$, and no initial cells. Suppose $\lambda=\alpha-\beta\ne0$, and write $\rho=\nu/\alpha$. It is well-known that
\[
Z(t)\sim\mathrm{Negative\mbox{-}Binomial}\left(\rho,p(t)\right),
\qquad
p(t)=\frac{\lambda}{\alpha e^{\lambda t}-\beta},
\]
for every $t>0$ (see the proof of \lemmaref{lemma:1type}). 
At a target $n>0$, the saddle point and the quantities entering \eref{eq:SP} are
\begin{equation}
\label{eq:nb_saddle}
\hat u(n)=\log\frac{n}{n+\rho}-\log\left(1-p(t)\right),
\qquad
\Sigma\left(\hat u(n)\right)=\frac{n\left(n+\rho\right)}{\rho},
\qquad
K\left(t,\hat u(n)\right)=\rho\log\frac{p(t)\left(n+\rho\right)}{\rho}.
\end{equation}
Intuitively, \eref{eq:SP} fits a different Gaussian at every $n$, with mean $n$ and variance $\Sigma\left(\hat u(n)\right)$. 
The central limit theorem instead uses a single Gaussian with the mean $\rho\left(1-p(t)\right)/p(t)$ and the variance $\rho\left(1-p(t)\right)/p(t)^2$ of $Z(t)$, regardless of $n$.
\fref{fig:sp_one_type_illustration} illustrates their differences for the case $\nu=2$, $\alpha=1$, $\beta=1/2$, and $t=8\log2$.

\begin{figure}[H]
    \centering
    \includegraphics[width=0.817\textwidth]{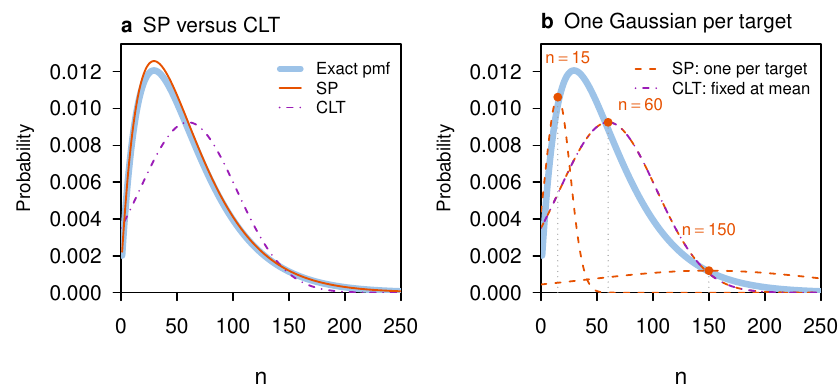}
    \caption{Saddle-point (SP) and central-limit theorem (CLT) approximations of the exact $\pmf$ for a one-type branching process with $\nu=2$, $\alpha=1$, $\beta=1/2$, $Z(0)=0$, and $t=8\log2$.
    (\textbf{a}) The $\pmf$ and the two approximations. (\textbf{b}) The Gaussian that \eref{eq:SP} fits at $n=15$, $60$, and $150$}
    \label{fig:sp_one_type_illustration}
\end{figure}

\subsection{Proofs for saddle-point approximation}
\label{app:saddle_proofs}

In this section, we first establish the connection between the characteristic ODE (\eref{eq:char}), which we can evaluate numerically, and the generating function of the process, namely \eref{eq:exact_solution}.
\citet{quinn1989hazard} and \citet{paterson2026fast} approached this problem through the PDE representation of the generating function, and their applications considered the arguments $q_i\in[0,1]$, where the generating function is automatically finite. 
In our saddle-point approximations, however, the arguments $q_i=e^{u_i}$ may go beyond one whenever the target counts $n_i$ exceed the expected counts. 
In this case, using the PDE representation directly would be circular unless finiteness of the generating function were established independently.
Therefore, we show in \thmref{thm:characteristic_mgf} that, whenever the characteristic remains finite on $[0,t]$, the characteristic representation \eref{eq:exact_solution} equals the true finite generating function. 
Together with \lemmaref{lemma:admissible_mgf_interior}, this justifies the cumulant generating function $K$ and its derivatives on the admissible domain $\mathcal U_t$.
The remaining results then establish uniqueness of the saddle point under \asref{assum:nondegen}, and existence for every positive target under the stronger \asref{assum:all-active}.

\begin{lemma}
\label{lemma:positive_characteristics}
For every $\bq\in[0,\infty)^d$, the solution of
\eref{eq:char} satisfies
\[
\gamma_i(0;\bq)=q_i\ge0,
\qquad
\gamma_i(s;\bq)>0
\]
for every $i\in V$ and every $s>0$ in its maximal interval of existence.
\end{lemma}

\begin{proof}
Fix $\bq\in[0,\infty)^d$ and $i\in V$. If $q_i>0$, continuity gives
$\gamma_i(s;\bq)>0$ for all sufficiently small $s>0$. If $q_i=0$, then
\[
\dot{\gamma}_i(0;\bq)=X_i(\bq)=\beta_i>0,
\]
so again $\gamma_i(s;\bq)>0$ for all sufficiently small $s>0$.

Suppose that $\gamma_i(\cdot;\bq)$ has a first positive zero $\tau$. Then
$\gamma_i(s;\bq)>0$ for $0<s<\tau$ and $\gamma_i(\tau;\bq)=0$, so
\[
\dot{\gamma}_i(\tau;\bq)
=
\lim_{h\downarrow0}
\frac{\gamma_i(\tau;\bq)-\gamma_i(\tau-h;\bq)}{h}
\le 0.
\]
On the other hand, \eref{eq:char} and \eref{eq:X} give

\[
\dot{\gamma}_i(\tau;\bq)
=
X_i\left(\bg(\tau;\bq)\right)
=
\beta_i>0,
\]
a contradiction. Since $i\in V$ was arbitrary, the result follows.
\end{proof}

\begin{theorem}
\label{thm:characteristic_mgf}
Fix $t>0$ and $\bq\in[0,\infty)^d$. If
$\bg(\cdot;\bq)$ is finite on $[0,t]$, then
\[
\E\left[\bq^{\bz(t)}\right]
=
\prod_{i\in V}\gamma_i(t;\bq)^{z_{0,i}}
\exp\left\{
\int_0^t\sum_{i\in V}\nu_i
\bigl(\gamma_i(s;\bq)-1\bigr)\rmd s
\right\}
<\infty.
\]
\end{theorem}

\begin{proof}
Let $\bz^{(i)}$ be the process started from one
type-$i$ particle without immigration, and define
\[
c_i:=\alpha_i+\beta_i+\sum_{j:(i,j)\in E}\mu_{ij},
\qquad
H_i(\mathbf x):=\beta_i+\alpha_i x_i^2
+\sum_{j:(i,j)\in E}\mu_{ij}x_ix_j.
\]
For a nonnegative function $\mathbf f$ on $[0,t]$, set
\[
\left(\mathcal T\mathbf f\right)_i(s)
=
e^{-c_is}q_i+
\int_0^s e^{-c_i\left(s-r\right)}
H_i\left(\mathbf f(r)\right)\rmd r.
\]
The characteristic equation \ref{eq:char} is equivalent to
\[
\bg(\cdot;\bq)=\mathcal T\bg(\cdot;\bq).
\]

We first show that $\bz^{(i)}$ is nonexplosive. Let $\mathbf Y_k$ be its jump chain and $S_k$ its holding times. On the event of infinitely many jumps,
\[
\left\lVert\mathbf Y_{k-1}\right\rVert_1\le k,
\qquad
r\left(\mathbf Y_{k-1}\right)
=
\sum_{j\in V}c_jY_{k-1,j}
\le ck,
\]
where $c:=\max_{j\in V}c_j$. By the jump-chain construction in Section~2.6 of \citet{norris1998markov},
\[
S_k
=
\frac{T_k}{r\left(\mathbf Y_{k-1}\right)}
\ge
\frac{T_k}{ck},
\]
where the $T_k$ are independent $\operatorname{Exp}(1)$ variables. Since $T_k/(ck)\sim\operatorname{Exp}(ck)$ and $\sum_{k=1}^{\infty}\frac1{ck}=\infty$,
Theorem~2.3.2(ii) of \citet{norris1998markov} gives
\[
\sum_{k=1}^{\infty}\frac{T_k}{ck}=\infty
\quad\Longrightarrow\quad
\sum_{k=1}^{\infty}S_k=\infty
\qquad
\text{almost surely on the event of infinitely many jumps}.
\]
Hence, infinitely many jumps cannot occur in finite time, so $\bz^{(i)}$ is nonexplosive.

Next, define
\[
\mathbf f^{(0)}:=\mathbf0,
\qquad
\mathbf f^{(n+1)}:=\mathcal T\mathbf f^{(n)}.
\]
For the type-$i$ process, there are three possible transitions: death ($i \rightarrow \emptyset$) at rate $\beta_i$, replication ($i \rightarrow i+i$) at rate $\alpha_i$, and mutation ($i \rightarrow i+j$) at rate $\mu_{ij}$ for each $(i,j)\in E$. 
Consider the binary genealogy of this process, where a particle either dies or replaces itself with two new particles at each transition.
By the memoryless property of exponential holding times, this construction has the same count process as the original model.
Let $\mathcal V(s)$ be the set of all particles born by time $s$ in this genealogy, including those that have died, and let $|v|$ denote the generation of particle $v$, with the initial particle in generation $0$. Define
\[
A_n(s):=
\left\{
\max_{v\in\mathcal V(s)}|v|<n
\right\},
\qquad
A_0(s):=\varnothing.
\]
We prove by induction that
\[
f_i^{(n)}(s)
=
\E_i\left[
\bq^{\bz^{(i)}(s)}
\mathds{1}_{A_n(s)}
\right]
\]
for all $n\ge0$.
For $n=0$,
\[
f_i^{(0)}(s)=0
=
\E_i\left[
\bq^{\bz^{(i)}(s)}
\mathds{1}_{A_0(s)}
\right].
\]
Assume that, for every $k\in V$ and $0\le r\le s$,
\[
f_k^{(n)}(r)
=
\E_k\left[
\bq^{\bz^{(k)}(r)}
\mathds{1}_{A_n(r)}
\right].
\]
Condition on the first transition of the initial type-$i$ particle. If no transition occurs before $s$, the contribution is $e^{-c_is}q_i$. If the first transition occurs at time $s-r$, then after death the contribution is $1$, while after replication or mutation the descendant families evolve independently for time $r$. Moreover, $A_{n+1}(s)$ holds exactly when each descendant family satisfies the corresponding event $A_n(r)$. Hence, by the induction hypothesis,
\[
\begin{aligned}
\E_i\left[
\bq^{\bz^{(i)}(s)}
\mathds{1}_{A_{n+1}(s)}
\right]
&=
e^{-c_is}q_i +
\int_0^s e^{-c_i(s-r)}
\left[
\beta_i
+\alpha_i\left(f_i^{(n)}(r)\right)^2
+\sum_{j:(i,j)\in E}
\mu_{ij}f_i^{(n)}(r)f_j^{(n)}(r)
\right]\rmd r\\
&=
f_i^{(n+1)}(s),
\end{aligned}
\]
and the claim follows.

On $A_n(s)$, the population is finite and bounded by $2^n$, so the truncated expectation is finite. Nonexplosion of the process implies that
\[
\max_{v\in\mathcal V(s)}|v|<\infty
\qquad \Longrightarrow \qquad
\mathds{1}_{A_n(s)}\uparrow1
\qquad\text{almost surely}.
\]
Since $\mathcal T$ is coordinatewise nondecreasing and
$\bg(\cdot;\bq)=\mathcal T\bg(\cdot;\bq)$, by
\lemmaref{lemma:positive_characteristics},
\[
\mathbf0\le\mathbf f^{(n)}
\le\mathbf f^{(n+1)}
\le\bg(\cdot;\bq).
\]
Hence, by monotone convergence,
\[
G_i(s;\bq)
:=
\E_i\left[\bq^{\bz^{(i)}(s)}\right]
=
\lim_{n\to\infty} \E_i\left[\bq^{\bz^{(i)}(s)} \mathds{1}_{A_n(s)}\right]
=
\lim_{n\to\infty}f_i^{(n)}(s)
\le\gamma_i(s;\bq)<\infty.
\]
Since $\mathbf f^{(n+1)}=\mathcal T\mathbf f^{(n)}$ and $\mathbf f^{(n)}\uparrow\mathbf G$, monotone convergence gives
\[
\mathbf G(s;\bq)=\mathcal T\mathbf G(s;\bq).
\]
Both $\mathbf G(\cdot;\bq)$ and $\bg(\cdot;\bq)$ lie in a bounded subset of $[0,\infty)^d$, on which $\mathbf H=\left(H_i\right)_{i\in V}$ is Lipschitz. Therefore, for some $L<\infty$,
\[
\left\lVert\mathbf G(s;\bq)-\bg(s;\bq)\right\rVert_\infty
\le
L\int_0^s
\left\lVert\mathbf G(r;\bq)-\bg(r;\bq)\right\rVert_\infty\rmd r.
\]
Gronwall's inequality yields
\begin{equation}
\label{eq:single_family_transform}
G_i(s;\bq)=\gamma_i(s;\bq),
\qquad i\in V,\quad 0\le s\le t.
\end{equation}
The independent families descending from the initial particles therefore contribute
\[
\prod_{i\in V}\gamma_i(t;\bq)^{z_{0,i}}.
\]
Represent the immigration mechanisms by independent Poisson processes, independent of all particle-founded families. If $\nu_i=0$, the type-$i$ immigration contribution is $1$. Suppose therefore that $\nu_i>0$.
Let $I_i(t)\sim\operatorname{Poisson}(\nu_it)$ and let $R_{i,1},\ldots,R_{i,I_i(t)}$ be the arrival times. By Theorem~2.4.6 of \citet{norris1998markov}, conditional on $I_i(t)=m$, these are the order statistics of $m$ independent uniform variables on $[0,t]$. Thus,
\[
\begin{aligned}
\E\left[
\prod_{k=1}^{I_i(t)}
\gamma_i\left(t-R_{i,k};\bq\right)
\right]
&=
e^{-\nu_it}
\sum_{m=0}^{\infty}
\frac{\left(\nu_it\right)^m}{m!}
\left[
\frac1t\int_0^t\gamma_i(r;\bq)\rmd r
\right]^m\\
&=
\exp\left\{
\nu_i\int_0^t
\left(\gamma_i(r;\bq)-1\right)\rmd r
\right\}.
\end{aligned}
\]
Multiplying the independent contributions from all initial particles and immigrant families gives
\[
\E\left[\bq^{\bz(t)}\right]
=
\prod_{i\in V}\gamma_i(t;\bq)^{z_{0,i}}
\exp\left\{
\int_0^t\sum_{i\in V}\nu_i
\left(\gamma_i(s;\bq)-1\right)\rmd s
\right\}.
\]
Since $\bg(\cdot;\bq)$ is continuous and finite on $[0,t]$, the right-hand side is finite.
\end{proof}

\begin{lemma}
\label{lemma:admissible_mgf_interior}
Fix $t>0$. The set $\mathcal U_t$ is open and satisfies
\[
\mathcal U_t\subseteq\operatorname{int}\mathcal M_t.
\]
\end{lemma}

\begin{proof}
Fix $\bu_0\in\mathcal U_t$ and set $\bq_0=e^{\bu_0}$. Since the vector
field $X$ in \eref{eq:X} is polynomial, and hence $C^1$ and locally
Lipschitz, Theorem~4 in Section 2.4 of
\citet{perko2001differential} gives $\delta>0$ such that
$\bg(\cdot;\bq)$ exists on $[0,t]$ whenever
\[
\lVert\bq-\bq_0\rVert<\delta.
\]
By continuity of the map $\bv\mapsto e^{\bv}$, there exists an open
neighborhood $B$ of $\bu_0$ such that
\[
\left\lVert e^{\bv}-e^{\bu_0}\right\rVert<\delta
\]
for every $\bv\in B$. Hence, $B\subseteq\mathcal U_t$, and $\mathcal U_t$
is open. By \thmref{thm:characteristic_mgf},
$\mathcal U_t\subseteq\mathcal M_t$, and hence
$\mathcal U_t\subseteq\operatorname{int}\mathcal M_t$.
\end{proof}

\begin{lemma}
\label{lemma:support_saddle}
Fix $t>0$. Under \asref{assum:nondegen},
\[
P\left(\bz(t)=\mathbf 0\right)>0
\qquad \text{and}\qquad
P\left(\bz(t)=L\be_i\right)>0,
\]
for every $i\in V$ and every integer $L\ge 1$.
\end{lemma}

\begin{proof}
At a fixed finite state, suppose the available jumps have rates $r_1, \dots, r_k>0$ with total rate $R=\sum_{j=1}^k r_j$. 
By competing exponentials, the probability that jump $j$ occurs next within an interval of length $\Delta>0$ is 
\[
\frac{r_j}{R}\left(1-e^{-R\Delta}\right)>0,
\]
and the probability that no jumps occur is $e^{-R\Delta}>0$.
Together with the Markov property, we can iteratively construct a finite sequence of jumps with positive rates that occurs before time $t$ with positive probability.

For $\bz(t)=\mathbf0$, we can assign each initial cell to die before it replicates or mutates, and require no immigration on $[0,t]$.
Fix $i\in V$ and $L\ge1$. By \asref{assum:nondegen}, choose $j\in\mathcal A$ and a directed path $j=v_0\to v_1\to\cdots\to v_r=i$. Retain one initial type-$j$ cell if $z_{0,j}>0$, and otherwise admit one type-$j$ immigrant. Require every other initial cell to die before producing offspring and permit no other immigration. Along the selected lineage, perform one mutation $v_{\ell-1}\to v_\ell$ for each $\ell=1,\dots,r$, killing the parent after each mutation. The remaining type-$i$ cell then undergoes $L-1$ replications, followed by no further jumps. This finite sequence has positive probability and ends at $\bz(t)=L\be_i$.
\end{proof}

\begin{proposition}
\label{prop:saddle_uniqueness}
Fix $t>0$. Under \asref{assum:nondegen}, $\Sigma(\bu)$ is positive definite for every $\bu\in\mathcal U_t$, and $\mathbf m:\mathcal U_t\to\R^d$ is injective.
\end{proposition}

\begin{proof}
By \lemmaref{lemma:admissible_mgf_interior}, $\mathcal U_t\subseteq\operatorname{int}\mathcal M_t$.
Let $\bu,\bv\in\mathcal M_t$ and $\theta\in(0,1)$. By H\"older's inequality,
\[
\begin{aligned}
M\left(t,\theta\bu+(1-\theta)\bv\right)
&=\E\left[
\left(e^{\bu^\top\bz(t)}\right)^\theta
\left(e^{\bv^\top\bz(t)}\right)^{1-\theta}
\right]\\
&\le
\E\left[e^{\bu^\top\bz(t)}\right]^\theta
\E\left[e^{\bv^\top\bz(t)}\right]^{1-\theta}\\
&=M(t,\bu)^\theta M(t,\bv)^{1-\theta}<\infty.
\end{aligned}
\]
Therefore, $\mathcal M_t$ is convex, and hence $\operatorname{int}\mathcal M_t$ is convex.

It is standard that $K(t,\cdot)$ is twice continuously differentiable on $\operatorname{int}\mathcal M_t$ and that
\[
\nabla^2K(t,\bw)=\operatorname{Cov}_{\bw}\left(\bz(t)\right), 
\qquad \bw \in \operatorname{int}\mathcal M_t,
\]
where the exponential tilting law is defined by
\[
P_{\bw}\left(\bz(t)=\mathbf x\right)
=\frac{e^{\bw^\top\mathbf x}}{M(t,\bw)}
P\left(\bz(t)=\mathbf x\right)
\]
and has the same support as the original law.
By \lemmaref{lemma:support_saddle}, for every $\ba\ne\mathbf0$ one may choose $i$ with $a_i\ne0$, so $\ba^\top\bz(t)$ takes the distinct values $0$ and $a_i$ with positive tilted probability. 
Hence,
\[
\ba^\top\nabla^2K(t,\bw)\ba
=\operatorname{Var}_{\bw}\left(\ba^\top\bz(t)\right)>0,
\]
which shows that $\nabla^2K(t,\bw)$ is positive definite for all $\bw\in\operatorname{int}\mathcal M_t$. 
In particular, $\Sigma(\bu)=\nabla^2K(t,\bu)$ is positive definite for all $\bu\in\mathcal U_t$.

For distinct $\bu,\bv\in\mathcal U_t$, their segment lies in $\operatorname{int}\mathcal M_t$, and therefore
\[
\begin{aligned}
(\bu-\bv)^\top\left(\mathbf m(\bu)-\mathbf m(\bv)\right)
&=\int_0^1(\bu-\bv)^\top
\nabla^2K\left(t,\bv+s(\bu-\bv)\right)
(\bu-\bv)\rmd s>0.
\end{aligned}
\]
Thus, $\mathbf m$ is strictly monotone on $\mathcal U_t$, and hence injective.
\end{proof}

\begin{lemma}
\label{lemma:K_boundary_divergence}
Fix $t>0$ and suppose that \asref{assum:all-active} holds. If
$\bu_*\in\partial\mathcal U_t$, then
\[
K(t,\bu)\longrightarrow\infty
\qquad\text{as }\bu\to\bu_*\text{ within }\mathcal U_t.
\]
\end{lemma}

\begin{proof}
Let $\bg_*:=\bg_{\bu_*}$ and $\gamma_{i,*}:=\gamma_{i,\bu_*}$ for $i\in V$.
By \lemmaref{lemma:admissible_mgf_interior}, $\mathcal U_t$ is open, so $\bu_*\notin\mathcal U_t$ and the right maximal time $T$ of $\bg_*$ satisfies $T\le t$. 
By Perko \cite[Theorems~3--4 in Section~2.4]{perko2001differential}, there exists a sequence $s_n\uparrow T$ such that
\[
\left\lVert\bg_*(s_n)\right\rVert>n,
\]
and
\[
\bg_{\bu}\longrightarrow\bg_*
\quad\text{uniformly on every }[0,S],\quad S<T,
\]
as $\bu\to\bu_*$.
By \lemmaref{lemma:positive_characteristics}, all these characteristics are positive before their maximal times.

Set $H:=1+\sum_i\gamma_{i,*}$. Positivity and the quadratic form of
$\mathbf X$ give, for some $C>0$,
\[
\dot H(s)\le CH(s)^2.
\]
Consequently, for $s<s_n<T$,
\[
\frac1{H(s)}
\le\frac1{H(s_n)}+C(s_n-s).
\]
Letting $n\to\infty$ yields
\[
H(s)\ge\frac1{C(T-s)},
\qquad 0\le s<T.
\]
Since $T<\infty$, some $i\in V$ therefore satisfies
\[
\int_0^T\gamma_{i,*}(s)\rmd s=\infty.
\]
In particular, $\gamma_{i,*}$ is unbounded. For every $S<T$,
\[
\liminf_{\substack{\bu\to\bu_*\\ \bu\in\mathcal U_t}}
\int_0^t\gamma_{i,\bu}(s)\rmd s
\ge
\int_0^S\gamma_{i,*}(s)\rmd s.
\]
Letting $S\uparrow T$ gives
\[
\int_0^t\gamma_{i,\bu}(s)\rmd s\longrightarrow\infty
\qquad\text{as }\bu\to\bu_*\text{ within }\mathcal U_t.
\]

Next, we show that the terminal value $\gamma_{i,\bu}(t)$ also diverges.
By \eref{eq:riccati} and positivity of all coordinates, the previously selected coordinate satisfies
\[
\dot\gamma_{i,\bu}(s) \ge p_i(\gamma_{i,\bu}(s)),
\qquad
p_i(x):=
\alpha_i x^2
-\left(\alpha_i+\beta_i+\sum_{k:(i,k)\in E}\mu_{ik}\right)x
+\beta_i.
\]
Choose $L_0>0$ such that
\[
p_i(L)>0
\qquad\text{for all }L\ge L_0.
\]
For $L\ge L_0$, choose $s_L<T$ with $\gamma_{i,*}(s_L)>L$.
By continuous dependence \cite[Theorem~4 in Section~2.4]{perko2001differential}, we have $\gamma_{i,\bu}(s_L)>L$ near $\bu_*$. A first subsequent crossing of the level $L$ would satisfy
\[
\dot\gamma_{i,\bu}(\tau)\le0,
\qquad
\dot\gamma_{i,\bu}(\tau)\ge p_i(L)>0,
\]
a contradiction. Hence, $\gamma_{i,\bu}(t)>L$ near $\bu_*$, and
\[
\gamma_{i,\bu}(t)\longrightarrow\infty
\qquad\text{as }\bu\to\bu_*\text{ within }\mathcal U_t.
\]

Finally, \thmref{thm:characteristic_mgf} and
\lemmaref{lemma:positive_characteristics} give, for $\bu\in\mathcal U_t$,
\[
K(t,\bu)=\bzs_0^\top\log\bg_{\bu}(t)+w_{\bu}(t).
\]
For
\[
d_j:=\alpha_j+\beta_j+\sum_{k:(j,k)\in E}\mu_{jk},
\]
\eref{eq:riccati} gives
\[
\dot\gamma_{j,\bu}(s)\ge\beta_j-d_j\gamma_{j,\bu}(s),
\qquad
\gamma_{j,\bu}(t)\ge\frac{\beta_j}{d_j}(1-e^{-d_jt})>0.
\]
Thus, all terminal logarithms are uniformly bounded below, while
$w_{\bu}(t)\ge-t\sum_j\nu_j$. If $\nu_i>0$, then
\[
w_{\bu}(t)
\ge
\nu_i\int_0^t\gamma_{i,\bu}(s)\rmd s
-t\sum_j\nu_j
\longrightarrow\infty.
\]
If $\nu_i=0$, then $z_{0,i}>0$ by \asref{assum:all-active}, and the terminal
divergence gives $z_{0,i}\log\gamma_{i,\bu}(t)\to\infty$. Thus, in either case,
$K(t,\bu)\to\infty$.
\end{proof}

\begin{theorem}
\label{thm:saddle_point_existence}
Fix $t>0$. Under \asref{assum:all-active},
\[
\mathbf m(\mathcal U_t)=(0,\infty)^d.
\]
\end{theorem}

\begin{proof}
Assumption~\ref{assum:all-active} implies \asref{assum:nondegen}. By \lemmaref{lemma:support_saddle},
\[
P\left(\bz(t)=\be_i\right)>0,
\qquad i\in V.
\]
By \lemmaref{lemma:admissible_mgf_interior}, $\mathcal U_t\subseteq\operatorname{int}\mathcal M_t$. By differentiability of $K(t,\cdot)$ on $\operatorname{int}\mathcal M_t$, for $\bu\in\mathcal U_t$,
\begin{equation}
\label{eq:tilted_mean}
m_i(\bu)
=
\frac{\E\left[Z_i(t)e^{\bu^\top\bz(t)}\right]}
{M(t,\bu)}.
\end{equation}
Moreover,
\[
\E\left[Z_i(t)e^{\bu^\top\bz(t)}\right]
\ge
e^{u_i}P\left(\bz(t)=\be_i\right)>0,
\]
while $M(t,\bu)>0$. 
Hence $m_i(\bu)>0$ for every $i\in V$, and therefore
\[
\mathbf m(\mathcal U_t)\subseteq(0,\infty)^d.
\]

To prove $(0,\infty)^d\subseteq\mathbf m(\mathcal U_t)$, fix $\bn\in(0,\infty)^d$ and define
\[
F_{\bn}(\bu):=K(t,\bu)-\bn^\top\bu,
\qquad \bu\in\mathcal U_t.
\]
Choose an integer $M>\sum_i n_i$ and set
\[
c_M:=\min\left\{\log P(\bz(t)=\mathbf0),
\log P(\bz(t)=M\be_i):i\in V\right\}.
\]
By \lemmaref{lemma:support_saddle}, $c_M$ is finite. Writing
\[
u^+:=\max\{0,u_1,\ldots,u_d\},
\qquad u_i^-:=\max\{-u_i,0\},
\]
we first obtain a lower bound for $K$. If $u^+=0$, the contribution from $\bz(t)=\mathbf0$ gives $K(t,\bu)\ge c_M$. 
If $u^+>0$, choose $j$ such that $u_j=u^+$. 
The contribution from $\bz(t)=M\be_j$ then gives
\[
K(t,\bu)\ge c_M+Mu^+.
\]
Thus, the latter inequality holds in both cases. Moreover,
\[
-\bn^\top\bu
\ge \sum_i n_i u_i^--\left(\sum_i n_i\right)u^+.
\]
Consequently,
\begin{equation}
\label{eq:coercive}
F_{\bn}(\bu)
\ge c_M+\left(M-\sum_i n_i\right)u^+
+\sum_i n_i u_i^-
\longrightarrow\infty
\qquad\text{as }\|\bu\|\to\infty.
\end{equation}

Extend $F_{\bn}$ to $\mathbb R^d$ by setting
\[
\widetilde F_{\bn}(\bu):=
\begin{cases}
F_{\bn}(\bu), & \bu\in\mathcal U_t,\\
+\infty, & \bu\notin\mathcal U_t.
\end{cases}
\]
The constant characteristic gives $\mathbf0\in\mathcal U_t$ and $\widetilde F_{\bn}(\mathbf0)=0$, so $\widetilde F_{\bn}$ is proper.
It is continuous on $\mathcal U_t$. At every $\bu_*\in\partial\mathcal U_t$, by \lemmaref{lemma:K_boundary_divergence} and the continuity of $\bn^\top\bu$,
\[
F_{\bn}(\bu)\longrightarrow\infty
\qquad\text{as }\bu\to\bu_*\text{ within }\mathcal U_t.
\]
Outside $\overline{\mathcal U_t}$, the function $\widetilde F_{\bn}$ is locally equal to $+\infty$. 
Hence, $\widetilde F_{\bn}$ is lower semicontinuous, and therefore closed. 
By \eref{eq:coercive}, $\widetilde F_{\bn}$ is coercive, and therefore attains its minimum on $\mathbb R^d$ by Proposition~3.2.1 of \citet{bertsekas2009convex}.
Since the minimum is finite and $\widetilde F_{\bn}=+\infty$ outside $\mathcal U_t$, every minimizer belongs to $\mathcal U_t$. 
Let $\bu_*$ be one such minimizer. 
Since $\mathcal U_t$ is open, the first-order condition gives
\[
\mathbf m(\bu_*)-\bn
=\nabla F_{\bn}(\bu_*)=\mathbf0.
\]
Thus, $(0,\infty)^d\subseteq\mathbf m(\mathcal U_t)$, completing the proof.
\end{proof}

\begin{proposition}
\label{prop:admissibility-caps}
Fix $t>0$ and let $\bu\in\mathcal U_t$. Then, for each $i\in V$,
\[
y_i(\sigma)>0,
\qquad 0\le\sigma\le t,
\]
and
\[
u_i<\log q_{i,\max}(t;\bu),
\qquad
q_{i,\max}(t;\bu)
:=\frac{\varphi_i(t)}{\alpha_i\varrho_i(t)}>0.
\]
\end{proposition}

\begin{proof}
Solving \eref{eq:riccati_transformation} gives
\[
y_i(\sigma)=
\exp\left(-\alpha_i\int_0^\sigma\gamma_{i,\bu}(r)\rmd r\right)
>0,
\qquad 0\le\sigma\le t.
\]
Both $y_i$ and $\varrho_i$ solve \eref{eq:linear}. Their Wronskian
\[
W_i:=y_i\varrho_i'-y_i'\varrho_i
\]
satisfies $W_i'=b_iW_i$ and $W_i(0)=1$, so
$W_i(\sigma)=\exp\{\int_0^\sigma b_i(r)\rmd r\}>0$. Hence,
\[
\left(\frac{\varrho_i}{y_i}\right)'
=\frac{W_i}{y_i^2}>0,
\qquad
\frac{\varrho_i(0)}{y_i(0)}=0,
\]
so $\varrho_i(\sigma)>0$ for $0<\sigma\le t$. At $\sigma=t$,
\eref{eq:linear_solution} and $y_i(t)>0$ give
\[
e^{u_i}<\frac{\varphi_i(t)}{\alpha_i\varrho_i(t)}
=q_{i,\max}(t;\bu).
\]
The right-hand side is positive, and taking logarithms proves the claim.
\end{proof}

\begin{proposition}
\label{prop:univariate-inheritance}
Fix $t>0$ and $j\in V$. Let the cumulant generating function and the admissible set restricted to the $j$-coordinate be 
\[
    k_j(u):= K \left(t,u \be_j\right),\qquad 
\mathcal U_{t,j}:=\{u\in\R:\ u \be_j\in\mathcal U_t\}.
\]
\begin{enumerate}[label=(\roman*)]

    \item \textbf{Existence and uniqueness:} Under \asref{assum:nondegen},
    \[
        k'_j:\mathcal U_{t,j}\longrightarrow(0,\infty)
    \]
    is injective. Under \asref{assum:all-active},
    $k'_j(\mathcal U_{t,j})=(0,\infty)$.

    \item \textbf{Admissible bounds:} Every $u\in\mathcal U_{t,j}$ satisfies
    \[
        u<\log q_{j,\max}(t;u\be_j)
        \qquad\text{and}\qquad
        0<\log q_{i,\max}(t;u\be_j)\quad\text{for all }i\neq j.
    \]
    
\end{enumerate}
\end{proposition}

\begin{proof}
For part~(i), admissibility and differentiation along the $j$-axis yield
\[
k'_j(u)
=\frac{\E\left[Z_j(t)e^{uZ_j(t)}\right]}
{M(t,u\be_j)}>0,
\qquad u\in\mathcal U_{t,j},
\]
where strict positivity follows from
\lemmaref{lemma:support_saddle}. The strict monotonicity established in the
proof of \propref{prop:saddle_uniqueness} specializes on the $j$-axis to
\[
(u-v)\bigl(k'_j(u)-k'_j(v)\bigr)>0,
\qquad u\neq v,\quad u,v\in\mathcal U_{t,j},
\]
so $k'_j$ is injective. Now suppose that \asref{assum:all-active} holds. If
$u_*\in\partial\mathcal U_{t,j}$ is finite, then
$u_*\be_j\in\partial\mathcal U_t$. For any $r>0$, the boundary-divergence and
coercivity argument in the proof of \thmref{thm:saddle_point_existence},
applied to $k_j(u)-ru$, ensures that this objective attains its minimum at an
interior point $u_r\in\mathcal U_{t,j}$. The first-order condition is
\[
0=k'_j(u_r)-r.
\]
Thus, $k'_j(\mathcal U_{t,j})=(0,\infty)$.

For part~(ii), apply \propref{prop:admissibility-caps} with $\bu=u\be_j$:
\[
(u\be_j)_i<\log q_{i,\max}(t;u\be_j),
\qquad i\in V.
\]
The left-hand side equals $u$ for $i=j$ and $0$ for $i\neq j$, which proves
both bounds.
\end{proof}

\subsection{A counterexample when not all types are active}
\label{app:saddle_counterexample}

Consider the chain $1\to2\to3$ at $t=1$, with $\bzs_0=\mathbf0$, $\bm\nu=(1,0,0)$, $\alpha_i=\beta_i=1$, $\mu_{23}=1/4$, and $\mu_{12}=0.09<(4e)^{-1}$. 
Here, $\mathcal A=\{1\}\ne V$, so \asref{assum:nondegen} holds, while \asref{assum:all-active} fails. 
Let $\bu_\theta=(0,0,\log\theta)$ for $1\le\theta<2$. 
Then, $\bu_\theta\in\mathcal U_1$, every $\bu\in\mathcal U_1$ satisfies $u_3<\log2$, and $\mathbf m(\bu_\theta)$ converges as $\theta\uparrow2$ to a finite positive vector $\mathbf m_*$.
By the strict monotonicity in \propref{prop:saddle_uniqueness}, $\mathbf m_*+c\be_3\notin\mathbf m(\mathcal U_1)$ for every $c>0$, so \asref{assum:nondegen} alone does not guarantee a joint saddle for every positive target.
\fref{fig:saddle_counterexample} illustrates that the tilted mean $m_3(\bu_\theta)$ approaches $(\mathbf m_*)_3\approx0.04$ as $\theta\uparrow2$, so $n_3=1$ has no marginal saddle even though $P(Z_3(1)=1)>0$.

\begin{figure}[H]
    \centering
    \includegraphics[width=0.817\textwidth]{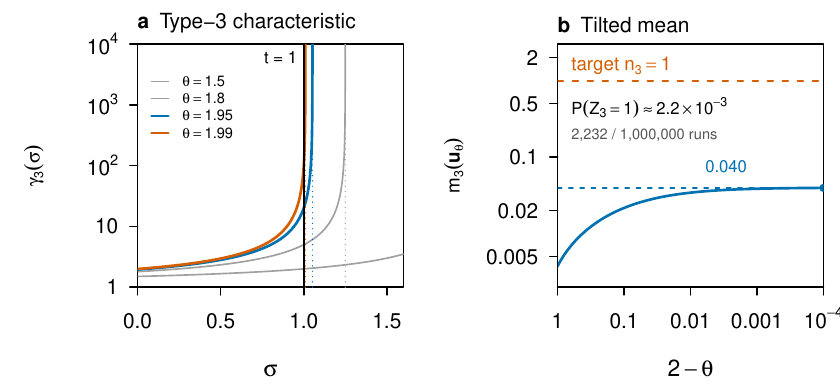}
    \caption{Three-type counterexample.
    (\textbf{a}) The computed type-$3$ characteristic $\gamma_3$: its pole at
    $s=1/(\theta-1)$ (dotted) decreases to $t=1$ (solid) as $\theta\uparrow2$.
    (\textbf{b}) The tilted mean $m_3(\bu_\theta)$ approaches
    $(\mathbf m_*)_3\approx0.04$. The horizontal line marks the marginal
    target $n_3=1$, observed in $2{,}232$ of $10^6$ exact-SSA simulations
    ($95\%$ Wilson interval $[2.14,2.33]\times10^{-3}$)}
    \label{fig:saddle_counterexample}
\end{figure}

\subsection{Univariate gradients and Hessians}
\label{app:saddle_univariate}

With
\[
\bg^{(j)}_u(\sigma):=\bg_{u\be_j}(\sigma),\qquad
\gamma^{(j)}_{i,u}(\sigma):=\gamma_{i,u\be_j}(\sigma),\qquad
w^{(j)}_u(\sigma):=w_{u\be_j}(\sigma),
\]
let
\begin{equation}
\mathbf s(\sigma):=\frac{\partial \bg^{(j)}_u(\sigma)}{\partial u},\qquad
\boldsymbol\tau(\sigma):=\frac{\partial^2 \bg^{(j)}_u(\sigma)}{\partial u^2},\qquad
\xi(\sigma):=\frac{\partial w^{(j)}_u(\sigma)}{\partial u},\qquad
\Xi(\sigma):=\frac{\partial^2 w^{(j)}_u(\sigma)}{\partial u^2},
\end{equation}
which satisfy
\begin{align}
\dot{\mathbf s}&=J(\bg^{(j)}_u) \mathbf s, & \mathbf s(0)&=e^u \be_j, \label{eq:1d_s}\\
\dot{\xi}&=\bm \nu^\top \mathbf s, & \xi(0)&=0, \label{eq:1d_xi}\\
\dot{\boldsymbol\tau}&=J(\bg^{(j)}_u) \boldsymbol\tau+\mathbf h(\bg^{(j)}_u,\mathbf s),
& \boldsymbol\tau(0)&=e^u \be_j, \label{eq:1d_tau}\\
\dot{\Xi}&=\bm \nu^\top \boldsymbol\tau, & \Xi(0)&=0, \label{eq:1d_Xi}
\end{align}
with $J$ as in \eref{eq:jacobian} and
\begin{equation}
\label{eq:1d_h}
h_a(\bg^{(j)}_u,\mathbf s)=\mathbf s^\top H_a(\bg^{(j)}_u) \mathbf s
=2\alpha_a s_a^2 + 2 \sum_{k:(a,k)\in E}\mu_{ak} s_a s_k.
\end{equation}
At time $t$,
\begin{align}
k_j'(u)&=\sum_{a\in V}\frac{z_{0,a}}{\gamma^{(j)}_{a,u}(t)} s_a + \xi, \label{eq:1d_kp}\\
k_j''(u)&=\sum_{a\in V} z_{0,a} \left( \frac{\tau_a}{\gamma^{(j)}_{a,u}(t)} - \frac{s_a^2}{\gamma^{(j)}_{a,u}(t)^2} \right)+\Xi. \label{eq:1d_kpp}
\end{align}
\subsection{Algorithms}
\label{app:algorithms}

\begin{algorithm}[H]
  \caption{\textsc{SaddlePointMTBP}$(t,\bn,\theta,\bzs_0,\text{tol},s_{\min},\epsilon,c_1,\text{max\_iter})$}
  \label{algo:saddlepoint_mtbp}

  \SetNlSkip{0.1em}
  \SetAlgoLined
  \DontPrintSemicolon
  \SetNoFillComment
  \footnotesize

  $\bu \leftarrow \mathbf 0$ \;

  $(K,\mathbf m,\Sigma,\bq_{\max},y_{\min},\_)
  \leftarrow \textsc{FF\_Eval}(t,\bu,\theta,\bzs_0)$ \;

  \For{$k = 1, \dots, \text{max\_iter}$}{

    $\mathbf r \leftarrow \bn - \mathbf m$ \;

    \If{$\|\mathbf r\|_2 \le \text{tol}\left(1+\|\bn\|_2\right)$}{
        $\log p^\star \leftarrow K-\bn^\top\bu-\tfrac12\log\det(2\pi\Sigma)$ \;
        \textbf{Return} $\exp(\log p^\star)$
    }

    Solve $\Sigma\delta=\mathbf r$ for the Newton direction \;

    Compute $\tilde{s}$ from \eref{eq:s} and set
    $s \leftarrow \min\left\{1,\epsilon\tilde{s}\right\}$,
    with $\tilde{s}=+\infty$ if $\delta_i\le0$ for all $i\in V$ \;

    \While{$s > s_{\min}$}{
        $\bu_{try}\leftarrow\bu+s\delta$ \;
        $(K',\mathbf m',\Sigma',\bq'_{\max},y'_{\min},\text{adm})
        \leftarrow\textsc{FF\_Eval}(t,\bu_{try},\theta,\bzs_0)$ \;
        \lIf{$\text{adm}=1$ and
        $F_{\bn}(\bu_{try})\le F_{\bn}(\bu)-c_1s\,\mathbf r^\top\delta$}{\textbf{break}}
        $s \leftarrow s/2$ \;
    }

    \lIf{$s \le s_{\min}$}{\textbf{Return} boundary diagnostic or numerical failure}

    $\bu \leftarrow \bu_{try}$ \;

    $(K,\mathbf m,\Sigma,\bq_{\max},y_{\min}) \leftarrow(K',\mathbf m',\Sigma',\bq'_{\max},y'_{\min})$ \;
  }

  \textbf{Return} failure
\end{algorithm}

\begin{algorithm}[H]
  \caption{\textsc{FF\_Eval}$(t,\bu,\theta,\bzs_0)$}
  \label{algo:ff_eval}
  
  \SetNlSkip{0.1em}
  \SetAlgoLined
  \DontPrintSemicolon
  \SetNoFillComment
  \footnotesize

  Initialize $\bq=e^{\bu}$ and
  $(\bg,w,S,\bm\xi,T,\Xi,\{Q_i\}_{i\in V})$ using the initial conditions in
  \eref{eq:w}, \eref{eq:char}, \eref{eq:Svar}--\eref{eq:Xivar}, and
  \eref{eq:admissible_ODE} \;

  $\sigma \leftarrow 0$, $y_{\min}\leftarrow1$ \;
  \While{$\sigma < t$}{
    Advance the coupled system and $\sigma$ one step using Heun's method \;

    \If{the numerical state is non-finite}{
        \textbf{Return} $(\_,\_,\_,\_,y_{\min},0)$
    }

    Update $y_{\min}$ using
    $y_i=(Q_i)_{11}-\alpha_iq_i(Q_i)_{12}$, $i\in V$ \;

    \If{$y_{\min}\le0$}{
        \textbf{Return} $(\_,\_,\_,\_,y_{\min},0)$
    }
  }

  Compute $K$, $\mathbf m$, $\Sigma$, and $\bq_{\max}$ from
  \eref{eq:K_v2}, \eref{eq:mean}, \eref{eq:Sigma}, and
  \eref{eq:coordinate_caps} \;

  \textbf{Return} $(K,\mathbf m,\Sigma,\bq_{\max},y_{\min},1)$
\end{algorithm}

\section{Large-time small-mutation-rate approximations}
\label{app:ltsm}

In this section, we provide the technical results and mathematical proofs for the large-time and small-mutation-rate approximations. We begin with the large-time behavior of a linear system along a path in \sref{app:convergence}. The associated marginal Laplace transform approximation is then derived in \sref{app:lt_path_decomp}. The bounds needed in this derivation are proved in \sref{app:bounded_mgf}, and the closed-form formulas are calculated in \sref{app:lt_exact}. The resulting approximate Laplace recursions are given in \sref{app:lt_closed} and extended to general DAGs in \sref{app:lt_general}. Finally, \sref{app:ILT-rectangular} derives the inverse Laplace transform formulas for marginal and joint probabilities.

\paragraph{Fixed-path notation}
In Sections~\ref{app:convergence}--\ref{app:lt_closed}, we fix a simple directed path $\pi=(v_0\to v_1\to\cdots\to v_L)$ of length $L$ in $\mathcal G$ and consider the associated linear system defined in \sref{sec:LT_linear}. We write this system as $\bz(t)=(Z_0(t),\dots,Z_L(t))$. For $0\le\ell\le L$ and $a\in\{0,1\}$, we abbreviate
\[
\begin{aligned}
Z_\ell(t)&:=Z_{v_\ell}(t),
& Z_\ell^a(t)&:=Z_{v_\ell}^a(t),
& \tilde Z_\ell&:=\tilde Z_{v_\ell},\\
(\nu_\ell,\alpha_\ell,\beta_\ell,\lambda_\ell)
&:=(\nu_{v_\ell},\alpha_{v_\ell},\beta_{v_\ell},\lambda_{v_\ell}),
& (\delta_\ell,r_\ell,w_\ell(t))
&:=(\delta^\pi_{v_\ell},r^\pi_{v_\ell},w^\pi_{v_\ell}(t)).
\end{aligned}
\]
For $1\le\ell\le L$, we also write $\mu_\ell:=\mu_{v_{\ell-1}v_\ell}$. Thus, population and parameter subscripts identify positions along the fixed path, while $v_\ell$ denotes the corresponding graph vertex. We retain vertex labels in graph-topological notation and exported edge maps.

We call the system a $Z^a$-rooted path of length $L$ if $Z_0(t)=Z_0^a(t)$ for all $t\ge0$. We consider $Z^0$- and $Z^1$-rooted paths separately. 
For a $Z^0$-rooted path with $\nu_0=0$, all coordinates are identically zero, their scaled limits equal zero, and their Laplace transforms equal one.
Thus, we assume $\nu_0>0$ in the subsequent $Z^0$-rooted arguments.
The results then extend to a root given by a sum of independent $Z_0^0$ and $Z_0^1$ processes.
By definition, the corresponding descendant systems are independent and additive, so their Laplace transforms multiply.
In \sref{app:lt_general}, we return to vertex-indexed notation.

\subsection{Convergence along a path}
\label{app:convergence}

We show in \propref{prop:Z1_path_limit} and \propref{prop:Z0_path_limit} that for all $0 \le \ell \le L$, there exists a nonnegative random variable $\tilde Z_\ell$ such that
\begin{equation}
\label{eq:large_time_limit}
    w_\ell(t)Z_\ell(t)\overset{d}{\longrightarrow}\tilde Z_\ell \qquad \text{as } t\to\infty.
\end{equation}

The results are organized as follows.
\lemmaref{lemma:large_time} gives a general convergence result for a one-type branching process seeded by a non-homogeneous Poisson process.
It extends Lemma~1 of \citet{nicholson2023sequential} by allowing $x=0$ and, when $\lambda>0$, $\delta\le0$.
This lemma is used to prove \eref{eq:large_time_limit} for the $Z^1$-rooted path in \propref{prop:Z1_path_limit}.

The $Z^0$-rooted case requires separate arguments before and after the first supercritical type.
\lemmaref{lemma:1type} gives the large-time limiting distributions for a $Z^0$ process, recovering Lemma~1 of \citet{Luo2025Bayesian}.
\lemmaref{lemma:first_supercritical_type} shows that the first supercritical type along a $Z^0$-rooted path has a finite nonnegative almost-sure limit after exponential scaling.
\lemmaref{lemma:subcritical_path_Z0} proves a stationary distribution along a strictly subcritical path.
\lemmaref{lemma:non_supercritical_path} uses a scalar Riccati equation to establish the endpoint limit for any non-supercritical path containing at least one critical type.
Finally, \propref{prop:Z0_path_limit} combines these results with \lemmaref{lemma:large_time} to cover every $Z^0$-rooted path.

\begin{lemma}
\label{lemma:large_time}
Let $\{Z^{1,m}\}_{m\in\mathbb N}$ be i.i.d.\ copies of the one-type branching process started from a single cell and without immigration, with net growth rate $\lambda=\alpha-\beta$. Let $\{T_m\}_{m\in\mathbb N}$ be the points of a Poisson process on $[0,\infty)$ with nonnegative intensity $f(\cdot)$, and set
\[
Z(t) = \sum_{m:T_m\le t} Z^{1,m}(t-T_m).
\]
Assume that $f$ is nonnegative and c\`adl\`ag, and that
\[
\lim_{t\to\infty} t^{-(r-1)} e^{-\delta t} f(t) = x,
\]
where $x \ge 0$, $r \ge 1$, and either (i) $\delta>0$, or (ii) $\delta\le 0$ and $\lambda>0$. 
Then, as $t\to\infty$,
\[
\begin{dcases}
t^{-(r-1)} e^{-\delta t} Z(t) \longrightarrow \dfrac{x}{\delta-\lambda}, & \text{if } \delta>\lambda,\\[1ex]
t^{-r} e^{-\delta t} Z(t) \longrightarrow \dfrac{x}{r}, & \text{if } \delta=\lambda,\\[1ex]
e^{-\lambda t} Z(t) \longrightarrow \tilde Z, & \text{if } \delta<\lambda,
\end{dcases}
\quad \text{almost surely,}
\]
where $\tilde Z$ is a nonnegative random variable with
\[
\E[\tilde Z] = \int_0^\infty e^{-\lambda s} f(s)\rmd s < \infty.
\]
\end{lemma}

\begin{proof}
Let $M(t):=e^{-\lambda t}Z(t)-\int_0^t e^{-\lambda s}f(s)\rmd s$ be the centered martingale in the proof of Lemma~1 of \citet{nicholson2023sequential}.
That proof extends from $x>0$ to $x\ge0$, and we do not require strict positivity of $\tilde Z$.

For case \textup{(i)} with $\lambda\ge0$, we apply the lemma with $r$ replaced by $r-1$.
If $\delta>0$ and $\lambda<0$, only $\delta>\lambda$ can occur.
For the same martingale and any fixed $m$, the second moment bound of $Z^1$ implies
\[
\E[Z^{1,m}(a)^2]\le C e^{\lambda a}
\qquad \implies \qquad
\E[M(t)^2]\le C'(1+t)^{r-1}e^{(\delta-2\lambda)t},
\]
for some constants $C,C'>0$.
Their Doob--Borel--Cantelli argument then proves the first limit.

For case \textup{(ii)}, $\delta\le0<\lambda$, and hence $\delta<\lambda$.
The martingale-convergence part of their $\delta<\lambda$ proof applies because $M$ is bounded in $L^2$ and $\int_0^\infty e^{-\lambda s}f(s)\rmd s<\infty$.
Thus, $M$ converges almost surely and in $L^2$. Rearranging its definition proves the third limit and the stated expectation.
The limit is nonnegative because $e^{-\lambda t}Z(t)\ge0$.
\end{proof}

We parameterize Gamma distributions by shape and rate, and negative-binomial distributions by shape and success probability.

\begin{lemma}
\label{lemma:1type}
Let $\left(X(t)\right)_{t\ge0}$ be a one-type branching process with immigration rate $\nu>0$, birth rate $\alpha>0$, death rate $\beta>0$, $X(0)=0$, $\lambda=\alpha-\beta$, and $\rho=\nu/\alpha$. Then, as $t\to\infty$,
\begin{enumerate}[label=(\roman*)]
\item if $\lambda<0$, then
\[
X(t)\overset{d}{\longrightarrow}Y,\qquad
Y\sim\mathrm{Negative\mbox{-}Binomial}\left(\rho,-\frac{\lambda}{\beta}\right).
\]
\item if $\lambda=0$, then
\[
t^{-1}X(t)\overset{d}{\longrightarrow}Y,\qquad
Y\sim\mathrm{Gamma}\left(\rho,\frac{1}{\alpha}\right).
\]
\item if $\lambda>0$, then
\[
e^{-\lambda t}X(t)\overset{a.s.}{\longrightarrow}Y,
\qquad
Y\sim\mathrm{Gamma}\left(\rho,\frac{\lambda}{\alpha}\right).
\]
\end{enumerate}
\end{lemma}

\begin{proof}
Let $G(s,t)=\mathbb E\left[s^{X(t)}\right]$. The Kolmogorov forward equation for $X$ is
\[
G_t=\left[\alpha s(s-1)+\beta(1-s)\right]G_s+\nu(s-1)G,
\]
with $G(s,0)=1$. Solving this equation by characteristics yields
\[
G(s,t)=
\begin{dcases}
\left[1-\alpha t(s-1)\right]^{-\rho}, & \lambda=0,\\
\left[\frac{\lambda}{\alpha s\left(1-\exp\left(\lambda t\right)\right)+\alpha\exp\left(\lambda t\right)-\beta}\right]^{\rho}, & \lambda\ne0.
\end{dcases}
\]
Equivalently,
\[
G(s,t)=\left[\frac{p(t)}{1-\left(1-p(t)\right)s}\right]^{\rho},
\]
where
\[
p(t)=
\begin{dcases}
\frac{1}{1+\alpha t}, & \lambda=0,\\
\frac{\lambda}{\alpha\exp\left(\lambda t\right)-\beta}, & \lambda\ne0.
\end{dcases}
\]
Hence, $X(t)\sim\mathrm{Negative\mbox{-}Binomial}\left(\rho,p(t)\right)$.
If $\lambda<0$, then $p(t)\to-\frac{\lambda}{\beta}$. Thus, the probability generating function above converges to that of $\mathrm{Negative\mbox{-}Binomial}\left(\rho,-\lambda/\beta\right)$, and hence
\[
X(t)\overset{d}{\longrightarrow}
\mathrm{Negative\mbox{-}Binomial}\left(\rho,-\frac{\lambda}{\beta}\right).
\]
Next, if $X_p\sim\mathrm{Negative\mbox{-}Binomial}\left(\rho,p\right)$, then for $s\ge0$,
\[
\mathbb E\left[\exp\left(-spX_p\right)\right] = 
\left[\frac{p}{1-\left(1-p\right)\exp\left(-sp\right)}\right]^{\rho}
\longrightarrow
\left(1+s\right)^{-\rho}
\]
as $p\downarrow0$. By the continuity theorem for Laplace transforms,
\begin{equation}
\label{eq:negative_binomial_gamma_limit}
pX_p\overset{d}{\longrightarrow}Z \quad\text{as }p\downarrow0,
\qquad Z\sim\mathrm{Gamma}\left(\rho,1\right).
\end{equation}
For $\lambda=0$, $p(t)\to0$ and $p(t)t\to\frac{1}{\alpha}$, so \eref{eq:negative_binomial_gamma_limit} and Slutsky's theorem imply
\[
\frac{X(t)}{t} = 
\frac{p(t)X(t)}{p(t)t}
\overset{d}{\longrightarrow}
\alpha Z
\sim\mathrm{Gamma}\left(\rho,\frac{1}{\alpha}\right).
\]
Finally, suppose $\lambda>0$. Since $p(t)\to0$ and $p(t)\exp\left(\lambda t\right)\to\frac{\lambda}{\alpha}$, \eref{eq:negative_binomial_gamma_limit} and Slutsky's theorem imply
\[
\exp\left(-\lambda t\right)X(t) = 
\frac{p(t)X(t)}{p(t)\exp\left(\lambda t\right)}
\overset{d}{\longrightarrow}
\frac{\alpha}{\lambda}Z
\sim\mathrm{Gamma}\left(\rho,\frac{\lambda}{\alpha}\right).
\]
To upgrade this convergence to almost-sure convergence, let $\mathcal L$ be the generator of $X$ and set $h(k)=k+\nu/\lambda$. Then
\[
\mathcal L h(k)=\nu+\lambda k=\lambda h(k),
\]
so Dynkin's formula shows that
\[
M_t=\exp\left(-\lambda t\right)\left(X(t)+\frac{\nu}{\lambda}\right)
\]
is a nonnegative martingale, where the required integrability follows from the finite first moments of $X$. Hence, it converges almost surely. Since $\exp\left(-\lambda t\right)\nu/\lambda\to0$, $\exp\left(-\lambda t\right)X(t)$ converges almost surely to the same limit, which we denote by $Y$. The same sequence converges in distribution to $(\alpha/\lambda)Z$, so uniqueness of distributional limits shows that $Y\sim\mathrm{Gamma}\left(\rho,\lambda/\alpha\right)$.
\end{proof}

\begin{lemma}
\label{lemma:first_supercritical_type}
Let $\bz(t)$ be a $Z^0$-rooted path of length $L$. Assume that, for some $1\le\ell\le L$, $v_\ell$ is the first supercritical type on the path:
\[
\lambda_{\ell}:=\alpha_{\ell}-\beta_{\ell}>0,
\qquad
\lambda_{k}\le 0 \text{ for all }k<\ell.
\]
Then, there exists a finite nonnegative random variable $\tilde Z_\ell$ such that
\[
e^{-\lambda_\ell t}Z_\ell(t)\overset{a.s.}{\longrightarrow}\tilde Z_\ell.
\]
\end{lemma}

\begin{proof}
Let $m_{\ell-1}(t)=\E[Z_{\ell-1}(t)]$.
Since all types before $v_\ell$ are non-supercritical and the root is $Z_0^0$, Jensen's inequality and \propref{prop:uniform_mgf} imply that there are constants $C<\infty$ and an integer $d\ge0$ such that
\[
m_{\ell-1}(t)\le C(1+t^d),\qquad t\ge0.
\]
Set $f(u)=\mu_\ell Z_{\ell-1}(u)$. By Tonelli's theorem and the polynomial bound above,
\[
\E\left[\int_0^\infty e^{-\lambda_\ell u}f(u)\rmd u\right]
=\mu_\ell\int_0^\infty e^{-\lambda_\ell u}m_{\ell-1}(u)\rmd u
\le C \mu_\ell\int_0^\infty e^{-\lambda_\ell u}(1+u^d)\rmd u
<\infty,
\]
so $I:=\int_0^\infty e^{-\lambda_\ell u}f(u)\rmd u<\infty$ almost surely. Fix a parent trajectory for which $I<\infty$, and write $\E[\,\cdot\mid f]$ for expectations with this trajectory fixed. Define
\[
Y(t)=e^{-\lambda_\ell t}Z_\ell(t),
\qquad
A(t)=\int_0^t e^{-\lambda_\ell u}f(u)\rmd u.
\]
Given this trajectory, the child population jumps up at rate $\alpha_\ell Z_\ell(u)+f(u)$ and down at rate $\beta_\ell Z_\ell(u)$. Therefore, $M(t)=Y(t)-A(t)$ is a martingale with predictable quadratic variation
\[
\langle M\rangle_t
=\int_0^t e^{-2\lambda_\ell u}
\left((\alpha_\ell+\beta_\ell)Z_\ell(u)+f(u)\right)\rmd u.
\]
Taking conditional expectations and using $\E[Z_\ell(u)\mid f]=e^{\lambda_\ell u}A(u)$ yields
\[
\E[\langle M\rangle_\infty\mid f]\le
\frac{\alpha_\ell+\beta_\ell}{\lambda_\ell}I+I<\infty.
\]
Thus, $M(t)$ is $L^2$-bounded for almost every fixed parent trajectory and converges for almost every child realization. Since $A(t)\to I$, $Y(t)=M(t)+A(t)$ also converges to a finite limit, which is nonnegative because $Y(t)\ge0$. Averaging over the parent trajectories proves the unconditional almost-sure convergence.
\end{proof}

\begin{lemma}
\label{lemma:subcritical_path_Z0}
Let $\bz(t)$ be a $Z^0$-rooted path of length $L$, and assume $\lambda_j<0$ for all $0\le j\le L$. For each $0\le\ell\le L$, set $\mathbf Z^{(\ell)}(t)=\left(Z_0(t),\dots,Z_\ell(t)\right)$. Then, $\mathbf Z^{(\ell)}(t)$ is ergodic with a unique stationary distribution $\omega_\ell$, and
\[
Z_\ell(t)\overset{d}{\longrightarrow}\tilde Z_\ell,
\]
where $\tilde Z_\ell$ is a nonnegative random variable with the $v_\ell$-coordinate marginal distribution of $\omega_\ell$.

Moreover, for $0\le j\le\ell$,
\[
\bar z_j:=\E_{\omega_\ell}[Z_j] \in(0,\infty),
\]
and
\begin{equation}
\label{eq:subcritical_coordinate_average}
\frac1t\int_0^t Z_j(u)\rmd u\longrightarrow\bar z_j
\qquad\text{a.s.}
\end{equation}
\end{lemma}

\begin{proof}
Fix $0\le\ell\le L$.
On the event of infinitely many jumps, if $\mathbf Y_{k-1}$ is the state before jump $k$, then $\|\mathbf Y_{k-1}\|_1\le k$ and the total jump rate satisfies
\[
q(\mathbf Y_{k-1})\le\nu_0+ak
\]
for some $a<\infty$.
Thus, the $k$th holding time satisfies $S_k\ge T_k/(\nu_0+ak)$ for independent $T_k\sim\operatorname{Exp}(1)$.
The jump-chain argument used in the proof of \thmref{thm:characteristic_mgf} shows that $\sum_k S_k=\infty$ almost surely, so $\mathbf Z^{(\ell)}$ is nonexplosive.

The chain is also irreducible: deaths can take any state to the origin, while root immigration, births, and successive mutations, followed by deaths of surplus cells, can reach any state from the origin.

Let $V(\mathbf z)=1+\sum_{k=0}^\ell c_k z_k$.
Since every $\lambda_j<0$, choose $c_\ell>0$ and then recursively choose
\[
c_j>\frac{c_{j+1}\mu_{j+1}}{-\lambda_j},
\qquad j=\ell-1,\ldots,0.
\]
The generator of $\mathbf Z^{(\ell)}$ satisfies
\[
\mathcal L V(\mathbf z)
=c_0\nu_0+
\sum_{j=0}^{\ell-1}\left(c_j\lambda_j+c_{j+1}\mu_{j+1}\right)z_j
+ c_\ell\lambda_\ell z_\ell
\le C-\gamma V(\mathbf z)
\]
for some $C, \gamma>0$. Since every $c_k>0$, the sublevel set
\[
K=\left\{\mathbf z:V(\mathbf z)\le\frac{2C}{\gamma}\right\}
\]
is finite and, by irreducibility on the countable discrete state space, petite in the sense of \citet{meyn1993stabilityIII}.
Moreover,
\[
\mathcal L V(\mathbf z)\le-\frac{\gamma}{2}V(\mathbf z)+C\mathds{1}_K(\mathbf z).
\]
By Theorem~4.2 of \citet{meyn1993stabilityIII}, $\mathbf Z^{(\ell)}$ is positive Harris recurrent and $\E_{\omega_\ell}[V]<\infty$ for the invariant distribution $\omega_\ell$.
By Theorems~3.5.2 and 3.6.2 of \citet{norris1998markov}, $\omega_\ell$ is unique and
\[
\mathbf Z^{(\ell)}(t)\overset{d}{\longrightarrow}\omega_\ell.
\]
By the continuous mapping theorem,
\[
Z_\ell(t)\overset{d}{\longrightarrow}\tilde Z_\ell,
\]
with $\tilde Z_\ell$ having the $v_\ell$-coordinate marginal distribution of $\omega_\ell$.

Since every function on the discrete state space is continuous, the chain is a $T$-process \citep[Section~7]{meyn1993stabilityIII}.
Since $V(\mathbf z)=1+\sum_{k=0}^\ell c_kz_k\ge c_jz_j$ and $c_j>0$,
\[
\bar z_j=\E_{\omega_\ell}[Z_j]
\le\frac{1}{c_j}\E_{\omega_\ell}[V]<\infty,
\]
so every coordinate is $\omega_\ell$-integrable.
Irreducibility implies that $\omega_\ell$ assigns positive mass to every state, so $\bar z_j>0$.
Theorems~3.2(ii) and 8.1(i) of \citet{meyn1993stabilityII} now yield \eref{eq:subcritical_coordinate_average}.
\end{proof}

The two- and three-type all-critical cases of the following lemma agree with the corresponding endpoint marginal limits of \citet{barczy2023critical} and \citet{barczy2024critical}, respectively, after matching the parameters of the discrete-time skeleton.

\begin{lemma}
\label{lemma:non_supercritical_path}
Let $\bz(t)$ be a $Z^0$-rooted path of length $L$.
Fix $0\le\ell\le L$.
Assume $\lambda_j\le0$ for all $0\le j\le\ell$ and $\lambda_j=0$ for at least one $0\le j\le\ell$.
Define the index of the first critical type among $v_0,\dots,v_\ell$ by
\[
\kappa=\min\{0\le j\le\ell:\lambda_j=0\}.
\]
Let
\[
\eta_\kappa=
\begin{dcases}
\nu_0, & \kappa=0,\\
\mu_\kappa\bar z_{\kappa-1}, & \kappa>0,
\end{dcases}
\qquad
C_\ell=
\prod_{j=\kappa+1}^{\ell}
\begin{dcases}
\mu_j, & \lambda_j=0,\\
\mu_j/(-\lambda_j), & \lambda_j<0,
\end{dcases}
\]
where $\bar z_{\kappa-1}$ is the expected stationary population of $v_{\kappa-1}$ in \lemmaref{lemma:subcritical_path_Z0}.
Then,
\[
t^{-(r_\ell-1)}Z_\ell(t)\overset{d}{\longrightarrow}Y_\ell,
\]
where $Y_\ell$ is nonnegative and has Laplace transform
\begin{equation}
\label{eq:non_supercritical_endpoint_LT}
\E[e^{-sY_\ell}]
=
\exp\left\{-\eta_\kappa\int_0^1 y_s(x)\rmd x\right\},
\qquad s\ge0,
\end{equation}
with $y_s$ solving the scalar Riccati equation
\begin{equation}
\label{eq:non_supercritical_Riccati}
\begin{dcases}
y_s'(x)=-\alpha_\kappa y_s(x)^2,
\quad y_s(0)=sC_\ell,
& r_\ell=2,\\[1ex]
y_s'(x)=\dfrac{sC_\ell x^{r_\ell-3}}{(r_\ell-3)!}
-\alpha_\kappa y_s(x)^2,
\quad y_s(0)=0,
& r_\ell\ge3.
\end{dcases}
\end{equation}
\end{lemma}

\begin{proof}
For $\kappa\le j\le\ell$, we have by definition (\eref{eq:r})
\[
r_j-1=\#\{\kappa\le a\le j:\lambda_a=0\}.
\]
Define
\[
C_j=\prod_{a=\kappa+1}^{j}
\begin{dcases}
\mu_a, & \lambda_a=0,\\
\mu_a/(-\lambda_a), & \lambda_a<0.
\end{dcases}
\]
For a locally integrable function $x$, define
\[
A_1x(t)=x(t),
\qquad
A_nx(t)=\frac{1}{(n-2)!}\int_0^t(t-u)^{n-2}x(u)\rmd u,
\quad n\ge2.
\]
Write
\[
Z_j(t)=Z_j(0)+\int_0^tB_j(u)\rmd u+M_j(t),
\]
where $B_0=\nu_0+\lambda_0Z_0$ and $B_j=\mu_jZ_{j-1}+\lambda_jZ_j$ for $j\ge1$. 
$M_j$ is a square-integrable martingale with $M_j(0)=0$.
The predictable quadratic-variation density of $M_j$ is
\[
\frac{\rmd\langle M_j\rangle_t}{\rmd t}
=
\begin{dcases}
(\alpha_0+\beta_0)Z_0(t)+\nu_0, & j=0,\\
(\alpha_j+\beta_j)Z_j(t)+\mu_jZ_{j-1}(t), & j\ge1.
\end{dcases}
\]
Since $w_j(t)=t^{-(r_j-1)}$ for $j\ge\kappa$, by \corref{cor:all_moments_bounded} and $r_{j-1} \le r_j$ for $j \ge 1$,
\begin{equation}
\label{eq:martingale_variation_bound}
\E[Z_j(t)]=\mathcal O((1+t)^{r_j-1})
\qquad \implies \qquad
\E\left[\frac{\rmd\langle M_j\rangle_t}{\rmd t}\right]
=\mathcal O((1+t)^{r_j-1}).
\end{equation}

Next, set
\[
R_j(t)=Z_j(t)-C_jA_{r_j-1}Z_\kappa(t).
\]
We prove by induction that as $t\to\infty$,
\begin{equation}
\label{eq:non_supercritical_endpoint_reduction}
\frac{\E[|R_j(t)|]}{t^{r_j-1}}\longrightarrow0,
\qquad \kappa\le j\le\ell.
\end{equation}
At $j=\kappa$, we have $C_\kappa=1$, $r_\kappa=2$, and $A_1Z_\kappa(t)=Z_\kappa(t)$, so $R_\kappa(t)=0$. 

Fix $j>\kappa$ and assume that \eref{eq:non_supercritical_endpoint_reduction} holds for every $\kappa\le i<j$.
If $\lambda_j=0$, then $r_j=r_{j-1}+1\ge3$. Since $C_j=\mu_jC_{j-1}$, we have
\[
\int_0^tA_{r_j-2}Z_\kappa(u)\rmd u=A_{r_j-1}Z_\kappa(t)
\qquad \implies \qquad
R_j(t)=\mu_j\int_0^tR_{j-1}(u)\rmd u+M_j(t).
\]
For any $\epsilon>0$, the inductive hypothesis allows us to choose $T$ such that
\[
\E[|R_{j-1}(u)|] \le \epsilon u^{r_j - 2}, \qquad u \ge T.
\]
For $t\ge T$, by triangle inequality and Tonelli's theorem, we have
\[
\frac1{t^{r_j-1}} \E\left[ \left| \int_0^t R_{j-1}(u) \rmd u \right| \right]
\le \frac1{t^{r_j-1}} \int_0^T \E[|R_{j-1}(u)|] \rmd u
+ \frac1{t^{r_j-1}} \int_T^t \E[|R_{j-1}(u)|] \rmd u.
\]
The first term on the right-hand side tends to zero, while the second is at most $\epsilon/(r_j-1)$.
Since $\epsilon$ is arbitrary, we conclude that
\[
\frac{1}{t^{r_j-1}} \E\left[ \left| \int_0^t R_{j-1}(u) \rmd u \right| \right] \longrightarrow 0.
\]
By the martingale isometry and \eref{eq:martingale_variation_bound},
\[
\|M_j(t)\|_2^2 = \E[\langle M_j\rangle_t] = \int_0^t \E\left[\frac{\rmd\langle M_j\rangle_u}{\rmd u}\right] \rmd u = \mathcal O(t^{r_j}).
\]
By Cauchy--Schwarz inequality, we have
\[
\frac{\E[|M_j(t)|]}{t^{r_j-1}}
\le \frac{\|M_j(t)\|_2}{t^{r_j-1}} = \mathcal O(t^{(2 - r_j)/2}) \longrightarrow 0,
\]
which proves \eref{eq:non_supercritical_endpoint_reduction} at $j$.

If $\lambda_j<0$, set $a_j=-\lambda_j$, $q_j=\mu_j/a_j$, and $D_j=Z_j-q_jZ_{j-1}$.
Since $D_j(0)=0$, variation of constants yields
\[
D_j(t)=-q_j\int_0^te^{-a_j(t-u)}B_{j-1}(u)\rmd u 
+\int_0^te^{-a_j(t-u)}\rmd M_j(u)
-q_j\int_0^te^{-a_j(t-u)}\rmd M_{j-1}(u).
\]
Here $r_j=r_{j-1}$. 
If $\lambda_{j-1}=0$ and $j-1=0$, then $B_{j-1}=\nu_0$ and $r_j=2$.
If $\lambda_{j-1}=0$ and $j-1\ge1$, then $B_{j-1}=\mu_{j-1}Z_{j-2}$, and $r_{j-2}=r_j-1$.
Hence, using \corref{cor:all_moments_bounded} in the second case,
\begin{equation}
\label{eq:critical_predecessor_immigration_bound}
\E[|B_{j-1}(t)|]=\mathcal O((1+t)^{r_j-2}),
\qquad \lambda_{j-1}=0.
\end{equation}
If $\lambda_{j-1}<0$, set $a_{j-1}=-\lambda_{j-1}$ and $q_{j-1}=\mu_{j-1}/a_{j-1}$.
In this case, $r_{j-2}=r_{j-1}=r_j$ and $C_{j-1}=q_{j-1}C_{j-2}$, so
\[
B_{j-1}
=-a_{j-1}\bigl(R_{j-1}-q_{j-1}R_{j-2}\bigr).
\]
The bound in \eref{eq:critical_predecessor_immigration_bound} handles $\lambda_{j-1}=0$, while the inductive hypothesis at $j-1$ and $j-2$ handles $\lambda_{j-1}<0$.
Thus, in either case,
\[
\frac{\E[|B_{j-1}(t)|]}{t^{r_j-1}}\longrightarrow0.
\]
For any $\epsilon>0$, we choose $T$ such that
\[
\E[|B_{j-1}(u)|]\le\epsilon u^{r_j-1},
\qquad u\ge T.
\]
For $t\ge T$, the triangle inequality and Tonelli's theorem yield
\[
\frac1{t^{r_j-1}}\E\left[
\left|\int_0^te^{-a_j(t-u)}B_{j-1}(u)\rmd u\right|
\right]
\le
\frac{e^{-a_jt}}{t^{r_j-1}}
\int_0^T e^{a_ju}\E[|B_{j-1}(u)|]\rmd u
+\frac{\epsilon}{a_j}.
\]
Taking $t\to\infty$ and then letting $\epsilon\downarrow0$ proves that this term converges to zero.
For $k\in\{j-1,j\}$, the martingale isometry and \eref{eq:martingale_variation_bound} imply
\[
\frac1{t^{r_j-1}}\left\|
\int_0^te^{-a_j(t-u)}\rmd M_k(u)
\right\|_2 
=\mathcal O(t^{-(r_j-1)/2})\longrightarrow 0.
\]
Thus,
\[
\frac{\E[|D_j(t)|]}{t^{r_j-1}}\longrightarrow 0.
\]
The identity $R_j=q_jR_{j-1}+D_j$ completes the induction.

The first critical population has immigration intensity
\[
b(u)=
\begin{dcases}
\nu_0, & \kappa=0,\\
\mu_\kappa Z_{\kappa-1}(u), & \kappa>0.
\end{dcases}
\]
For every continuous $g$ on $[0,1]$,
\begin{equation}
\label{eq:first_critical_weighted_average}
\frac1t\int_0^tb(u)g(u/t)\rmd u
\longrightarrow
\eta_\kappa\int_0^1g(x)\rmd x
\qquad\text{a.s.}
\end{equation}
This is immediate when $\kappa=0$.
When $\kappa>0$, by \lemmaref{lemma:subcritical_path_Z0},
\[
\frac1{t}\int_0^tb(u)\rmd u
=\frac{\mu_\kappa}{t}\int_0^tZ_{\kappa-1}(u)\rmd u
\longrightarrow\mu_\kappa\bar z_{\kappa-1}=\eta_\kappa
\qquad\text{a.s.}
\]
Let $h$ be a step function on $[0,1]$ with $h(x)=h_i$ on $(x_{i-1},x_i]$, where $0=x_0<x_1<\dots<x_m=1$. Then
\[
\frac1t\int_0^tb(u)h(u/t)\rmd u
=\sum_{i=1}^m h_i\frac1t\int_{x_{i-1}t}^{x_it}b(u)\rmd u
\longrightarrow\sum_{i=1}^m h_i\eta_\kappa(x_i-x_{i-1})
=\eta_\kappa\int_0^1h(x)\rmd x
\qquad\text{a.s.}
\]
Choosing $h$ to approximate $g$ uniformly on $[0,1]$ gives \eref{eq:first_critical_weighted_average}.

Conditional on $b$, the immigrants initiate independent critical clones. Let $Z^1$ denote one clone. A clone of age $a$ contributes
\[
A_{r_\ell-1}Z^1(a)=
\begin{dcases}
Z^1(a), & r_\ell=2,\\[1ex]
\dfrac{1}{(r_\ell-3)!}\displaystyle\int_0^a
(a-v)^{r_\ell-3}Z^1(v)\rmd v, & r_\ell\ge3,
\end{dcases}
\]
to $A_{r_\ell-1}Z_\kappa(t)$.
For $s \ge 0$, the Poisson construction applied to these independent clone contributions gives
\begin{equation}
\label{eq:non_supercritical_conditional_LT}
\E\left[
\exp\left\{-\frac{sC_\ell}{t^{r_\ell-1}}A_{r_\ell-1}Z_\kappa(t)\right\}
\,\middle|\,b
\right]
=
\exp\left\{-\int_0^tb(u)h_t(t-u)\rmd u\right\},
\end{equation}
where
\[
h_t(a)=1-\E\left[
\exp\left\{
-\frac{sC_\ell}{t^{r_\ell-1}}A_{r_\ell-1}Z^1(a)
\right\}
\right].
\]

By a first-event decomposition of the critical birth--death clone and the branching property that descendant lineages evolve independently according to the same law, $h_t$ satisfies the backward equation
\[
\begin{dcases}
h_t'(a)=-\alpha_\kappa h_t(a)^2,\quad h_t(0)=1-e^{-sC_\ell/t}, & r_\ell=2,\\[1ex]
h_t'(a)=\frac{sC_\ell a^{r_\ell-3}}{(r_\ell-3)!t^{r_\ell-1}}\left(1-h_t(a)\right)-\alpha_\kappa h_t(a)^2,\quad h_t(0)=0, & r_\ell\ge3.
\end{dcases}
\]
Define $y_{s,t}(x)=th_t(tx)$ for $0\le x\le1$.
Then, $y_{s,t}$ solves
\begin{equation}
\label{eq:non_supercritical_Riccati_finite}
\begin{dcases}
y_{s,t}'=-\alpha_\kappa y_{s,t}^2,\quad
y_{s,t}(0)=t(1-e^{-sC_\ell/t}), & r_\ell=2,\\[1ex]
y_{s,t}'=\dfrac{sC_\ell x^{r_\ell-3}}{(r_\ell-3)!}
\left(1-\dfrac{y_{s,t}}{t}\right)-\alpha_\kappa y_{s,t}^2,\quad
y_{s,t}(0)=0, & r_\ell\ge3.
\end{dcases}
\end{equation}
For $r_\ell=2$, $y_{s,t}$ is non-increasing and
\[
0 \le y_{s,t}(x) \le y_{s,t}(0) = t(1-e^{-sC_\ell/t}) \le sC_\ell.
\]
For $r_\ell\ge3$, $0\le h_t\le1$ and the differential equation gives
\[
0\le y_{s,t}(x)\le\frac{sC_\ell x^{r_\ell-2}}{(r_\ell-2)!}\le sC_\ell.
\]
The same bounds hold for $y_s$, so the finite and limiting solutions remain in the common rectangle $0\le x\le1$ and $0\le y\le sC_\ell$. 
Let $F_{s,t}$ and $F_s$ denote the right-hand sides in \eref{eq:non_supercritical_Riccati_finite} and \eref{eq:non_supercritical_Riccati}, respectively.
For $t\ge1$, $0\le x\le1$, and $y,z\in[0,sC_\ell]$,
\[
\max\left\{
|F_{s,t}(x,y)-F_{s,t}(x,z)|,
|F_s(x,y)-F_s(x,z)|
\right\}
\le L_s|y-z|,
\]
where $L_s=(1+2\alpha_\kappa)sC_\ell$ is independent of $t$.
The initial values converge, and $F_{s,t}\to F_s$ uniformly on the common rectangle.
Therefore, by Gronwall's inequality,
\[
\sup_{0\le x\le1}|y_{s,t}(x)-y_s(x)|\longrightarrow0.
\]
Together with \eref{eq:first_critical_weighted_average}, the exponent in \eref{eq:non_supercritical_conditional_LT} converges:
\[
\frac1t\int_0^tb(u)y_{s,t}(1-u/t)\rmd u
\longrightarrow
\eta_\kappa\int_0^1y_s(x)\rmd x
\qquad\text{a.s.}
\]
Dominated convergence then yields \eref{eq:non_supercritical_endpoint_LT} for the unconditional transforms. 
Since $0\le y_s(x)\le sC_\ell$, the limiting transform is continuous at $s=0$. 
By the continuity theorem for Laplace transforms, 
\[
t^{-(r_\ell-1)}C_\ell A_{r_\ell-1}Z_\kappa(t)\overset{d}{\longrightarrow}Y_\ell,
\]
where $Y_\ell$ is nonnegative and has Laplace transform \eref{eq:non_supercritical_endpoint_LT}.
By \eref{eq:non_supercritical_endpoint_reduction} at $j=\ell$ and Slutsky's theorem, we conclude that
\[
t^{-(r_\ell-1)}Z_\ell(t)\overset{d}{\longrightarrow}Y_\ell.
\]
\end{proof}

\begin{corollary}
\label{cor:non_supercritical_gamma_limit}
If $r_\ell=2$, then
\[
t^{-1}Z_\ell(t)\overset{d}{\longrightarrow}C_\ell G_\kappa,
\qquad
G_\kappa\sim
\mathrm{Gamma}\left(\frac{\eta_\kappa}{\alpha_\kappa},
\frac{1}{\alpha_\kappa}\right).
\]
\end{corollary}

\begin{proof}
When $r_\ell=2$, solving \eref{eq:non_supercritical_Riccati} yields
\[
\int_0^1y_s(x)\rmd x
=\frac1{\alpha_\kappa}\log(1+\alpha_\kappa C_\ell s).
\]
Substitution into \eref{eq:non_supercritical_endpoint_LT} identifies the law of $Y_\ell$ as that of $C_\ell G_\kappa$. The conclusion follows from \lemmaref{lemma:non_supercritical_path}.
The Gamma form of $G_\kappa$ coincides with case \textup{(ii)} of \lemmaref{lemma:1type}, with the effective immigration rate $\eta_\kappa$ in place of $\nu$.
\end{proof}

\begin{proposition}
\label{prop:Z1_path_limit}
Let $\bz(t)$ be a $Z^1$-rooted path of length $L$. Then, for each $0\le\ell\le L$, there exists a nonnegative random variable $\tilde Z_\ell$ such that
\[
w_\ell(t)Z_\ell(t)\overset{a.s.}{\longrightarrow}\tilde Z_\ell,
\]
which implies $w_\ell(t)Z_\ell(t)\overset{d}{\longrightarrow}\tilde Z_\ell$.
\end{proposition}

\begin{proof}
We prove by induction on $\ell$. For $\ell=0$, the result is the standard one-type birth-death limit: if $\lambda_0\le0$, the process becomes extinct in finite time almost surely, while if $\lambda_0>0$,
\[
e^{-\lambda_0t}Z_0(t)\overset{a.s.}{\longrightarrow}\tilde Z_0.
\]
Assume the claim holds for $v_{\ell-1}$. Conditional on the parent trajectory, mutations into $v_\ell$ form an inhomogeneous Poisson process with intensity
\[
f(u)=\mu_\ell Z_{\ell-1}(u).
\]
By the inductive hypothesis,
\[
t^{-\left(r_{\ell-1}-1\right)}e^{-\delta_{\ell-1}t}f(t)
\overset{a.s.}{\longrightarrow}
\mu_\ell\tilde Z_{\ell-1}.
\]
If $\delta_{\ell-1}>0$, then applying case \textup{(i)} of \lemmaref{lemma:large_time} conditionally on the parent path gives
\[
w_\ell(t)Z_\ell(t)\overset{a.s.}{\longrightarrow}\tilde Z_\ell.
\]
It remains to consider $\delta_{\ell-1}=0$. If $\lambda_\ell>0$, then case \textup{(ii)} of \lemmaref{lemma:large_time} gives
\[
e^{-\lambda_\ell t}Z_\ell(t)
=w_\ell(t)Z_\ell(t)
\overset{a.s.}{\longrightarrow}\tilde Z_\ell.
\]
If $\lambda_\ell\le0$, then all types $v_0,\dots,v_\ell$ are non-supercritical and the path has no immigration. Therefore, the finite multi-type branching process becomes extinct in finite time almost surely, and
\[
w_\ell(t)Z_\ell(t)\overset{a.s.}{\longrightarrow}0.
\]
This completes the induction, and convergence in distribution follows from almost-sure convergence.
\end{proof}

\begin{proposition}
\label{prop:Z0_path_limit}
Let $\bz(t)$ be a $Z^0$-rooted path of length $L$.
For each $0\le\ell\le L$, there exists a nonnegative random variable $\tilde Z_\ell$ such that
\[
w_\ell(t)Z_\ell(t)\overset{d}{\longrightarrow}\tilde Z_\ell.
\]
If $\delta_\ell>0$, then the convergence is almost sure.
\end{proposition}

\begin{proof}
Fix $0\le\ell\le L$. First suppose $\delta_\ell>0$, and let $v_k$ be the first supercritical type in $\pi_{:\ell}$. 
By \lemmaref{lemma:1type} for $k=0$ and \lemmaref{lemma:first_supercritical_type} for $k>0$,
\[
w_k(t)Z_k(t)
=e^{-\lambda_kt}Z_k(t)
\overset{a.s.}{\longrightarrow}\tilde Z_k.
\]
For $k<j\le\ell$, suppose inductively that $w_{j-1}(t)Z_{j-1}(t)\overset{a.s.}{\longrightarrow} \tilde Z_{j-1}$. 
Conditional on the parent trajectory, mutations into $v_j$ have intensity $f(u)=\mu_jZ_{j-1}(u)$, so
\[
t^{-(r_{j-1}-1)}e^{-\delta_{j-1}t}f(t)
\overset{a.s.}{\longrightarrow}
\mu_j\tilde Z_{j-1}.
\]
Since $\delta_{j-1}>0$, case \textup{(i)} of \lemmaref{lemma:large_time} gives $w_j(t)Z_j(t)\overset{a.s.}{\longrightarrow}\tilde Z_j$, which completes the induction.

If $\delta_\ell=0$, every type in $\pi_{:\ell}$ is non-supercritical, and $w_\ell(t)=t^{-(r_\ell-1)}$. 
If $r_\ell=1$, the conclusion follows from \lemmaref{lemma:subcritical_path_Z0}. 
If $r_\ell\ge2$, it follows from \lemmaref{lemma:non_supercritical_path}.
\end{proof}

\subsection{Laplace transforms along a path}
\label{app:lt_path_decomp}

Fix $t > 0$ and $s > 0$. Conditioned on the parent trajectory $\{Z_{\ell-1}(u)\}_{u\le t}$, the Laplace transform of $Z_\ell(t)$ satisfies
\begin{equation}
\label{eq:general_cond_LP}
\E\left[\exp\left(-s\, w_\ell(t)Z_\ell(t)\right)\,\middle|\,\left\{Z_{\ell-1}(u)\right\}_{u\le t}\right]
=
\exp\left(-\mu_\ell\int_0^t Z_{\ell-1}(u)H_t(s;u)\rmd u\right),
\end{equation}
where
\[
H_t(s;u)=1-\mathcal L_\ell^1\left(sw_\ell(t);t-u\right),
\]
and $\mathcal L_\ell^1$ is the Laplace transform of the $Z_\ell^1$ process defined in \eref{eq:L_1}.
To express the transform in terms of the large-time parent population, we decompose the parent trajectory around the following time-independent variable:
\[
X_{\ell-1}=
\begin{dcases}
\bar z_{\ell-1}, & Z_0=Z_0^0,\delta_{\ell-1}=0,r_{\ell-1}=1,\\
\tilde Z_{\ell-1}, & \text{otherwise}.
\end{dcases}
\]
In the first case, $Z_{\ell-1}(t)$ is ergodic (\lemmaref{lemma:subcritical_path_Z0}) and converges to a stationary distribution with nonvanishing fluctuations. 
Therefore, we use the stationary mean $\bar z_{\ell-1}$ (\eref{eq:subcritical_coordinate_average}). 
In the second case, when the parent converges as in \eref{eq:large_time_limit}, we use an independent random variable $\tilde Z_{\ell-1}$ with the corresponding limiting distribution.
Next, define
\begin{align}
I_t(s)&=\int_0^t w_{\ell-1}(u)^{-1}H_t(s;u)\rmd u,\\
R_t(s)&=\int_0^t\left(Z_{\ell-1}(u)-w_{\ell-1}(u)^{-1}X_{\ell-1}\right)H_t(s;u)\rmd u,\\
M_t(s)&=R_t(s)-\E\left[R_t(s)\right].
\end{align}
Taking expectations of \eref{eq:general_cond_LP} leads to the following decomposition:
\begin{equation}
\label{eq:LP_decomp}
    \E \left[ \exp\left(-s \, w_{\ell}(t)Z_\ell(t)\right) \right] = \exp(-\mu_{\ell} \E[R_t(s)]) \;
        \E \left[\exp(-\mu_{\ell}  I_t(s) X_{\ell - 1}) \;\exp(-\mu_{\ell} M_t(s))\right].
\end{equation}
By \thmref{thm:second_order_error} in \aref{app:bounded_mgf}, for each fixed scaled $s>0$, uniformly in $t>0$,
\begin{equation}
\label{eq:second_order_error}
    \E\left[\exp\left(-s\,w_{\ell}(t)Z_\ell(t)\right)\right]
    =\exp\left(-\mu_{\ell}\,\E\left[R_t(s)\right]\right)\,
    \E\left[\exp\left(-\mu_{\ell}\,I_t(s)\,X_{\ell - 1}\right)\right]\left(1 +\mathcal{O}(\mu_{\ell}^2)\right).
\end{equation}
Therefore, dropping the $M_t(s)$ term introduces a relative error of order $\mathcal{O}(\mu_{\ell}^2)$ for each fixed scaled argument $s>0$.

Since $I_t(s)\to I_\infty(s)$ (\propref{prop:It}) and $\E\left[R_t(s)\right]\to R_\infty(s)$ (\propref{prop:Rt}), for large $t$ and small $\mu_\ell$,
\begin{equation}
\label{eq:scaled_LT_limit}
    \E\left[\exp\left(-s\,w_{\ell}(t)Z_\ell(t)\right)\right]  \approx  \exp\left(-\mu_{\ell}\,R_{\infty}(s)\right)\,
\E\left[\exp\left(-\mu_{\ell}\,I_{\infty}(s)\,X_{\ell - 1}\right)\right].
\end{equation}
The corresponding formulas for $I_\infty$ and $R_\infty$ are provided in \aref{app:lt_exact}.

\subsection{Bounded moment generating functions and finite moments}
\label{app:bounded_mgf}

The results in this section are organized as follows.
We first derive bounds for the elementary one-type moment generating function in \lemmaref{lemma:elementary_kernel_bounds}.
Then, we use these bounds to prove uniform boundedness of the moment generating functions of $Z_\ell(t)$ and $w_\ell(t)Z_\ell(t)$ in \propref{prop:uniform_mgf}.
The corollaries \corref{cor:all_moments_bounded} and \corref{cor:limit_mgf_bounded} establish uniform moment bounds and transfer the MGF bound to the limiting variable $\tilde Z_\ell$.
Finally, we prove uniform boundedness of the moment generating functions of $R_t(s)$ and $M_t(s)$ in \propref{prop:Rt_uniform_mgf} and \corref{cor:Mt_uniform_mgf}, which is used to control the second-order error in \thmref{thm:second_order_error}.
Note that it suffices to prove the bounds for nonnegative $\theta$, since for any nonnegative random variable $Y$ and any $\theta\in\mathbb R$, $e^{\theta Y}\le e^{|\theta|Y}$.

\begin{lemma}
\label{lemma:elementary_kernel_bounds}
Fix $\varepsilon>0$ and an edge $v_{\ell-1}\to v_\ell$.
Let $t$ be the observation time of a clone founded at the mutation time $u\in[0,t]$ with age $h=t-u$.
Let
\[
\Psi_\ell(h,\theta)=\E\left[e^{\theta Z_\ell^1(h)}\right]-1.
\]
There exist constants $\theta_0>0$ and $K\in(0,\infty)$ such that
\begin{equation}
\label{eq:elementary_kernel_small_h_bound}
\Psi_\ell(h,\theta)\le K\theta
\quad\text{for }0\le h\le\varepsilon,\ 0\le\theta\le\theta_0,
\end{equation}
\begin{equation}
\label{eq:elementary_kernel_scaled_bound}
\Psi_\ell(t-u,\theta w_\ell(t))
\le K\theta e^{\lambda_\ell(t-u)}w_\ell(t)
\quad\text{for }t\ge\varepsilon,\ 0\le u\le t,\ 0\le\theta\le\theta_0,
\end{equation}
and
\begin{equation}
\label{eq:elementary_kernel_integral_bound}
\sup_{t\ge\varepsilon}\int_\varepsilon^t
\frac{e^{\lambda_\ell(t-u)}w_\ell(t)}{w_{\ell-1}(u)}\rmd u\le K.
\end{equation}
\end{lemma}

\begin{proof}
For a clone of age $h$, by the elementary one-type formula \eref{eq:L_1}, we have
\begin{equation}
\label{eq:shifted_mgf}
\Psi_\ell(h,\theta)=e^{\lambda_\ell h}
\frac{e^\theta-1}{1-\rho_\ell(h)(e^\theta-1)},
\qquad
\rho_\ell(h)=\alpha_\ell\int_0^h e^{\lambda_\ell s}\rmd s,
\end{equation}
whenever $\rho_\ell(h)(e^\theta-1)<1$.

Define
\[
C_{\varepsilon,\ell}:=
\begin{dcases}
\dfrac{\alpha_\ell\varepsilon^{-(r_\ell-1)}}{\delta_\ell},
& \delta_\ell>0,\\[2ex]
\dfrac{\alpha_\ell\varepsilon^{-(r_\ell-1)}}{-\lambda_\ell},
& \delta_\ell=0,\ \lambda_\ell<0,\\[2ex]
\alpha_\ell\varepsilon^{2-r_\ell},
& \delta_\ell=\lambda_\ell=0,
\end{dcases}
\]
where $r_\ell\ge2$ for $\delta_\ell=\lambda_\ell=0$ by definition \eref{eq:r}.
Since $\delta_\ell\ge\max\{0,\lambda_\ell\}$, we have $\rho_\ell(t-u)w_\ell(t)\le C_{\varepsilon,\ell}$ for all $t\ge\varepsilon$ and $0\le u\le t$.
The function $\rho_\ell$ is non-decreasing, whereas $w_\ell(t)=t^{-(r_\ell-1)}e^{-\delta_\ell t}$ is non-increasing. Define
\[
\theta_0:=\min\left\{
\frac{1}{1+w_\ell(\varepsilon)},
\frac{1}{2e\left(1+\max\{\rho_\ell(\varepsilon),C_{\varepsilon,\ell}\}\right)}
\right\}>0.
\]
For $0\le h\le\varepsilon$, $t\ge\varepsilon$, $0\le u\le t$, and $0\le\theta\le\theta_0$, we have $\theta<1$ and $\theta w_\ell(t)\le\theta_0w_\ell(\varepsilon)<1$.
Applying $e^\theta-1\le e\theta$ to both arguments and using monotonicity and $\rho_\ell(t-u)w_\ell(t)\le C_{\varepsilon,\ell}$, we have
\[
\max\left\{
\rho_\ell(h)(e^\theta-1),
\rho_\ell(t-u)\bigl(e^{\theta w_\ell(t)}-1\bigr)
\right\}
\le e\theta_0\max\{\rho_\ell(\varepsilon),C_{\varepsilon,\ell}\}
\le\frac12.
\]
By \eref{eq:shifted_mgf}, we have
\[
\begin{aligned}
\Psi_\ell(h,\theta)
&\le2e^{1+\max\{\lambda_\ell,0\}\varepsilon}\theta,\\
\Psi_\ell(t-u,\theta w_\ell(t))
&\le2e^{1+\max\{\lambda_\ell,0\}\varepsilon}
\theta e^{\lambda_\ell(t-u)}w_\ell(t).
\end{aligned}
\]

Finally, $r_{\ell-1}\ge1$. By the definitions \eref{eq:delta}, \eref{eq:r}, and \eref{eq:large_time_weight}, we have
\[
\int_\varepsilon^t
\frac{e^{\lambda_\ell(t-u)}w_\ell(t)}{w_{\ell-1}(u)}\rmd u
\le
\begin{dcases}
\dfrac{(r_{\ell-1}-1)!}
{(\lambda_\ell-\delta_{\ell-1})^{r_{\ell-1}}},
& \lambda_\ell>\delta_{\ell-1},\\[2ex]
\dfrac{1}{r_{\ell-1}},
& \lambda_\ell=\delta_{\ell-1},\\[2ex]
\dfrac{1}{\delta_{\ell-1}-\lambda_\ell},
& \lambda_\ell<\delta_{\ell-1}
\end{dcases}.
\]
Define
\[
K:=\max\left\{
2e^{1+\max\{\lambda_\ell,0\}\varepsilon},
\sup_{t\ge\varepsilon}\int_\varepsilon^t
\frac{e^{\lambda_\ell(t-u)}w_\ell(t)}{w_{\ell-1}(u)}\rmd u
\right\}<\infty.
\]
Since $2e\le2e^{1+\max\{\lambda_\ell,0\}\varepsilon}\le K$, we obtain all required bounds.
\end{proof}

\begin{proposition}
\label{prop:uniform_mgf}
Fix $\varepsilon>0$. For every $0\le\ell\le L$, there exist constants $\eta_\ell>0$ and $C_\ell<\infty$ such that for all $|\theta|\le\eta_\ell$,
\begin{subequations}
\label{eq:uniform_mgf}
\par\noindent
\begin{minipage}{0.49\linewidth}
\begin{equation}
\sup_{0\le t\le\varepsilon}\E\left[e^{\theta Z_\ell(t)}\right]
\le C_\ell,
\label{eq:uniform_mgf_early}
\end{equation}
\end{minipage}\hfill
\begin{minipage}{0.49\linewidth}
\begin{equation}
\sup_{t\ge\varepsilon}\E\left[e^{\theta w_\ell(t)Z_\ell(t)}\right]
\le C_\ell.
\label{eq:uniform_mgf_late}
\end{equation}
\end{minipage}\par
\end{subequations}
\end{proposition}

\begin{proof}
Since $Z_\ell(t)\ge0$, it suffices to consider $\theta\ge0$. We prove
\eref{eq:uniform_mgf} by induction on $\ell$.

For $\ell=0$, by \eref{eq:L_0} and \eref{eq:large_time_weight}, $p(t)$ and
$w_0(t)$ are non-increasing and
\[
\inf_{t\ge\varepsilon}\frac{p(t)}{w_0(t)}>0.
\]
Define
\[
\eta_0:=\min\left\{
\log(1+p(\varepsilon)/2),\frac{1}{w_0(\varepsilon)},
\frac{1}{2e}\inf_{t\ge\varepsilon}\frac{p(t)}{w_0(t)}
\right\}\in(0,1).
\]
For $0\le\theta\le\eta_0$, we have $0\le\theta w_0(t)\le1$ for
$t\ge\varepsilon$. Using $e^\theta-1\le e\theta$ for $0\le\theta\le1$,
we have
\[
\begin{aligned}
1-(1-p(t))e^\theta
&\ge p(t)-(e^\theta-1)
\ge\frac{p(t)}2,
&&0\le t\le\varepsilon,\\
1-(1-p(t))e^{\theta w_0(t)}
&\ge p(t)-e\theta w_0(t)
\ge\frac{p(t)}2,
&&t\ge\varepsilon.
\end{aligned}
\]
By \eref{eq:L_0} and \eref{eq:L_1}, with $1-b(t)=p(t)$, we have
\[
\begin{aligned}
\sup_{0\le t\le\varepsilon}\E\left[e^{\theta Z_0^0(t)}\right]
&\le2^{\nu_0/\alpha_0},
&
\sup_{0\le t\le\varepsilon}\E\left[e^{\theta Z_0^1(t)}\right]
&\le1+2e,\\
\sup_{t\ge\varepsilon}\E\left[e^{\theta w_0(t)Z_0^0(t)}\right]
&\le2^{\nu_0/\alpha_0},
&
\sup_{t\ge\varepsilon}\E\left[e^{\theta w_0(t)Z_0^1(t)}\right]
&\le1+2e.
\end{aligned}
\]
Thus, \eref{eq:uniform_mgf} holds at level $0$ with
\[
C_0:=\max\left\{2^{\nu_0/\alpha_0},1+2e\right\}<\infty.
\]

Fix $1\le\ell\le L$ and assume that \eref{eq:uniform_mgf} holds at level
$\ell-1$ with $\eta_{\ell-1}$ and $C_{\ell-1}$. Let $\theta_{0,\ell}$ and
$K_\ell$ be the constants from \lemmaref{lemma:elementary_kernel_bounds}, and
define
\[
\eta_\ell:=\min\left\{
\theta_{0,\ell},
\frac{\eta_{\ell-1}}{
\max\left\{
K_\ell\mu_\ell\varepsilon,\,
2\mu_\ell K_\ell e^{\max\{\lambda_\ell,0\}\varepsilon}
w_\ell(\varepsilon)\varepsilon,\,
2\mu_\ell K_\ell^2
\right\}}
\right\}>0.
\]
Fix $0\le\theta\le\eta_\ell$. By the exponential formula for the
mutation-time Poisson process, we have
\begin{equation}
\label{eq:conditional_clone_mgf}
\E\left[e^{\theta Z_\ell(t)}\middle|
\sigma\{Z_{\ell-1}(u):0\le u\le t\}\right]
=\exp\left(
\mu_\ell\int_0^t Z_{\ell-1}(u)\Psi_\ell(t-u,\theta)\rmd u
\right).
\end{equation}
For $0<t\le\varepsilon$, every clone has age at most $\varepsilon$.
By \eref{eq:elementary_kernel_small_h_bound}, Jensen's inequality, and
\eref{eq:uniform_mgf_early} at level $\ell-1$, we have
\[
\E\left[e^{\theta Z_\ell(t)}\right]
\le
\sup_{0\le u\le\varepsilon}
\E\left[e^{K_\ell\mu_\ell\theta\varepsilon Z_{\ell-1}(u)}\right]
\le C_{\ell-1}.
\]
Since $Z_\ell(0)=0$,
\eref{eq:uniform_mgf_early} holds at level $\ell$ with bound $C_{\ell-1}$.

For $t\ge\varepsilon$, set
\[
A_t:=\int_0^\varepsilon Z_{\ell-1}(u)
\Psi_\ell(t-u,\theta w_\ell(t))\rmd u,
\qquad
B_t:=\int_\varepsilon^t Z_{\ell-1}(u)
\Psi_\ell(t-u,\theta w_\ell(t))\rmd u.
\]
By \eref{eq:conditional_clone_mgf} at $\theta w_\ell(t)$ and
Cauchy--Schwarz, we have
\[
\E\left[e^{\theta w_\ell(t)Z_\ell(t)}\right]
=\E\left[e^{\mu_\ell(A_t+B_t)}\right]
\le
\E[e^{2\mu_\ell A_t}]^{1/2}\E[e^{2\mu_\ell B_t}]^{1/2}.
\]
For $0\le u\le\varepsilon$,
\[
e^{\lambda_\ell(t-u)}w_\ell(t)
\le e^{\max\{\lambda_\ell,0\}\varepsilon}w_\ell(\varepsilon).
\]
By \eref{eq:elementary_kernel_scaled_bound}, Jensen's
inequality, and \eref{eq:uniform_mgf_early} at level $\ell-1$, we have
\[
\E[e^{2\mu_\ell A_t}]
\le
\sup_{0\le u\le\varepsilon}
\E\left[
e^{2\mu_\ell K_\ell e^{\max\{\lambda_\ell,0\}\varepsilon}
w_\ell(\varepsilon)\theta\varepsilon
Z_{\ell-1}(u)}
\right]
\le C_{\ell-1}.
\]

For $t>\varepsilon$, the kernel
$e^{\lambda_\ell(t-u)}w_\ell(t)/w_{\ell-1}(u)$ on $[\varepsilon,t]$
has mass at most $K_\ell$ by \eref{eq:elementary_kernel_integral_bound}.
By \eref{eq:elementary_kernel_scaled_bound}, Jensen's inequality for the
normalized kernel, and \eref{eq:uniform_mgf_late} at level $\ell-1$, we have
\[
\E[e^{2\mu_\ell B_t}]
\le\sup_{u\ge\varepsilon}\E\left[
e^{2\mu_\ell K_\ell^2\theta
w_{\ell-1}(u)Z_{\ell-1}(u)}
\right]
\le C_{\ell-1}.
\]
For $t=\varepsilon$, $B_t=0$. By Cauchy--Schwarz,
\eref{eq:uniform_mgf_late} also holds with bound $C_{\ell-1}$.
Setting $C_\ell:=C_{\ell-1}$ completes the induction.
\end{proof}

\begin{corollary}
\label{cor:all_moments_bounded}
Fix $\varepsilon>0$. For every $0\le\ell\le L$ and integer $m\ge1$,
there exists a constant $C_{\ell,m}<\infty$ such that
\begin{subequations}
\label{eq:uniform_moments}
\par\noindent
\begin{minipage}{0.49\linewidth}
\begin{equation}
\sup_{0\le t\le\varepsilon}\E\left[Z_\ell(t)^m\right]\le C_{\ell,m},
\label{eq:uniform_moments_early}
\end{equation}
\end{minipage}\hfill
\begin{minipage}{0.49\linewidth}
\begin{equation}
\sup_{t\ge\varepsilon}\E\left[\left(w_\ell(t)Z_\ell(t)\right)^m\right]
\le C_{\ell,m}.
\label{eq:uniform_moments_late}
\end{equation}
\end{minipage}\par
\end{subequations}
\end{corollary}

\begin{proof}
Let $\eta_\ell$ and $C_\ell$ be as in \propref{prop:uniform_mgf}.
Since $y^m\le m!\eta_\ell^{-m}e^{\eta_\ell y}$ for $y\ge0$,
we may take $C_{\ell,m}:=m!C_\ell\eta_\ell^{-m}$.
\end{proof}

\begin{corollary}
\label{cor:limit_mgf_bounded}
For every $0\le\ell\le L$ such that
$w_\ell(t)Z_\ell(t)\overset{d}{\to}\tilde Z_\ell$ as $t\to\infty$,
there exist constants $\eta_\ell>0$ and $C_\ell<\infty$ such that
for all $|\theta|\le\eta_\ell$,
\begin{equation}
\label{eq:limit_mgf}
\E\left[e^{\theta\tilde Z_\ell}\right]\le C_\ell.
\end{equation}
In particular, $\E[\tilde Z_\ell^m]\le m!C_\ell\eta_\ell^{-m}<\infty$
for every integer $m\ge1$.
\end{corollary}

\begin{proof}
Fix $\varepsilon>0$ and take $\eta_\ell$ and $C_\ell$ from
\propref{prop:uniform_mgf}.
For $|\theta|\le\eta_\ell$, the function $x\mapsto e^{\theta x}$ is
nonnegative and continuous. By the Portmanteau lemma and
\eref{eq:uniform_mgf_late}, we have
\[
\E\left[e^{\theta\tilde Z_\ell}\right]
\le
\liminf_{t\to\infty}
\E\left[e^{\theta w_\ell(t)Z_\ell(t)}\right]
\le C_\ell.
\]
For the moment bound, apply $y^m\le m!\eta_\ell^{-m}e^{\eta_\ell y}$
with $y=\tilde Z_\ell\ge0$.
\end{proof}

\begin{proposition}
\label{prop:Rt_uniform_mgf}
Fix $s>0$. For every $1\le\ell\le L$, there exist constants
$\eta_\ell>0$ and $C_\ell<\infty$ such that for all $|\theta|\le\eta_\ell$,
\begin{equation}
\label{eq:Rt_mgf_uniform}
\sup_{t>0}\E\left[e^{\theta |R_t(s)|}\right]\le C_\ell.
\end{equation}
In particular, for every integer $m\ge1$,
\begin{equation}
\label{eq:Rt_moment_uniform}
\sup_{t>0}\E\left[|R_t(s)|^m\right]\le m!C_\ell\eta_\ell^{-m}<\infty.
\end{equation}
\end{proposition}

\begin{proof}
Fix $\varepsilon>0$ and take $\eta_{\ell-1}$ and $C_{\ell-1}$ from
\propref{prop:uniform_mgf}.
The variable $X_{\ell-1}$ is either $\tilde Z_{\ell-1}$ or its stationary mean.
It follows from \corref{cor:limit_mgf_bounded} and Jensen's inequality that
\[
\E\left[e^{\eta_{\ell-1}X_{\ell-1}}\right]
\le\E\left[e^{\eta_{\ell-1}\tilde Z_{\ell-1}}\right]\le C_{\ell-1}.
\]
Define
\[
\eta_\ell:=\frac{\eta_{\ell-1}}{3\max\{\varepsilon,I_\infty(s)\}}>0.
\]
Fix $0\le\theta\le\eta_\ell$ and set
\[
A_t:=\int_0^{t\wedge\varepsilon}Z_{\ell-1}(u)H_t(s;u)\rmd u,
\qquad
B_t:=\int_{t\wedge\varepsilon}^t Z_{\ell-1}(u)H_t(s;u)\rmd u.
\]
Since $I_t(s)\le I_\infty(s)$ by \propref{prop:It},
\[
|R_t(s)|=|A_t+B_t-X_{\ell-1}I_t(s)|
\le A_t+B_t+I_\infty(s)X_{\ell-1}.
\]
Using $0\le H_t(s;u)\le1$ and Jensen's inequality, we obtain from
\eref{eq:uniform_mgf_early} at level $\ell-1$ that
\[
\E[e^{3\theta A_t}]
\le\sup_{0\le u\le\varepsilon}
\E\left[e^{3\theta\varepsilon Z_{\ell-1}(u)}\right]
\le C_{\ell-1}.
\]
For $t>\varepsilon$, the kernel $w_{\ell-1}(u)^{-1}H_t(s;u)$ on
$[\varepsilon,t]$ has mass at most $I_\infty(s)$.
By Jensen's inequality for the normalized kernel and
\eref{eq:uniform_mgf_late} at level $\ell-1$, we have
\[
\E[e^{3\theta B_t}]
\le\sup_{u\ge\varepsilon}
\E\left[e^{3\theta I_\infty(s)w_{\ell-1}(u)Z_{\ell-1}(u)}\right]
\le C_{\ell-1}.
\]
For $t\le\varepsilon$, $B_t=0$. We now apply H\"older's inequality:
\[
\E\left[e^{\theta|R_t(s)|}\right]
\le\left(
\E[e^{3\theta A_t}]
\E[e^{3\theta B_t}]
\E[e^{3\theta I_\infty(s)X_{\ell-1}}]
\right)^{1/3}
\le C_{\ell-1}.
\]
Taking $C_\ell:=C_{\ell-1}$ proves \eref{eq:Rt_mgf_uniform}.
The moment bound follows from $|x|^m\le m!\eta_\ell^{-m}e^{\eta_\ell|x|}$.
\end{proof}

\begin{corollary}
\label{cor:Mt_uniform_mgf}
Fix $s>0$. For every $1\le\ell\le L$, there exist constants
$\eta_\ell>0$ and $C_\ell<\infty$ such that for all $|\theta|\le\eta_\ell$,
\begin{equation}
\label{eq:Mt_mgf_uniform}
\sup_{t>0}\E\left[e^{\theta |M_t(s)|}\right]\le C_\ell.
\end{equation}
In particular, for every integer $m\ge1$,
\begin{equation}
\label{eq:Mt_moment_uniform}
\sup_{t>0}\E\left[|M_t(s)|^m\right]\le m!C_\ell\eta_\ell^{-m}<\infty.
\end{equation}
\end{corollary}

\begin{proof}
Take $\eta_\ell$ from \propref{prop:Rt_uniform_mgf} and define
\[
C_\ell:=\left(\sup_{t>0}\E\left[e^{\eta_\ell|R_t(s)|}\right]\right)^2<\infty.
\]
For $0\le\theta\le\eta_\ell$, applying the triangle inequality and then
Jensen's inequality, we obtain
\[
\E\left[e^{\theta|M_t(s)|}\right]
\le e^{\theta\E[|R_t(s)|]}\E\left[e^{\theta|R_t(s)|}\right]
\le\left(\E\left[e^{\theta|R_t(s)|}\right]\right)^2
\le C_\ell.
\]
The moment bound follows from $|x|^m\le m!\eta_\ell^{-m}e^{\eta_\ell|x|}$.
\end{proof}

\begin{theorem}
\label{thm:second_order_error}
Fix $s>0$ and $1\le\ell\le L$, with all model parameters other than
$\mu_\ell$ fixed. Define
\[
\mathcal E_t(\mu_\ell)=
\frac{
\E\left[e^{-\mu_\ell I_t(s)X_{\ell-1}}e^{-\mu_\ell M_t(s)}\right]
}{
\E\left[e^{-\mu_\ell I_t(s)X_{\ell-1}}\right]
}-1.
\]
There exist constants $\mu_{0,\ell}>0$ and $C_\ell<\infty$ such that
for all $0<\mu_\ell\le\mu_{0,\ell}$,
\begin{equation}
\label{eq:second_order_error_uniform}
\sup_{t>0}|\mathcal E_t(\mu_\ell)|\le C_\ell\mu_\ell^2.
\end{equation}
\end{theorem}

\begin{proof}
The quantities $I_t(s)$, $X_{\ell-1}$, and $M_t(s)$ do not depend on
$\mu_\ell$. Take $\eta_\ell$ from \corref{cor:Mt_uniform_mgf} and set
$\mu_{0,\ell}:=\eta_\ell/2$.
Fix $0<\mu_\ell\le\mu_{0,\ell}$ and set
\[
A_t:=e^{-\mu_\ell I_t(s)X_{\ell-1}},
\qquad
B_t:=e^{-\mu_\ell M_t(s)}.
\]
Applying the Taylor bound $|e^{-x}-1+x|\le x^2e^{|x|}/2$, we write
\[
B_t=1-\mu_\ell M_t(s)+N_t,
\qquad
|N_t|\le\frac{\mu_\ell^2}{2}M_t(s)^2e^{\mu_\ell|M_t(s)|}.
\]

By \corref{cor:limit_mgf_bounded} and Jensen's inequality,
$\E[X_{\ell-1}^2]<\infty$.
Using $\E[M_t(s)]=0$, \propref{prop:It},
$1-e^{-x}\le x$ for $x\ge0$, and $2ab\le a^2+b^2$, we obtain
\[
|\E[A_tM_t(s)]| =|\E[(A_t-1)M_t(s)]|
\le\frac{\mu_\ell}{2}\left(
I_\infty(s)^2\E[X_{\ell-1}^2]+\E[M_t(s)^2]\right).
\]
Since $x^2e^{\mu_{0,\ell}x}\le8\eta_\ell^{-2}e^{\eta_\ell x}$ for $x\ge0$,
\[
\sup_{t>0}\E\left[M_t(s)^2e^{\mu_{0,\ell}|M_t(s)|}\right]
\le\frac{8}{\eta_\ell^2}
\sup_{t>0}\E\left[e^{\eta_\ell|M_t(s)|}\right]<\infty.
\]
Combining the linear term and the remainder, with $0<A_t\le1$, we obtain
\[
|\E[A_tB_t]-\E[A_t]|
\le\mu_\ell|\E[A_tM_t(s)]|+\E[|N_t|].
\]
Finally, by Jensen's inequality,
\[
\E[A_t]\ge
e^{-\mu_\ell I_t(s)\E[X_{\ell-1}]}
\ge e^{-\mu_{0,\ell}I_\infty(s)\E[X_{\ell-1}]}>0.
\]
Since $\mathcal E_t(\mu_\ell)=(\E[A_tB_t]-\E[A_t])/\E[A_t]$,
\eref{eq:second_order_error_uniform} holds with
\[
C_\ell:=
e^{\mu_{0,\ell}I_\infty(s)\E[X_{\ell-1}]}
\times\left(
\frac{I_\infty(s)^2}{2}\E[X_{\ell-1}^2]
+\sup_{t>0}\E\left[M_t(s)^2e^{\mu_{0,\ell}|M_t(s)|}\right]
\right)<\infty.
\]
\end{proof}

\subsection{Exact limits entering the Laplace transforms}
\label{app:lt_exact}

Throughout this subsection, we fix an edge $v_{\ell-1}\to v_\ell$. The definitions of $I_t$, $R_t$, and $M_t$ are provided in \aref{app:lt_path_decomp}.

\begin{proposition}
\label{prop:It}
For $s>0$,
\[
\sup_{t>0}I_t(s)=I_\infty(s)<\infty,
\]
and $I_t(s)\to I_\infty(s)$ as $t\to\infty$, where
\[
I_\infty(s)=
\begin{dcases}
s\,\dfrac{\Gamma(r_{\ell-1})}{\lambda_\ell^{r_{\ell-1}}}\,
\LerchPhi\left(-\dfrac{\alpha_\ell}{\lambda_\ell}s,r_{\ell-1},1-\dfrac{\delta_{\ell-1}}{\lambda_\ell}\right),
& \lambda_\ell>\delta_{\ell-1},\\[2ex]
\dfrac{s}{r_{\ell-1}}, & \lambda_\ell=\delta_{\ell-1}>0\text{ or }(\lambda_\ell=\delta_{\ell-1}=0\text{ and }r_{\ell-1}\ge2),\\[2ex]
\dfrac{1}{\alpha_\ell}\log\left(1+\alpha_\ell s\right), & \lambda_\ell=\delta_{\ell-1}=0\text{ and }r_{\ell-1}=1,\\[2ex]
\dfrac{s}{\delta_{\ell-1}-\lambda_\ell}, & \lambda_\ell<\delta_{\ell-1}\text{ and }(\delta_{\ell-1}>0\text{ or }r_{\ell-1}>1),\\[2ex]
-\dfrac{1}{\alpha_\ell}\log\left(\dfrac{\lambda_\ell}{\alpha_\ell e^{-s}-\beta_\ell}\right),
& \lambda_\ell<\delta_{\ell-1}=0\text{ and }r_{\ell-1}=1.
\end{dcases}
\]
\end{proposition}

\begin{proof}
With $\theta_t=sw_\ell(t)$ and $\rho_\ell$ from \eref{eq:shifted_mgf}, we have
\begin{equation}
\label{eq:It_kernel}
H_t(s;u)=
\frac{(1-e^{-\theta_t})e^{\lambda_\ell(t-u)}}
{1+\rho_\ell(t-u)(1-e^{-\theta_t})},
\qquad
0\le H_t(s;u)\le\theta_t e^{\lambda_\ell(t-u)}.
\end{equation}

If $\lambda_\ell>\delta_{\ell-1}$, then $w_\ell(t)=e^{-\lambda_\ell t}$ and, for each fixed $u\ge0$,
\[
H_t(s;u)\longrightarrow
\frac{se^{-\lambda_\ell u}}
{1+\frac{\alpha_\ell}{\lambda_\ell}se^{-\lambda_\ell u}}.
\]
For $0\le u\le t$,
\[
0\le u^{r_{\ell-1}-1}e^{\delta_{\ell-1}u}H_t(s;u)
\le su^{r_{\ell-1}-1}e^{-(\lambda_\ell-\delta_{\ell-1})u}.
\]
Extending the integrand by zero for $u>t$, dominated convergence implies
\[
I_t(s)\longrightarrow
\int_0^\infty u^{r_{\ell-1}-1}e^{\delta_{\ell-1} u}
\frac{se^{-\lambda_\ell u}}{1+\frac{\alpha_\ell}{\lambda_\ell}se^{-\lambda_\ell u}}\rmd u
=I_\infty(s).
\]
The Lerch form follows by substituting $x=\lambda_\ell u$ in the identity
\[
\LerchPhi(z,r,a)
=\frac{1}{\Gamma(r)}
\int_0^\infty\frac{x^{r-1}e^{-ax}}{1-ze^{-x}}\rmd x,
\qquad z<0,\ a>0,\ r=1,2,\ldots.
\]
For $0\le a\le t$, integration of \eref{eq:It_kernel} also shows that
\[
\begin{aligned}
\int_a^t H_t(s;u)\rmd u
&=\frac{1}{\alpha_\ell}\log\left(1+\rho_\ell(t-a)(1-e^{-\theta_t})\right)\\
&\le\frac{1}{\alpha_\ell}\log\left(1+\frac{\alpha_\ell}{\lambda_\ell}se^{-\lambda_\ell a}\right)
=\int_a^\infty
\frac{se^{-\lambda_\ell u}}{1+\frac{\alpha_\ell}{\lambda_\ell}se^{-\lambda_\ell u}}\rmd u.
\end{aligned}
\]
By integration by parts with the non-decreasing weight $u^{r_{\ell-1}-1}e^{\delta_{\ell-1} u}$,
we have $I_t(s)\le I_\infty(s)$.

If $\lambda_\ell=\delta_{\ell-1}>0$, or $\lambda_\ell=\delta_{\ell-1}=0$ and $r_{\ell-1}\ge2$,
then $\theta_t=st^{-r_{\ell-1}}e^{-\lambda_\ell t}\to0$ and
\[
\theta_t\rho_\ell(t)\le
\begin{dcases}
\dfrac{\alpha_\ell s}{\lambda_\ell}t^{-r_{\ell-1}}, & \lambda_\ell>0,\\[1ex]
\alpha_\ell s t^{1-r_{\ell-1}}, & \lambda_\ell=0
\end{dcases}
\longrightarrow0.
\]
Using \eref{eq:It_kernel} and monotonicity of $\rho_\ell$, we obtain
\[
\frac{s}{r_{\ell-1}}
\frac{1-e^{-\theta_t}}{\theta_t(1+\theta_t\rho_\ell(t))}
\le I_t(s)\le\frac{s}{r_{\ell-1}}.
\]
Thus, $I_t(s)\to s/r_{\ell-1}$.

If $\lambda_\ell=\delta_{\ell-1}=0$ and $r_{\ell-1}=1$, then $\theta_t=s/t$ and
\[
I_t(s)=\frac{1}{\alpha_\ell}\log\left(1+\alpha_\ell t(1-e^{-s/t})\right)
\le\frac{1}{\alpha_\ell}\log(1+\alpha_\ell s).
\]
Since $t(1-e^{-s/t})\to s$, we have $I_t(s)\to I_\infty(s)$.

If $\lambda_\ell<\delta_{\ell-1}$ and either $\delta_{\ell-1}>0$ or $r_{\ell-1}>1$,
then $\theta_t=st^{-(r_{\ell-1}-1)}e^{-\delta_{\ell-1}t}\to0$.
Substituting $h=t-u$ in $I_t(s)$, we obtain
\[
I_t(s)=\int_0^t s
\left(1-\frac{h}{t}\right)^{r_{\ell-1}-1}
e^{-(\delta_{\ell-1}-\lambda_\ell)h}
\frac{1-e^{-\theta_t}}
{\theta_t\left[1+\rho_\ell(h)(1-e^{-\theta_t})\right]}\rmd h.
\]
The integrand, extended by zero for $h>t$, converges to
$se^{-(\delta_{\ell-1}-\lambda_\ell)h}$ and is bounded by this integrable function.
Therefore, by dominated convergence,
\[
I_t(s)\longrightarrow\frac{s}{\delta_{\ell-1}-\lambda_\ell},
\qquad
I_t(s)\le\frac{s}{\delta_{\ell-1}-\lambda_\ell}.
\]

Finally, if $\lambda_\ell<\delta_{\ell-1}=0$ and $r_{\ell-1}=1$, then $\theta_t=s$.
Integrating \eref{eq:It_kernel}, we have
\[
I_t(s)=\frac{1}{\alpha_\ell}\log\left(1+\rho_\ell(t)(1-e^{-s})\right)
\longrightarrow-\frac{1}{\alpha_\ell}\log\left(\frac{\lambda_\ell}{\alpha_\ell e^{-s}-\beta_\ell}\right)
=I_\infty(s).
\]
In every case, $I_t(s)\le I_\infty(s)$ and $I_t(s)\to I_\infty(s)$, so
$\sup_{t>0}I_t(s)=I_\infty(s)<\infty$.
\end{proof}

\begin{proposition}
\label{prop:Rt}
For $s>0$, as $t\to\infty$,
\[
\E[R_t(s)] \longrightarrow R_\infty(s):=
\begin{dcases}
\displaystyle\int_0^\infty f(u)g_s(u)\rmd u, & \lambda_\ell>\delta_{\ell-1},\\[2ex]
0, & \lambda_\ell\le\delta_{\ell-1},
\end{dcases}
\]
where, for $u>0$,
\[
f(u)=\E\left[w_{\ell-1}(u)Z_{\ell-1}(u)\right]-\E[X_{\ell-1}],
\]
and, when $\lambda_\ell>\delta_{\ell-1}$,
\[
g_s(u)=u^{r_{\ell-1}-1}e^{\delta_{\ell-1} u}
\frac{se^{-\lambda_\ell u}}{1+\frac{\alpha_\ell}{\lambda_\ell}se^{-\lambda_\ell u}}.
\]
\end{proposition}

\begin{proof}
By \corref{cor:all_moments_bounded} with $\varepsilon=1$ and $m=2$, $\{w_{\ell-1}(u)Z_{\ell-1}(u):u\ge1\}$ is uniformly integrable, and by \eref{eq:large_time_limit},
\[
\E\left[w_{\ell-1}(u)Z_{\ell-1}(u)\right]
\longrightarrow \E[\tilde Z_{\ell-1}]=\E[X_{\ell-1}].
\]
If $Z_0=Z_0^0$, $\delta_{\ell-1}=0$, and $r_{\ell-1}=1$, the last equality follows from $\E[\tilde Z_{\ell-1}]=\bar z_{\ell-1}$ in \lemmaref{lemma:subcritical_path_Z0}.
Thus, $f(u)\to0$.

For $u\ge0$, define
\[
k_t(u)=
\begin{dcases}
w_{\ell-1}(u)^{-1}H_t(s;u), & 0<u\le t,\\
0, & \text{otherwise}.
\end{dcases}
\]
Since $0\le H_t(s;u)\le1$ and $I_t(s)\le I_\infty(s)$ by \propref{prop:It},
\[
\E\int_0^t
\left|
Z_{\ell-1}(u)-w_{\ell-1}(u)^{-1}X_{\ell-1}
\right|H_t(s;u)\rmd u
\le
\int_0^t\E[Z_{\ell-1}(u)]\rmd u
+\E[X_{\ell-1}]I_t(s)
<\infty.
\]
By Fubini's theorem,
\[
\E[R_t(s)]
=\int_0^\infty f(u)k_t(u)\rmd u.
\]

If $\lambda_\ell>\delta_{\ell-1}$, then $k_t(u)\to g_s(u)$ for each $u>0$ by \eref{eq:It_kernel}, and
\[
|f(u)|k_t(u)
\le s\Bigl(
\E[Z_{\ell-1}(u)]e^{-\lambda_\ell u}
+\E[X_{\ell-1}]u^{r_{\ell-1}-1}e^{-(\lambda_\ell-\delta_{\ell-1})u}
\Bigr).
\]
The right-hand side is integrable by \corref{cor:all_moments_bounded} and $\lambda_\ell>\delta_{\ell-1}$. 
Dominated convergence then implies
\[
\E[R_t(s)]
\longrightarrow
\int_0^\infty f(u)g_s(u)\rmd u.
\]

If $\lambda_\ell\le\delta_{\ell-1}$, then $k_t(u)\to0$ for each $u>0$ by \eref{eq:It_kernel}.
For each $U\ge1$, the bound
\[
|f(u)|k_t(u)
\le
\E[Z_{\ell-1}(u)]
+\E[X_{\ell-1}]u^{r_{\ell-1}-1}e^{\delta_{\ell-1} u},
\qquad 0<u\le U,
\]
has an integrable right-hand side, so by dominated convergence, $\int_0^U f(u)k_t(u)\rmd u \longrightarrow 0$.
Since $\int_0^\infty k_t(u)\rmd u=I_t(s)\le I_\infty(s)$ by \propref{prop:It},
\[
\limsup_{t\to\infty}|\E[R_t(s)]|
\le I_\infty(s)\sup_{u\ge U}|f(u)|
\longrightarrow0
\qquad\text{as }U\to\infty.
\]
\end{proof}

\begin{lemma}
\label{lemma:mean_recursion}
For a $Z^1$-rooted path, let $m_\ell(t)=\E[Z_\ell(t)]$ for $0\le\ell\le L$, and define
\[
\hat\delta_\ell=\max_{0\le j\le\ell}\lambda_j,\qquad
\hat r_\ell=\#\{0\le j\le\ell:\lambda_j=\hat\delta_\ell\}.
\]
Then, for $0\le\ell\le L$,
\[
m_\ell(t)\sim A_\ell t^{\hat r_\ell-1}e^{\hat\delta_\ell t}
\qquad\text{as }t\to\infty,
\]
where $\sim$ means that the ratio tends to $1$. The coefficients satisfy $A_0=1$ and, for $1\le\ell\le L$,
\[
A_\ell=
\begin{dcases}
\dfrac{\mu_\ell A_{\ell-1}}{\hat r_{\ell-1}}, & \lambda_\ell=\hat\delta_{\ell-1},\\[2ex]
\dfrac{\mu_\ell A_{\ell-1}}{\hat\delta_{\ell-1}-\lambda_\ell}, & \lambda_\ell<\hat\delta_{\ell-1},\\[2ex]
\displaystyle\prod_{j=1}^{\ell}\frac{\mu_j}{\lambda_\ell-\lambda_{j-1}}, & \lambda_\ell>\hat\delta_{\ell-1}.
\end{dcases}
\]
\end{lemma}

\begin{proof}
For $\ell=0$, we have $m_0(t)=e^{\lambda_0t}$. For $1\le\ell\le L$, the mean satisfies
\begin{equation}
\label{eq:path_mean_ode}
m_\ell'(t)=\lambda_\ell m_\ell(t)+\mu_\ell m_{\ell-1}(t),
\qquad m_\ell(0)=0.
\end{equation}
Multiplying by $e^{-\lambda_\ell t}$ and integrating, we obtain
\begin{equation}
\label{eq:path_mean_integral}
m_\ell(t)=\mu_\ell e^{\lambda_\ell t}\int_0^t m_{\ell-1}(u)e^{-\lambda_\ell u}\rmd u.
\end{equation}
Assume the conclusion holds through index $\ell-1$. If $\lambda_\ell\le\hat\delta_{\ell-1}$, l'H\^opital's rule implies
\[
\int_0^t m_{\ell-1}(u)e^{-\lambda_\ell u}\rmd u
\sim
\begin{dcases}
\dfrac{A_{\ell-1}}{\hat r_{\ell-1}}t^{\hat r_{\ell-1}},
& \lambda_\ell=\hat\delta_{\ell-1},\\[2ex]
\dfrac{A_{\ell-1}}{\hat\delta_{\ell-1}-\lambda_\ell}
t^{\hat r_{\ell-1}-1}e^{(\hat\delta_{\ell-1}-\lambda_\ell)t},
& \lambda_\ell<\hat\delta_{\ell-1}.
\end{dcases}
\]
The asserted asymptotic for $m_\ell(t)$ follows from \eref{eq:path_mean_integral}.

If $\lambda_\ell>\hat\delta_{\ell-1}$, then, for all sufficiently large $u$,
\[
0\le m_{\ell-1}(u)e^{-\lambda_\ell u}
\le 2A_{\ell-1}u^{\hat r_{\ell-1}-1}e^{-(\lambda_\ell-\hat\delta_{\ell-1})u}.
\]
Since $\lambda_\ell-\hat\delta_{\ell-1}>0$, this bound is integrable, so $\int_0^\infty m_{\ell-1}(u)e^{-\lambda_\ell u}\rmd u<\infty$.
For $1\le j\le\ell-1$, multiplying \eref{eq:path_mean_ode} at index $j$ by $e^{-\lambda_\ell t}$ and integrating by parts, we obtain
\[
\int_0^\infty e^{-\lambda_\ell u}m_j(u)\rmd u
=\frac{\mu_j}{\lambda_\ell-\lambda_j}
\int_0^\infty e^{-\lambda_\ell u}m_{j-1}(u)\rmd u.
\]
Iterating this identity and using \eref{eq:path_mean_integral}, we conclude that
\[
e^{-\lambda_\ell t}m_\ell(t)
\longrightarrow \mu_\ell\int_0^\infty e^{-\lambda_\ell u}m_{\ell-1}(u)\rmd u
=\prod_{j=1}^{\ell}\frac{\mu_j}{\lambda_\ell-\lambda_{j-1}}
=A_\ell.
\]
\end{proof}

\begin{corollary}
\label{cor:Rt}
Suppose the path is $Z^1$-rooted and $v_\ell$ is its first supercritical type:
\[
\lambda_j\le0\quad(0\le j\le\ell-1),
\qquad
\lambda_\ell>0.
\]
For $s>0$, with $m_{\ell-1}(t)=\E[Z_{\ell-1}(t)]$,
\[
\E[R_t(s)]\longrightarrow
R_\infty(s)=
\int_0^\infty m_{\ell-1}(u)
\frac{se^{-\lambda_\ell u}}{1+\frac{\alpha_\ell}{\lambda_\ell}se^{-\lambda_\ell u}}\rmd u
\qquad\text{as }t\to\infty.
\]
\end{corollary}

\begin{proof}
By the extinction argument in \propref{prop:Z1_path_limit}, $X_{\ell-1}=\tilde Z_{\ell-1}=0$.
Thus, $f(u)=w_{\ell-1}(u)m_{\ell-1}(u)$.
Since $\delta_{\ell-1}=0<\lambda_\ell$, the formula follows from \propref{prop:Rt}.
\end{proof}

\begin{lemma}
\label{lemma:Z0_path_stationary_mean}
Fix $0\le j\le L$. Suppose the path is $Z^0$-rooted and $\lambda_k<0$ for $0\le k\le j$.
The stationary means satisfy
\[
\bar z_0=\frac{\nu_0}{-\lambda_0},
\qquad
\bar z_k=\frac{\mu_k}{-\lambda_k}\bar z_{k-1},\quad 1\le k\le j,
\]
and, for $0\le k\le j$,
\[
\bar z_k=
\frac{\nu_0}{-\lambda_0}
\prod_{a=1}^{k}\frac{\mu_a}{-\lambda_a}.
\]
\end{lemma}

\begin{proof}
In the stationary distribution from \lemmaref{lemma:subcritical_path_Z0}, the mean equations reduce to
\[
0=\nu_0+\lambda_0\bar z_0,
\qquad
0=\mu_k\bar z_{k-1}+\lambda_k\bar z_k,\quad 1\le k\le j.
\]
Solving these equations and iterating the recursion establishes the formulas.
\end{proof}

\subsection{Closed-form Laplace transform recursions}
\label{app:lt_closed}

For each fixed scaled argument $s>0$, \eref{eq:scaled_LT_limit} combines the second-order approximation in \eref{eq:second_order_error} with the limits $I_t(s)\to I_\infty(s)$ and $\E[R_t(s)]\to R_\infty(s)$. The relative $\mathcal O(\mu_\ell^2)$ bound controls the removal of the centered fluctuation $M_t$, with all other model parameters fixed. To obtain a closed-form recursion, we distinguish three cases.

First, if $\lambda_\ell\le\delta_{\ell-1}$, then $R_\infty(s)=0$ by \propref{prop:Rt}, and
\begin{equation}
\label{eq:no_mean_correction}
\E\left[\exp\left(-s\,w_\ell(t)Z_\ell(t)\right)\right]
\approx
\E\left[\exp\left(-\mu_\ell I_\infty(s)X_{\ell-1}\right)\right].
\end{equation}

Second, suppose the path is $Z^1$-rooted, $\delta_{\ell-1}=0$, and $\lambda_\ell>0$. Then $X_{\ell-1}=\tilde Z_{\ell-1}=0$, so only the mean correction remains:
\[
\E\left[\exp\left(-s\,w_\ell(t)Z_\ell(t)\right)\right]
\approx \exp\left(-\mu_\ell R_\infty(s)\right).
\]
In the formula for $R_\infty(s)$ in \corref{cor:Rt}, we replace the parent mean by its large-time equivalent from \lemmaref{lemma:mean_recursion} and define, for $s\ge0$,
\[
\begin{aligned}
J_\infty(s)&:=
\int_0^\infty A_{\ell-1}u^{\hat r_{\ell-1}-1}e^{\hat\delta_{\ell-1}u}
\frac{se^{-\lambda_\ell u}}
{1+\frac{\alpha_\ell}{\lambda_\ell}se^{-\lambda_\ell u}}\rmd u\\
&=
A_{\ell-1}s\frac{\Gamma(\hat r_{\ell-1})}{\lambda_\ell^{\hat r_{\ell-1}}}
\LerchPhi\left(
-\frac{\alpha_\ell}{\lambda_\ell}s,
\hat r_{\ell-1},
1-\frac{\hat\delta_{\ell-1}}{\lambda_\ell}
\right).
\end{aligned}
\]
The equality follows from the Lerch integral representation in the proof of \propref{prop:It}, and $J_\infty(0)=0$.
Replacing $R_\infty(s)$ by $J_\infty(s)$, we obtain the closed-form approximation
\begin{equation}
\label{eq:J_approximation}
\E\left[\exp\left(-s\,w_\ell(t)Z_\ell(t)\right)\right]
\approx \exp\left(-\mu_\ell J_\infty(s)\right).
\end{equation}
The mean asymptotic alone does not provide a bound on $|J_\infty(s)-R_\infty(s)|$, so this substitution introduces an additional approximation.

Finally, suppose $\lambda_\ell>\delta_{\ell-1}$ and either $\delta_{\ell-1}>0$ or the path is $Z^0$-rooted. Here, $R_\infty(s)$ need not vanish. Motivated by approximations that replace parent histories by their large-time forms \citep{durrett2015branching,nicholson2023sequential}, we neglect $R_\infty(s)$ and take
\begin{equation}
\label{eq:dropped_mean_correction}
\E\left[\exp\left(-s\,w_\ell(t)Z_\ell(t)\right)\right]
\approx
\E\left[\exp\left(-\mu_\ell I_\infty(s)X_{\ell-1}\right)\right].
\end{equation}
Relative to the right-hand side of \eref{eq:scaled_LT_limit}, this omission introduces the relative error $e^{\mu_\ell R_\infty(s)}-1$. With $s$ and all parameters except $\mu_\ell$ fixed, this error is first order in $\mu_\ell$ when $R_\infty(s)\ne0$.

In \eref{eq:no_mean_correction} and \eref{eq:dropped_mean_correction},
\[
\E\left[\exp\left(-\mu_\ell I_\infty(s)X_{\ell-1}\right)\right]
=
\begin{dcases}
\exp\left(-\mu_\ell I_\infty(s)\bar z_{\ell-1}\right),
& Z_0=Z_0^0,\ \delta_{\ell-1}=0,\ r_{\ell-1}=1,\\[1ex]
\E\left[\exp\left(-\mu_\ell I_\infty(s)\tilde Z_{\ell-1}\right)\right],
& \text{otherwise},
\end{dcases}
\]
where, in the stationary case, \lemmaref{lemma:Z0_path_stationary_mean} states that
\[
\bar z_{\ell-1}
=\frac{\nu_0}{-\lambda_0}
\prod_{k=1}^{\ell-1}\frac{\mu_k}{-\lambda_k}.
\]

When $X_{\ell-1}=\tilde Z_{\ell-1}$, convergence in distribution in \eref{eq:large_time_limit} implies, for each fixed $s>0$,
\[
\E\left[e^{-\mu_\ell I_\infty(s)w_{\ell-1}(t)Z_{\ell-1}(t)}\right]
\longrightarrow
\E\left[e^{-\mu_\ell I_\infty(s)X_{\ell-1}}\right],
\]
because the exponential is bounded and continuous on $[0,\infty)$. We use this replacement in \eref{eq:no_mean_correction} and \eref{eq:dropped_mean_correction}, retaining $X_{\ell-1}=0$ for $Z^1$-rooted paths with $\delta_{\ell-1}=0$ and $X_{\ell-1}=\bar z_{\ell-1}$ in the stationary case.

For an unscaled Laplace argument $s>0$, set $\tilde s=s/w_\ell(t)$. Since $\tilde s$ may grow with $t$, the preceding fixed-argument limits and second-order bound do not provide an error bound for this substitution. We use it as an additional approximation to obtain \eref{eq:laplace_recursion}:
\[
    \E\left[\exp\left(-s \, Z_\ell(t)\right)\right]\approx
    \exp\left(-c^{\pi_{:{\ell}}}_{v_{\ell-1}v_\ell}(s)\right)\,
    \E\left[\exp\left(-h^{\pi_{:{\ell}}}_{v_{\ell-1}v_\ell}(s)\,Z_{\ell-1}(t)\right)\right].
\]
The edge maps, with their dependence on $t$ suppressed, are
\begin{equation}
\label{eq:c_edge_map}
c^{\pi_{:{\ell}}}_{v_{\ell-1}v_\ell}(s) =
\begin{dcases}
    \mu_{\ell} \,J_\infty(\tilde s) & Z_0=Z_0^1,\,\delta_{\ell-1}=0,\,\lambda_\ell>0, \\[1ex]
    \mu_{\ell}\,I_\infty(\tilde s)\,\bar z_{\ell - 1} & Z_0=Z_0^0,\,\delta_{\ell-1}=0,\,r_{\ell-1}=1, \\[1ex]
    0  & \text{otherwise},
\end{dcases}
\end{equation}
and
\begin{equation}
\label{eq:h_edge_map}
h^{\pi_{:{\ell}}}_{v_{\ell-1}v_\ell}(s) =
\begin{dcases}
    0 & Z_0=Z_0^1,\,\delta_{\ell-1}=0, \\[1ex]
    0 & Z_0=Z_0^0,\,\delta_{\ell-1}=0,\,r_{\ell-1}=1, \\[1ex]
    \mu_{\ell}\,w_{\ell - 1}(t)\,I_\infty(\tilde{s})  & \text{otherwise}.
\end{dcases}
\end{equation}

\subsection{Laplace transforms of the general system}
\label{app:lt_general}

Let $\mathcal G=(V,E)$ be a finite DAG and fix $t>0$ and $\bs\in[0,\infty)^V$. We return to vertex-indexed notation and suppress the dependence of the recursive maps on $t$.

For a source $i$ and a path $\pi\in\Pi_{i\to j}$, let $Z_j^{(i),[0],\pi}$ and $Z_j^{(i),[1],m,\pi}$ denote the contributions from the immigration process $Z_i^{(i),0}$ and the $m$-th initial-cell process $Z_i^{(i),1,m}$, respectively. Their sum is $Z_j^{(i),\pi}$ in \eref{eq:path_decomposition}. These root channels are independent, but distinct paths within one channel can be dependent. We suppress $m$ when considering one initial cell.

\paragraph{Recursive maps}

We write the edge maps \eref{eq:c_edge_map} and \eref{eq:h_edge_map} separately for the immigration channel $a=0$ and the initial-cell channel $a=1$. For $s\ge0$ and a path $\pi=(v_0\to\cdots\to v_\ell)$ with $\ell\ge1$, set $\tilde s=s/w_{v_\ell}^\pi(t)$ as in \aref{app:lt_closed}. The functions $I_\infty$ and $J_\infty$ are evaluated for the final edge of $\pi$, using the parameters of the parent prefix $\pi_{:\ell-1}$. We extend $I_\infty$ to zero by $I_\infty(0)=0$.
When $\lambda_{v_k}<0$ for $0\le k\le\ell-1$, the stationary immigration-channel mean is
\[
\bar z_{v_{\ell-1}}^\pi
=\frac{\nu_{v_0}}{-\lambda_{v_0}}
\prod_{k=1}^{\ell-1}\frac{\mu_{v_{k-1}v_k}}{-\lambda_{v_k}}
\]
by \lemmaref{lemma:Z0_path_stationary_mean}. The edge maps are
\begin{align}
c_{v_{\ell-1}v_\ell}^{[0],\pi_{:\ell}}(s)&=
\begin{dcases}
\mu_{v_{\ell-1}v_\ell}I_\infty(\tilde s)\bar z_{v_{\ell-1}}^{\pi},
& \delta_{v_{\ell-1}}^\pi=0,\ r_{v_{\ell-1}}^\pi=1,\\
0, & \text{otherwise},
\end{dcases}
\label{eq:c_edge_map_0}\\
h_{v_{\ell-1}v_\ell}^{[0],\pi_{:\ell}}(s)&=
\begin{dcases}
0, & \delta_{v_{\ell-1}}^\pi=0,\ r_{v_{\ell-1}}^\pi=1,\\
\mu_{v_{\ell-1}v_\ell}w_{v_{\ell-1}}^\pi(t)I_\infty(\tilde s),
& \text{otherwise},
\end{dcases}
\label{eq:h_edge_map_0}\\
c_{v_{\ell-1}v_\ell}^{[1],\pi_{:\ell}}(s)&=
\begin{dcases}
\mu_{v_{\ell-1}v_\ell}J_\infty(\tilde s),
& \delta_{v_{\ell-1}}^\pi=0,\ \lambda_{v_\ell}>0,\\
0, & \text{otherwise},
\end{dcases}
\label{eq:c_edge_map_1}\\
h_{v_{\ell-1}v_\ell}^{[1],\pi_{:\ell}}(s)&=
\begin{dcases}
0, & \delta_{v_{\ell-1}}^\pi=0,\\
\mu_{v_{\ell-1}v_\ell}w_{v_{\ell-1}}^\pi(t)I_\infty(\tilde s),
& \text{otherwise}.
\end{dcases}
\label{eq:h_edge_map_1}
\end{align}

All four edge maps are nonnegative and vanish at zero. For each $a\in\{0,1\}$ and each path $\pi=(v_0\to\cdots\to v_\ell)$, including $\ell=0$, define the backward maps
\begin{equation}
\label{eq:G_map_channels}
G_\ell^{[a]}\left(\bs;\pi\right)
=s_{v_\ell}
+\sum_{u\in\child(v_\ell)}
h_{v_\ell u}^{[a],\pi\to u}
\left(G_{\ell+1}^{[a]}\left(\bs;\pi\to u\right)\right),
\end{equation}
\begin{equation}
\label{eq:C_map_channels}
C_\ell^{[a]}\left(\bs;\pi\right)
=\sum_{u\in\child(v_\ell)}
\left[
c_{v_\ell u}^{[a],\pi\to u}
\left(G_{\ell+1}^{[a]}\left(\bs;\pi\to u\right)\right)
+C_{\ell+1}^{[a]}\left(\bs;\pi\to u\right)
\right],
\end{equation}
with terminal conditions $G_\ell^{[a]}\left(\bs;\pi\right)=s_{v_\ell}$ and $C_\ell^{[a]}\left(\bs;\pi\right)=0$ at leaves.
Distinct paths ending at the same vertex remain separate nodes of the recursion. Since the DAG is finite, the recursions terminate and all transform arguments are nonnegative.

\paragraph{Recursive LTSM construction}
Under the recursive endpoint approximation specified below, these maps define the joint approximation \eref{eq:joint_LT_channels}. For each $i\in V$ and $s\ge0$, its marginal form is
\begin{equation}
\label{eq:marginal_LT_channels}
\begin{aligned}
\E\left[e^{-sZ_i(t)}\right]
&\approx
\prod_{v\in\Anc^\ast(i)}\Biggl[
\exp\left(
-C^{[0]}_0\left(s\be_i;(v)\right)
-z_{0,v}C^{[1]}_0\left(s\be_i;(v)\right)
\right)\\
&\qquad\times
\mathcal L_v^0\left(G^{[0]}_0\left(s\be_i;(v)\right);t\right)
\left[
\mathcal L_v^1\left(G^{[1]}_0\left(s\be_i;(v)\right);t\right)
\right]^{z_{0,v}}\Biggr].
\end{aligned}
\end{equation}

\paragraph{Derivation}
By \eref{eq:subsystem_decomposition} and the independence of the source subsystems,
\begin{equation}
\label{eq:source_LT_factorization}
\E\left[e^{-\bs^\top\bz(t)}\right]
=
\prod_{i\in V}
\E\left[e^{-\bs^\top\bz^{(i)}(t)}\right].
\end{equation}

Fix a source $i$ and a channel $a\in\{0,1\}$, with one initial cell fixed when $a=1$. Recursive use of \eref{eq:cond_laplace_recursion} treats the full child continuation branches as conditionally independent given their parent endpoint. It also replaces each child's continuation law by a function of that child's endpoint alone, even when ancestor endpoints are also conditioned on. These are additional approximations to the dependence on full population histories.

For a path $\pi=(i\to\cdots\to v_\ell)$ and a path $\rho$ from $v_\ell$ to a descendant, write $\pi\to\rho$ for their concatenation, including $v_\ell$ only once. Under the recursive endpoint approximation, reverse induction on the finite tree of paths establishes
\begin{equation}
\label{eq:approx_conditional_subtree}
\begin{split}
&\E\left[
\left.
\exp\left(
-\sum_{j\in\Desc(v_\ell)}
\sum_{\rho\in\Pi_{v_\ell\to j}}
s_jZ_j^{(i),[a],\pi\to\rho}(t)
\right)
\right|Z_{v_\ell}^{(i),[a],\pi}(t)
\right]\\
&\qquad\approx
\exp\left(-C^{[a]}_\ell(\bs;\pi)\right)
\exp\left(
-\left(G^{[a]}_\ell(\bs;\pi)-s_{v_\ell}\right)
Z_{v_\ell}^{(i),[a],\pi}(t)
\right).
\end{split}
\end{equation}
At a leaf, the sums are empty and both sides equal one by the terminal conditions.

At an internal node, suppose \eref{eq:approx_conditional_subtree} holds for every child $u$. Including the child population $Z_u^{(i),[a],\pi\to u}(t)$ in its continuation transform changes the coefficient from $G^{[a]}_{\ell+1}(\bs;\pi\to u)-s_u$ to $G^{[a]}_{\ell+1}(\bs;\pi\to u)$. Under the recursive endpoint approximation, the conditional transform is therefore approximated by
\[
\begin{aligned}
&\prod_{u\in\child(v_\ell)}\Biggl[
\exp\left(
-c_{v_\ell u}^{[a],\pi\to u}
\left(G^{[a]}_{\ell+1}(\bs;\pi\to u)\right)
-C^{[a]}_{\ell+1}(\bs;\pi\to u)
\right)\\
&\qquad\times
\exp\left(
-h_{v_\ell u}^{[a],\pi\to u}
\left(G^{[a]}_{\ell+1}(\bs;\pi\to u)\right)
Z_{v_\ell}^{(i),[a],\pi}(t)
\right)\Biggr].
\end{aligned}
\]
By \eref{eq:G_map_channels} and \eref{eq:C_map_channels}, this expression equals the right-hand side of \eref{eq:approx_conditional_subtree}.

At the root, we include the root population and average over its elementary law. The immigration-channel approximation is
\[
\exp\left(-C^{[0]}_0(\bs;(i))\right)
\mathcal L_i^0\left(G^{[0]}_0(\bs;(i));t\right),
\]
and, for each initial cell, the approximation
\[
\exp\left(-C^{[1]}_0(\bs;(i))\right)
\mathcal L_i^1\left(G^{[1]}_0(\bs;(i));t\right).
\]
The root channels are independent. Multiplying the immigration factor and the $z_{0,i}$ initial-cell factors, then substituting into \eref{eq:source_LT_factorization}, we obtain \eref{eq:joint_LT_channels}.

For the marginal transform, set $\bs=s\be_i$. If $v\notin\Anc^\ast(i)$, there is no directed path from $v$ to $i$. Since the edge maps vanish at zero, backward induction in \eref{eq:G_map_channels} and \eref{eq:C_map_channels} implies
\[
G^{[a]}_0(s\be_i;(v))=C^{[a]}_0(s\be_i;(v))=0,
\qquad a\in\{0,1\}.
\]
Also, $\mathcal L_v^0(0;t)=\mathcal L_v^1(0;t)=1$. Hence, all factors with $v\notin\Anc^\ast(i)$ equal one, which proves the reduction to \eref{eq:marginal_LT_channels}.

Factorization across source subsystems and their root channels is exact. The edge-map approximations in \aref{app:lt_closed} and the recursive replacement of histories by endpoints are additional approximations. No error bound for the joint recursion is established here.

\subsection{Inverse Laplace transforms for rectangular probabilities}
\label{app:ILT-rectangular}

Let $\bn=\left(n_1,\dots,n_d\right)\in\mathbb Z_+^d$, $\bs=\left(s_1,\dots,s_d\right)\in(0,\infty)^d$, and $\bb=\left(b_1,\dots,b_d\right)\in\mathbb Z_+^d$. With continuity corrections, set
\[
\tilde{\bn}=\left(n_1+\tfrac{1}{2},\dots,n_d+\tfrac{1}{2}\right),\qquad
\tilde{\bb}=\left(b_1+\tfrac{1}{2},\dots,b_d+\tfrac{1}{2}\right).
\]
In one dimension, let $F(t,x)=P(Z(t)\le x)$ for $x\ge0$ and $F(t,x)=0$ for $x<0$, and extend the probability mass function to the unit-width piecewise-constant function
\[
p(t,x)=F(t,x)-F(t,x-1)u(x-1),\qquad x\ge0,
\]
where $u(\cdot)$ is the Heaviside step function. The transforms used in \sref{sec:laplace} then follow from
\[
\mathcal L\left\{F(t,\cdot)\right\}(s)=\frac{1}{s}\E\left[e^{-sZ(t)}\right],
\]
and
\[
\mathcal L\{p(t,\cdot)\}(s)
=\frac{1-e^{-s}}{s}\E\left[e^{-sZ(t)}\right].
\]
For integers $0\le a_j\le b_j$, write $\ba=\left(a_1,\dots,a_d\right)$ and $m_j=b_j-a_j+1$, and set $\mathbf{m}=\left(m_1,\dots,m_d\right)$. Define the axis-aligned rectangle window
\[
B_{\mathbf{m}}\left(\mathbf{x}\right)=\prod_{j=1}^d 1_{\left(0,m_j\right]}\left(x_j\right),\qquad \mathbf{x}\in\mathbb R^d.
\]
Let $\mu_t$ be the probability measure 
\[
\mu_t=\sum_{\bn\in\mathbb Z_+^d} p\left(t,\bn\right)\delta_{\bn},
\]
where $\delta_{\bn}$ denotes the unit point mass at $\bn$. 
Consider the convolution of $B_{\mathbf{m}}$ with $\mu_t$,
\[
f_{\mathbf{m}}\left(\mathbf{x}\right)=\int_{\mathbb R_+^d} B_{\mathbf{m}}\left(\mathbf{x}-\mathbf{y}\right)\rmd\mu_t\left(\mathbf{y}\right)=\sum_{\mathbf{k}\in\mathbb Z_+^d} p\left(t,\mathbf{k}\right)\,B_{\mathbf{m}}\left(\mathbf{x}-\mathbf{k}\right).
\]
Evaluating at the midpoint $\tilde{\bb}$, and using that $\mathbf{k}\in\mathbb Z^d$, we have for each coordinate $j$,
\[
1_{\left(0,m_j\right]}\left(\tilde b_j-k_j\right)=1
\ \Longleftrightarrow\
0<b_j+\tfrac{1}{2}-k_j\le m_j=b_j-a_j+1
\ \Longleftrightarrow\
a_j\le k_j\le b_j.
\]
Hence, $B_{\mathbf{m}}\left(\tilde{\bb}-\mathbf{k}\right)=1$ if and only if $\ba\le \mathbf{k}\le \bb$ coordinatewise, and therefore,
\[
f_{\mathbf{m}}\left(\tilde{\bb}\right)=\sum_{\mathbf{k}} p\left(t,\mathbf{k}\right)\,B_{\mathbf{m}}\left(\tilde{\bb}-\mathbf{k}\right)
= P\left(a_1\le Z_1(t)\le b_1,\dots,a_d\le Z_d(t)\le b_d\right).
\]
By separability of $B_{\mathbf{m}}$ and of the Laplace kernel $e^{-\bs^\top \mathbf{x}}=\prod_{j=1}^d e^{-s_j x_j}$, and by Fubini--Tonelli, the $d$-dimensional Laplace transform satisfies
\[
\mathcal L\left\{f_{\mathbf{m}}\right\}\left(\bs\right)=\mathcal L\left\{B_{\mathbf{m}}\right\}\left(\bs\right)\,\mathcal L_{\bz(t)}\left(\bs\right),\qquad
\mathcal L\left\{B_{\mathbf{m}}\right\}\left(\bs\right)=\prod_{j=1}^d \frac{1-e^{-m_j s_j}}{s_j}.
\]
Inverting at $\tilde{\bb}$ gives the rectangular probability
\[
P\left(a_1\le Z_1(t)\le b_1,\dots,a_d\le Z_d(t)\le b_d\right)
=\mathcal L^{-1}\left\{\left(\prod_{j=1}^d \frac{1-e^{-m_j s_j}}{s_j}\right)\mathcal L_{\bz(t)}\left(\bs\right)\right\}\left(\tilde{\bb}\right).
\]
The unit box is the special case $\mathbf{m}=\mathbf{1}$, for which
\[
B\left(\mathbf{x}\right)=B_{\mathbf{1}}\left(\mathbf{x}\right)=\prod_{j=1}^d 1_{\left(0,1\right]}\left(x_j\right),\qquad
\mathcal L\left\{B\right\}\left(\bs\right)=\prod_{j=1}^d \frac{1-e^{-s_j}}{s_j}.
\]
Writing $f\left(\mathbf{x}\right)=\sum_{\mathbf{k}} p\left(t,\mathbf{k}\right)\,B\left(\mathbf{x}-\mathbf{k}\right)$ gives $f\left(\tilde{\bn}\right)=p\left(t,\bn\right)$, hence the joint $\pmf$
\[
p(t,\bn)=\mathcal L^{-1}\left\{\left(\prod_{j=1}^d \frac{1-e^{-s_j}}{s_j}\right)\mathcal L_{\bz(t)}\left(\bs\right)\right\}\left(\tilde{\bn}\right).
\]

\section{Simulations}
\label{sec:sim_appendix}

\subsection{Phase 1 model parameters}
\label{sec:sim_params}

All three diagnostic scenarios use the extended diamond graph with edges $1\to 2$, $2\to 3$, $2\to 4$, $3\to 5$, $4\to 5$, and mutation rates $0.01$, $0.015$, $0.01$, $0.0025$, $0.002$ respectively (in the same order). The cyclic variant additionally includes the feedback edge $5\to 2$ at rate $0.1$. The birth rates $\bm\alpha$, death rates $\bm\beta$, net growth rates $\bm\lambda = \bm\alpha - \bm\beta$, immigration rates $\bm\nu$, and initial counts $\bzs_0$ for types 1--5 are given in \tref{tab:sim_params}.

\begin{table}[!htbp]
\centering
\caption{Parameter values for the three root-regime scenarios.}
\label{tab:sim_params}
\small
\begin{tabular}{llccccc}
\hline
Scenario & Parameter & Type 1 & Type 2 & Type 3 & Type 4 & Type 5 \\
\hline
\multirow{5}{*}{Supercritical} &
  $\bm\alpha$ & 1.3 & 1.5 & 1.1 & 1.6 & 1.7 \\
& $\bm\beta$  & 1.0 & 1.0 & 0.9 & 0.8 & 1.1 \\
& $\bm\lambda$& 0.3 & 0.5 & 0.2 & 0.8 & 0.6 \\
& $\bm\nu$    & 0.3 & 0.7 & 0.8 & 0.1 & 0.2 \\
& $\bzs_0$    & 5   & 5   & 10  & 3   & 7   \\
\hline
\multirow{5}{*}{Critical} &
  $\bm\alpha$ & 1.0 & 1.5 & 1.1 & 1.6 & 1.7 \\
& $\bm\beta$  & 1.0 & 1.0 & 0.9 & 0.8 & 1.1 \\
& $\bm\lambda$& 0.0 & 0.5 & 0.2 & 0.8 & 0.6 \\
& $\bm\nu$    & 0.7 & 0.7 & 0.8 & 0.1 & 0.2 \\
& $\bzs_0$    & 5   & 5   & 10  & 3   & 7   \\
\hline
\multirow{5}{*}{Subcritical} &
  $\bm\alpha$ & 0.77 & 1.5 & 1.1 & 1.6 & 1.7 \\
& $\bm\beta$  & 1.0  & 1.0 & 0.9 & 0.8 & 1.1 \\
& $\bm\lambda$& $-0.23$ & 0.5 & 0.2 & 0.8 & 0.6 \\
& $\bm\nu$    & 0.0 & 0.7 & 0.8 & 0.1 & 0.2 \\
& $\bzs_0$    & 100 & 5   & 10  & 3   & 7   \\
\hline
\end{tabular}
\end{table}

\subsection{Phase 2 simulation setup}
\label{sec:sim_params_phase2}

For each combination of $N\in\{10,20,50\}$ and $D\in\{1,2,3\}$, we generate 50 independent DAGs. The vertices are ordered from 1 to $N$. For every $j=2,\ldots,N$, one parent is chosen uniformly from $\{1,\ldots,j-1\}$. The resulting $N-1$ edges form a spanning arborescence rooted at vertex~1. Consequently, vertex~1 is an ancestor of every other vertex, the DAG has a unique source, and each of the remaining $N-1$ vertices has a well-defined topological depth. The latter is required by our choice of evaluation nodes, which are selected by structural role. Let $M_N=N(N-1)/2$ be the number of possible forward edges. Each remaining forward edge is included independently with probability
\[
p_{N,D}=\frac{ND-(N-1)}{M_N-(N-1)},
\]
so the expected total number of edges is $ND$.

For each random DAG, the parameters are sampled independently. Immigration rates are $\nu_i\sim\mathrm{Unif}(0,1)$. Net growth rates are $g_i \sim \mathrm{Unif}(-0.3, 0.3)$, with $\alpha_i = 1 + g_i/2$ and $\beta_i = 1 - g_i/2$. Mutation rates are $\mu_{ij} \sim \mathrm{Unif}(10^{-4}, 10^{-3})$. Initial counts are sampled uniformly from the integers $\{1, \ldots, 100\}$.

\subsection{Accuracy and runtime metrics}
\label{sec:sim_metrics}

\paragraph{Absolute CDF error (ACE)}
For an evaluation triple $(i, t, n)$ of node type, time point, and population count,
the point-wise absolute error is $|F^m_i(t,n) - \hat{F}^{\mathrm{SSA}}_i(t,n)|$,
where $F^m_i(t,n)$ is the approximated $\cdf$ and $\hat{F}^{\mathrm{SSA}}_i(t,n)$ is
the empirical SSA $\cdf$.

Accuracy is summarised using a two-level aggregation in both experiments.
For the diagnostic experiment (\fref{fig:runtime_vs_accuracy}), evaluated at $t = 10$:
the per-type average ACE for type $i$ is the mean absolute error over all evaluation
counts $n$. The displayed point is the median of these per-type averages across all
five types, with interquartile-range bars spanning the 25th--75th percentiles across
types.
For the large-scale DAG experiments (\fref{fig:accuracy_vs_runtime_phase2}):
for each graph replicate and type $i$, the per-type average ACE is the mean absolute
error over all counts $n$ at $t = 10$. Each faint scatter dot is the median of these
per-type averages across types for one graph replicate. The bold point is the median of
those per-replicate values across graph replicates, with interquartile-range bars.

\paragraph{SSA reference ACE}
For the SSA reference curve at sample size $M$, let $p_k = F_i(t, n_k)$ be the true
$\cdf$\ value at grid point $k$. For $M$ independent trajectories,
$M\hat{F}_M(n_k) \sim \mathrm{Bin}(M, p_k)$, and the exact expected absolute error is
\[
  \mathbb{E}\bigl|\hat{F}_M(n_k) - F(n_k)\bigr|
  = 2p_k(1-p_k)\,\Pr\!\left(B_{M-1,\,p_k} = \lfloor Mp_k \rfloor\right),
\]
where $B_{M-1,p_k} \sim \mathrm{Bin}(M-1, p_k)$. The SSA reference ACE at sample size
$M$ is the mean of this quantity over the same evaluation grid used for the method
comparison. In practice, $p_k$ is estimated by the $M = 10^6$ empirical $\cdf$\ $\hat{F}_{10^6}(t, n_k)$, which itself carries a small Monte Carlo error of order $M^{-1/2}$. Asymptotically, $\mathbb{E}|\hat{F}_M - F| \approx \sqrt{2p(1-p)/(\pi M)}$,
so the SSA ACE scales as $M^{-1/2}$.

\paragraph{Runtime}
SP and LTSM were implemented in C++17 using the Eigen linear algebra library.
CLT is based on the open-source MATLAB implementation of \citet{gunnarsson2023statistical},
adapted here to compute marginal variances directly rather than the full covariance matrix.
All methods are timed at the level of a single $\cdf$\ evaluation on one CPU thread,
without intra-query parallelization, on the ETH Zurich Euler HPC cluster.
Runtime is summarised following the same two-level scheme as ACE.
For the diagnostic experiment, the per-type average runtime for type $i$ is the mean
wall-clock time per $\cdf$\ evaluation over all counts $n$ at $t = 10$. The displayed
point is the median of these per-type averages across types, with interquartile-range bars.
For the large-scale DAG experiments, per-type averages are computed over all counts $n$
at $t = 10$. The per-replicate summary is the median across types. The bold point is the
median across replicates, with interquartile-range bars.
For SSA, $10^4$ trajectories were timed on a single thread; the projected runtime for
$M$ trajectories is obtained by linear scaling $T(M) = T(10^4) \cdot M / 10^4$.

\subsection{Additional figures}
\label{sec:sim_appendix_figs}

\begin{figure}[H]
    \centering
    \includegraphics[width=\textwidth]{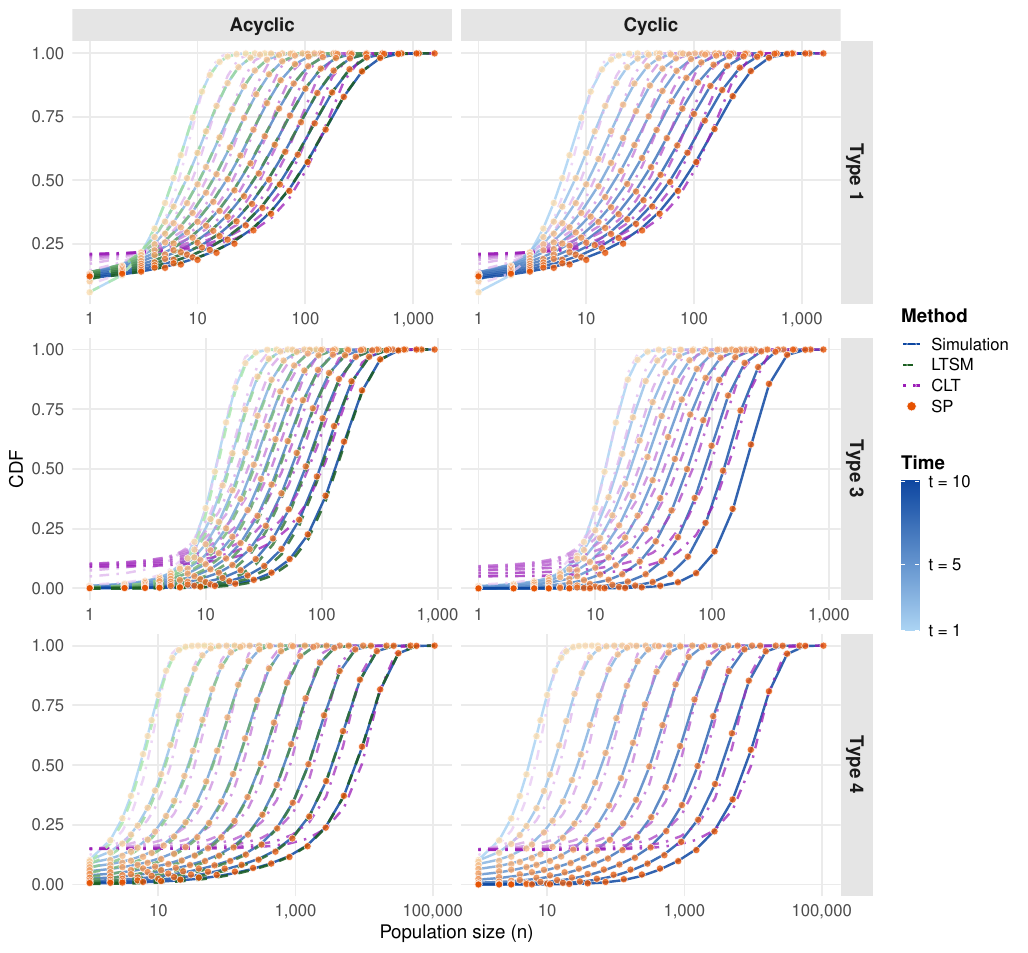}
    \caption{Marginal $\cdf$\ comparison for types~1 (root), 3, and 4 (branches) in the supercritical scenario, for the acyclic and cyclic variants}
    \label{fig:cdf_supercritical_appendix}
\end{figure}

\begin{figure}[H]
    \centering
    \includegraphics[width=\textwidth]{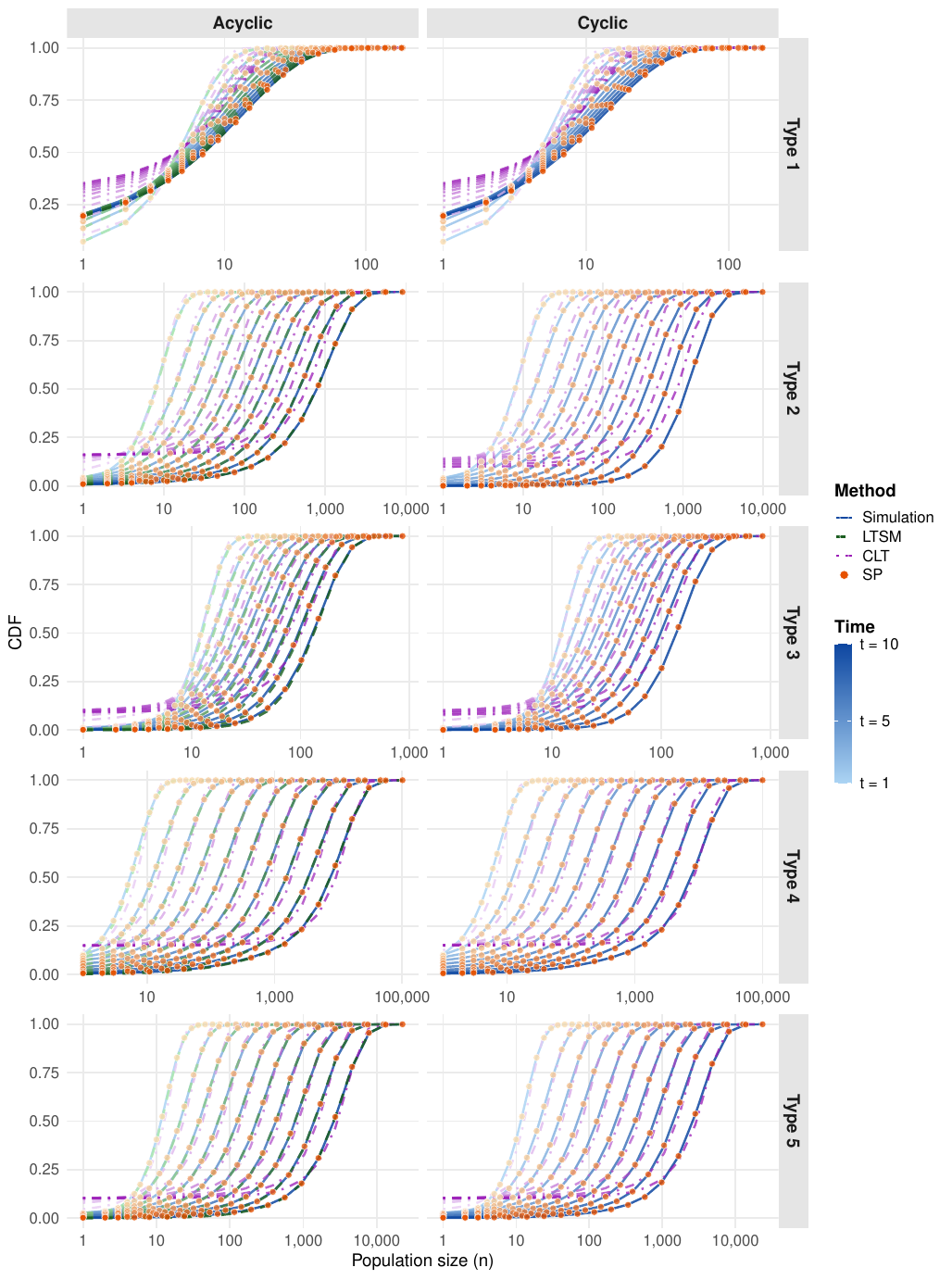}
    \caption{Marginal $\cdf$\ comparison for all five types in the critical root scenario ($\lambda_1 = 0$)}
    \label{fig:cdf_critical_appendix}
\end{figure}

\begin{figure}[H]
    \centering
    \includegraphics[width=\textwidth]{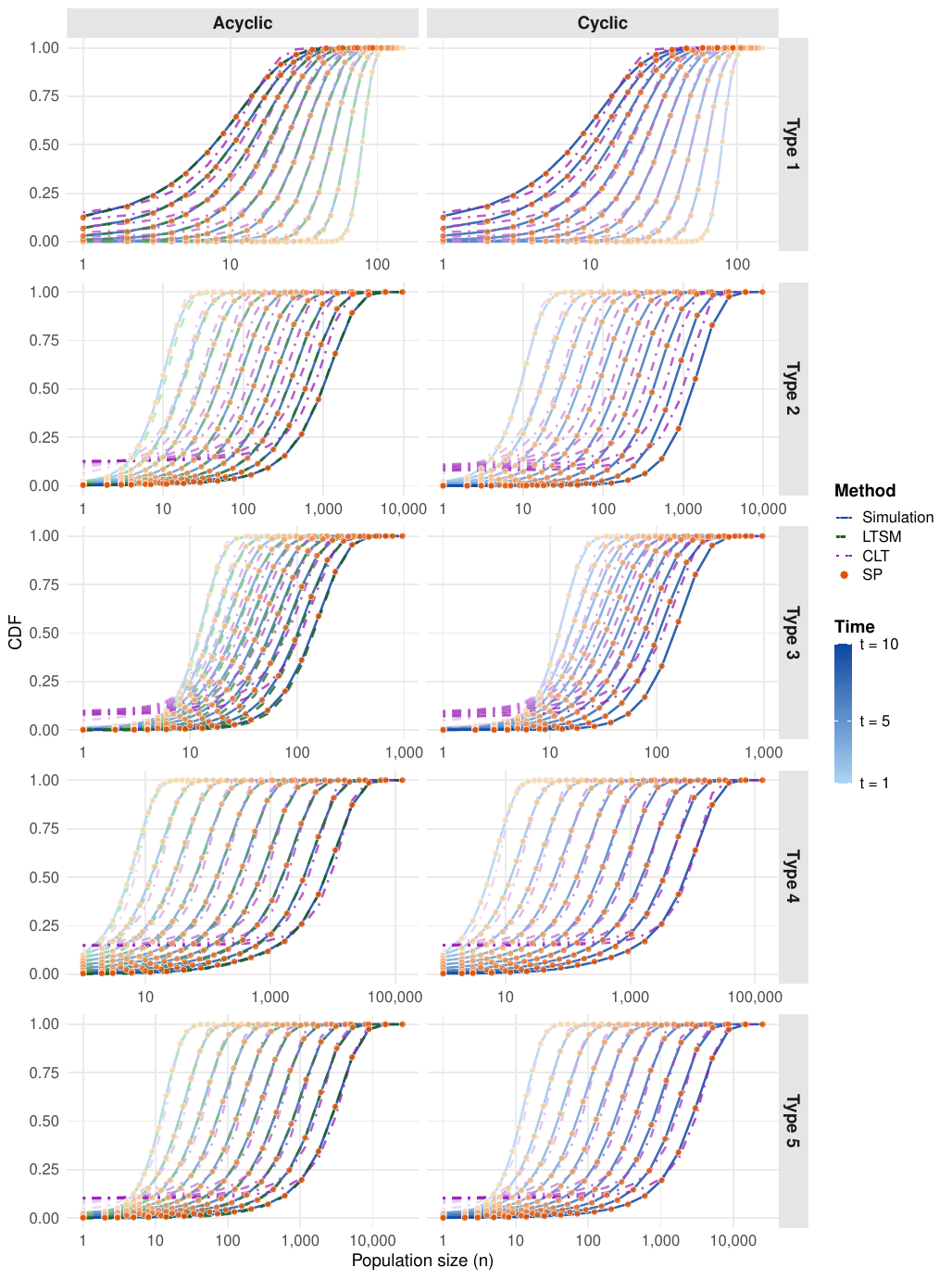}
    \caption{Marginal $\cdf$\ comparison for all five types in the subcritical root scenario ($\lambda_1 = -0.23$, $\nu_1 = 0$, $z_{1,0} = 100$)}
    \label{fig:cdf_subcritical_appendix}
\end{figure}

\begin{figure}[H]
    \centering
    \includegraphics[width=\textwidth]{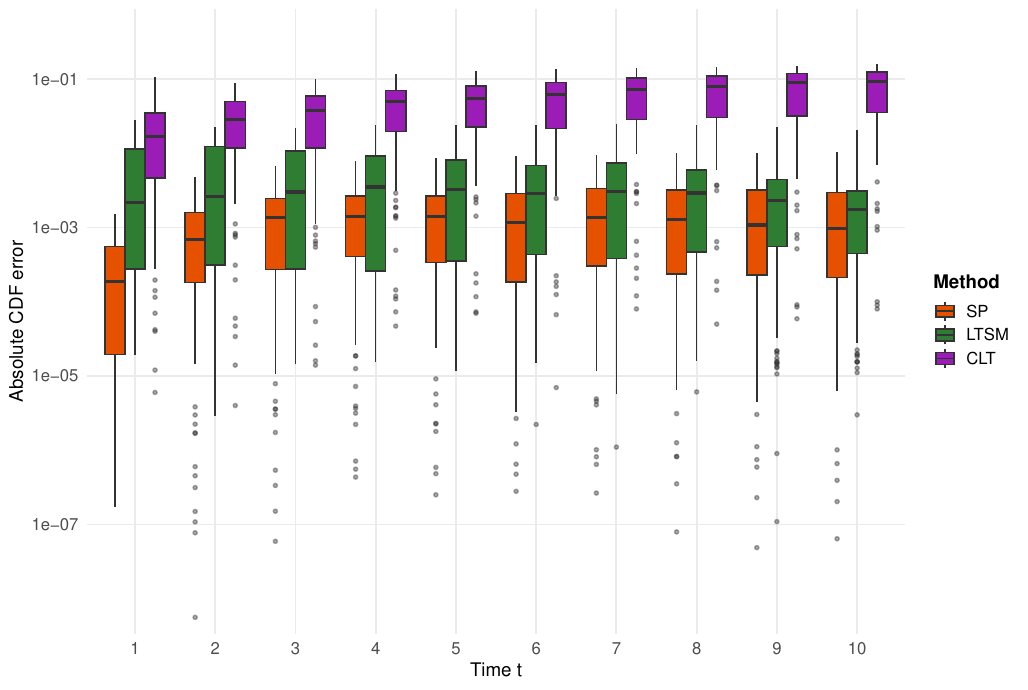}
    \caption{Time-resolved accuracy comparison for the supercritical scenario, acyclic variant. Each point in the boxplot shows the absolute $\cdf$\ error for one type at one time point, colored by methods. Boxplot elements show the median (center line), interquartile range (box), and 1.5$\times$ interquartile range (whiskers) across types at each time point}
    \label{fig:error_over_time}
\end{figure}

\begin{figure}[H]
    \centering
    \includegraphics[width=0.787\textwidth]{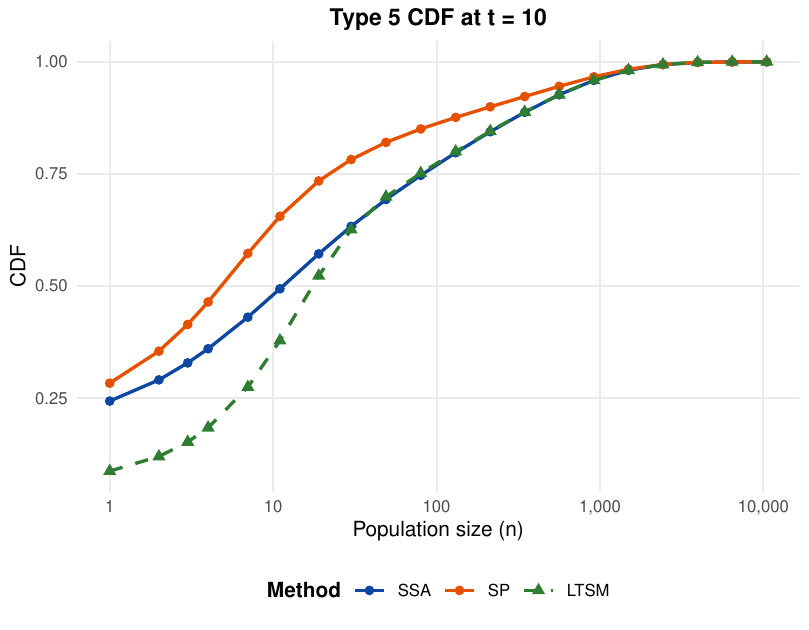}
    \caption{Marginal $\cdf$\ comparison for type~5 at $t = 10$ in the supercritical acyclic scenario, with initial condition set to $z_{5,0}=0$ and mutation rates divided by ten}
    \label{fig:ltsm_special_case_type5_t10_mutdiv10}
\end{figure}

\begin{figure}[H]
    \centering
    \includegraphics[width=\textwidth]{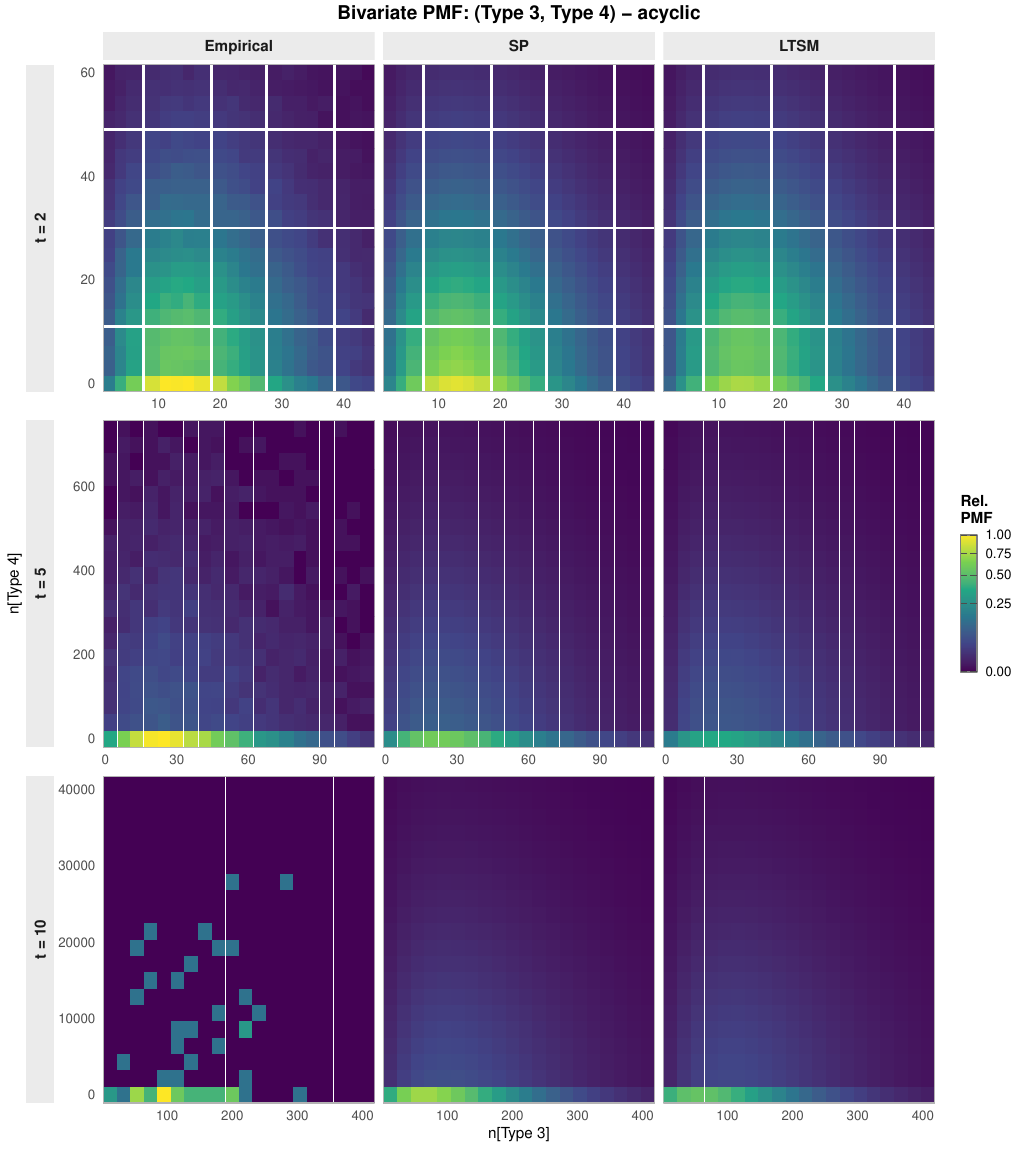}
    \caption{Bivariate $\pmf$\ comparison for the sibling pair $(Z_3, Z_4)$ at $t \in \{2, 5, 10\}$ (rows), supercritical acyclic scenario. Columns show the Gillespie reference, SP, and LTSM approximations. Color encodes probability mass; absolute-difference heatmaps are shown in the rightmost column}
    \label{fig:bivariate_23_acyclic}
\end{figure}

\begin{figure}[H]
    \centering
    \includegraphics[width=\textwidth]{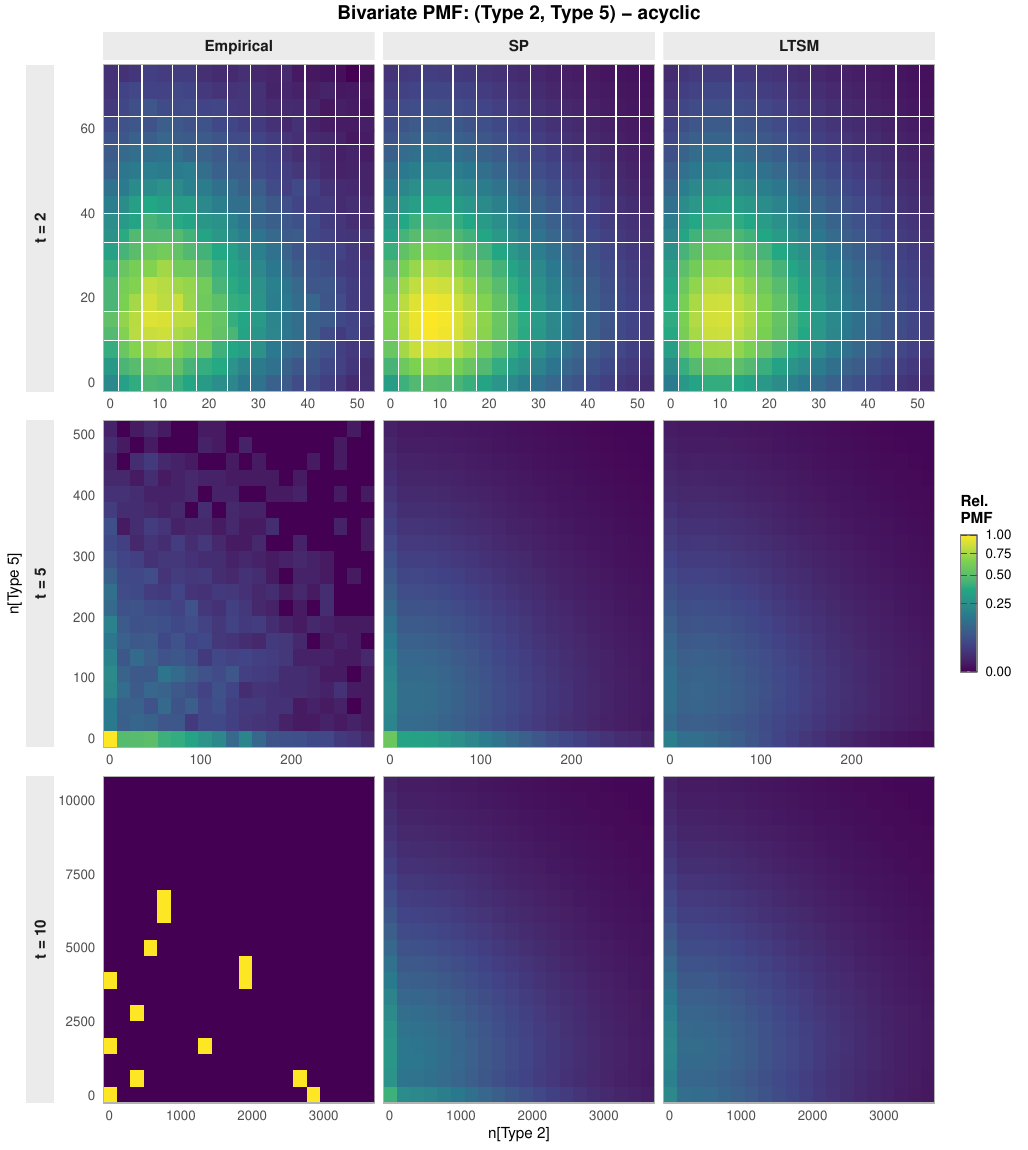}
    \caption{Bivariate $\pmf$\ comparison for the ancestor-descendant pair $(Z_2, Z_5)$ at $t \in \{2, 5, 10\}$, supercritical acyclic scenario. Format as in \fref{fig:bivariate_23_acyclic}}
    \label{fig:bivariate_14_acyclic}
\end{figure}

\begin{figure}[H]
    \centering
    \includegraphics[width=\textwidth]{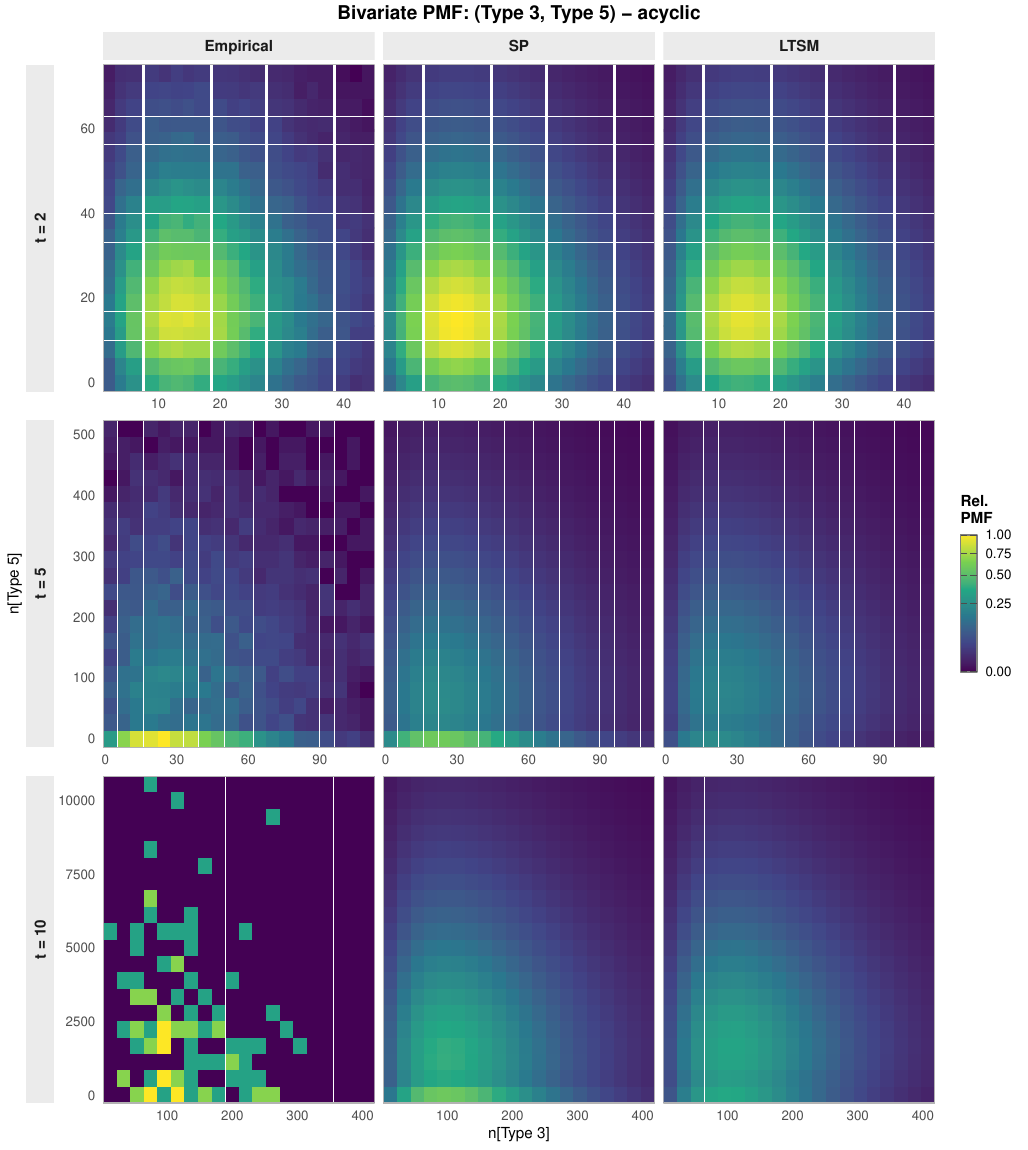}
    \caption{Bivariate $\pmf$\ comparison for the parent-child pair $(Z_3, Z_5)$ at $t \in \{2, 5, 10\}$, supercritical acyclic scenario. Format as in \fref{fig:bivariate_23_acyclic}}
    \label{fig:bivariate_24_acyclic}
\end{figure}

\begin{figure}[H]
    \centering
    \includegraphics[width=\textwidth]{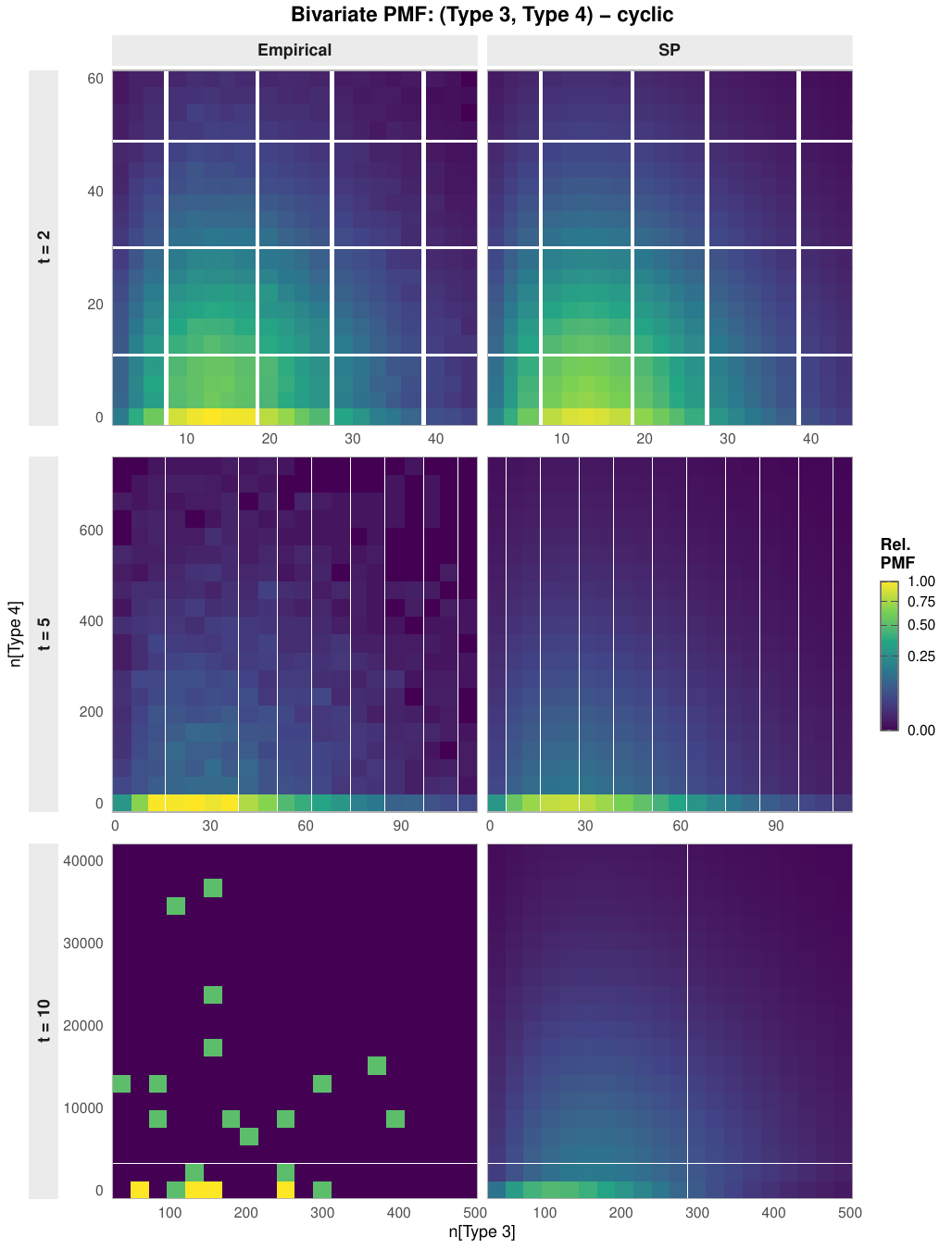}
    \caption{Bivariate $\pmf$\ comparison for the sibling pair $(Z_3, Z_4)$ at $t \in \{2, 5, 10\}$, supercritical cyclic scenario. Columns show the Gillespie reference and SP approximation; LTSM does not apply to cyclic graphs}
    \label{fig:bivariate_23_cyclic}
\end{figure}

\begin{figure}[H]
    \centering
    \includegraphics[width=\textwidth]{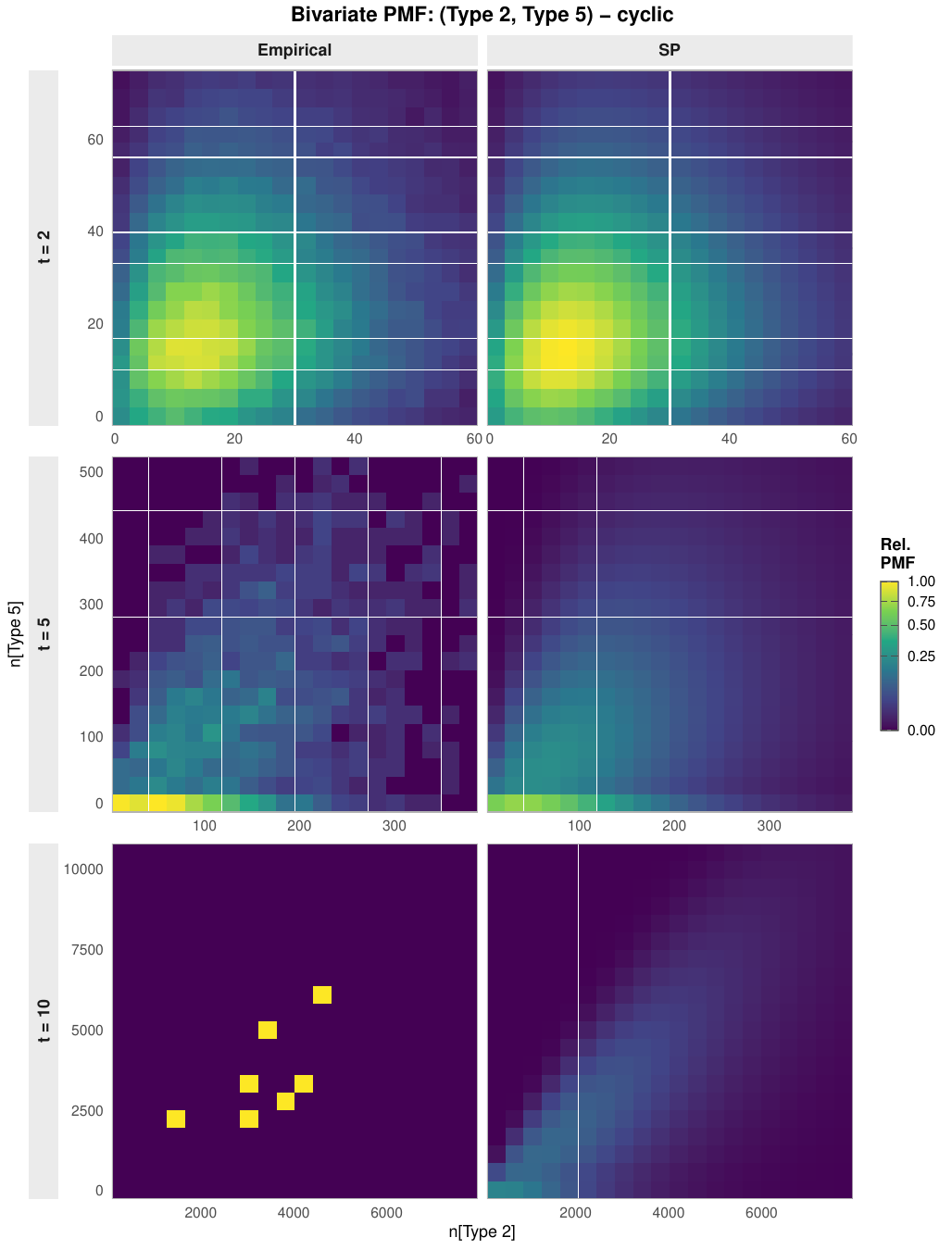}
    \caption{Bivariate $\pmf$\ comparison for the ancestor-descendant pair $(Z_2, Z_5)$ at $t \in \{2, 5, 10\}$, supercritical cyclic scenario. Columns show the Gillespie reference and SP approximation; LTSM does not apply to cyclic graphs}
    \label{fig:bivariate_14_cyclic}
\end{figure}

\begin{figure}[H]
    \centering
    \includegraphics[width=\textwidth]{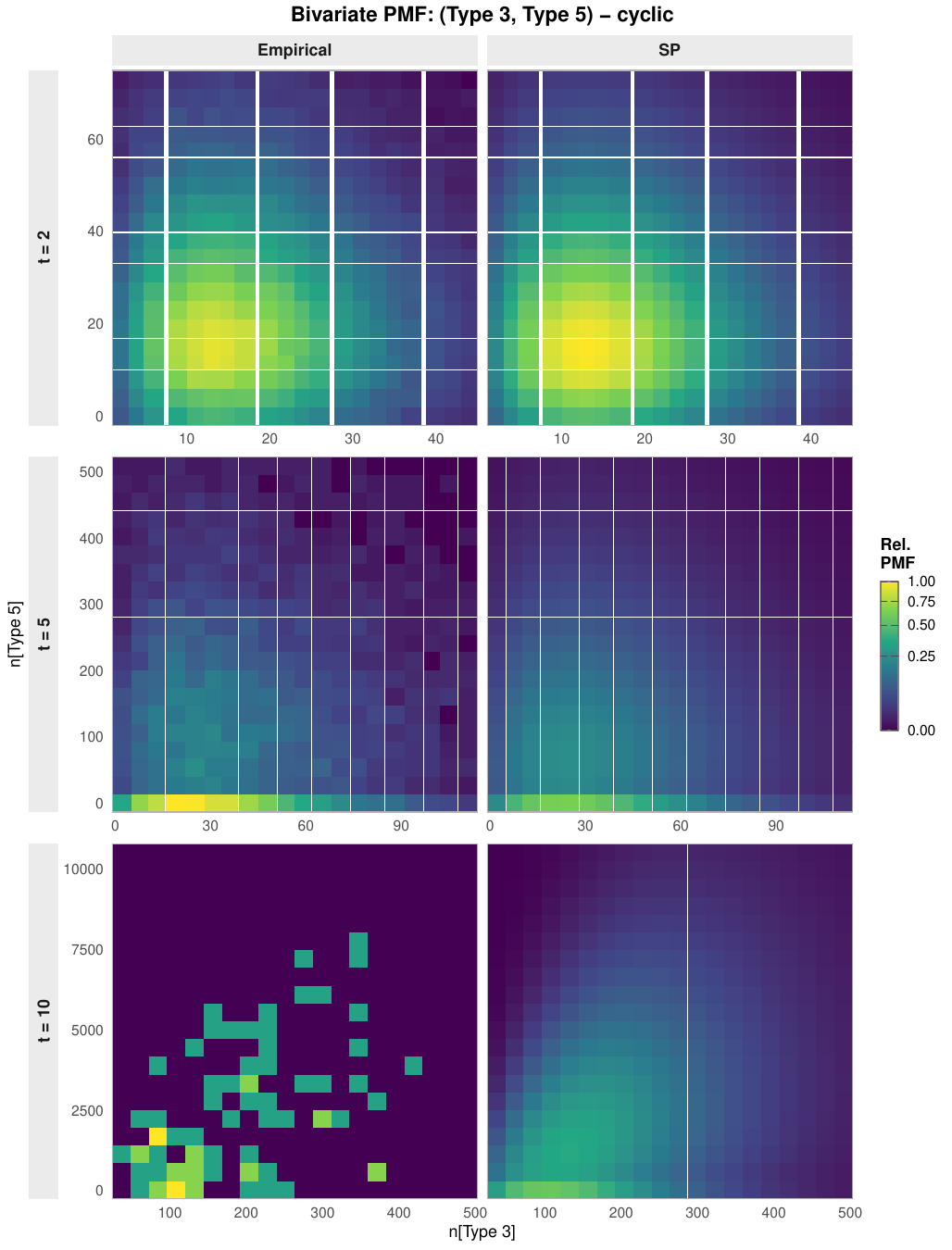}
    \caption{Bivariate $\pmf$\ comparison for the parent-child pair $(Z_3, Z_5)$ at $t \in \{2, 5, 10\}$, supercritical cyclic scenario. Columns show the Gillespie reference and SP approximation; LTSM does not apply to cyclic graphs}
    \label{fig:bivariate_24_cyclic}
\end{figure}

\section{AML-09 application}
\label{sec:aml09_appendix}

\paragraph{Data inputs}
The AML-09 mutation tree is the SCITE tree~\citep{jahn2016tree} reported by \citet{Morita2020ClonalGenomics}. It gives the mutation order and the single-cell subclone proportions at the observed diagnosis and relapse time points. The intervening remission state is unobserved. For confidentiality, the figure merges diagnosis and treatment at display time 0. It shows remission at 155 days and relapse 135 days later. For the branching-process calculation, we reset the model clock to $t=0$ at remission. By the Markov property, the relapse distribution depends on the unobserved remission state and the rounded remission-to-relapse interval
\[
  t_{\Delta}=135/365.25=0.370\ \mathrm{yr}.
\]

\paragraph{Clinical burden estimates}
At diagnosis, the total leukemic burden is estimated as in FiTree~\citep{Luo2025Bayesian},
\[
\begin{split}
C_{\mathrm{tumor}}
&= \mathrm{PB\_WBC}\times \mathrm{PB\_Blast}\times
   \text{adult blood volume} \\
&\quad + \mathrm{BM\_Blast}\times
   \text{adult bone marrow mononuclear cell count}.
\end{split}
\]
We use an adult blood volume of \SI{5e6}{\micro\liter} and an adult bone marrow mononuclear cell count of $7.53\times10^{11}$~\citep{bianconi2013estimation}. At relapse, only the bone-marrow blast percentage is available. We therefore estimate the peripheral-blood blast contribution from the diagnosis ratio between peripheral-blood and marrow blasts. The raw clinical fields are confidential and are not reported. The clinical measurements give total population-size estimates $C_{\mathrm{tumor,diag}}$ and $C_{\mathrm{tumor,rel}}$. We multiply these totals by the corresponding single-cell subclone proportions to obtain the diagnosis and relapse population-size estimates. Let $\pi_{v,\mathrm{rel}}$ denote the single-cell relapse proportion of subclone $v$. We denote the estimated relapse population size by
\[
  \hat z_{v,\mathrm{rel}}
  =C_{\mathrm{tumor,rel}}\pi_{v,\mathrm{rel}}.
\]
The unobserved remission burden is denoted $N_{\mathrm{rem}}$. The clinical complete-remission threshold gives the upper bound $N_{\mathrm{rem}}\le 0.05\times7.53\times10^{11}=3.765\times10^{10}$ cells~\citep{bisel1956criteria}. The inputs are summarized in \tref{tab:aml09_assumptions}.

\paragraph{Six-type branching process}
The root is a static wild-type compartment with $10^5$ cells. It supplies direct immigration only to the founding \gene{NPM1} p.L287fs subclone, at rate $\nu_{\gene{NPM1}}=10^5\mu_{\mathrm{root},\gene{NPM1}}=1.35\times10^{-3}\ \mathrm{yr}^{-1}$. All other direct immigration rates $\nu_v$ are zero. Mutation rates $\mu_{ij}$ on tree edges use the gene-level rates of \citet{Luo2025Bayesian}. Death rates are fixed at $\beta_v=1\ \mathrm{yr}^{-1}$ for all subclones.

\paragraph{Remission initialization}
The sorafenib-sensitive subclones \gene{FLT3}-ITD and \gene{FLT3} p.D835E are initialized at zero at remission~\citep{Smith2015FLT3}. Let
\[
  S=\{\gene{NPM1}\ \mathrm{p.L287fs},\ \gene{KRAS}\ \mathrm{p.G13D},\ \gene{FLT3}\ \mathrm{p.D835Y},\ \gene{WT1}\ \mathrm{p.P372fs}\}
\]
be the non-sensitive subclones. The remission proportions are obtained by renormalizing their diagnosis frequencies,
\[
  \pi_{v,\mathrm{rem}}
  =
    \frac{\pi_{v,\mathrm{diag}}}{\sum_{u\in S}\pi_{u,\mathrm{diag}}},
  \qquad v\in S.
\]
Thus, at model time $t=0$, the initial state is
\[
  \bz(0)=\bzs_0,
  \qquad
  z_{0,v}=z_{v,\mathrm{rem}}
  =\left\lfloor N_{\mathrm{rem}}\pi_{v,\mathrm{rem}}\right\rfloor,
  \qquad v\in S.
\]
The sorafenib-sensitive subclones have $z_{0,v}=0$.

For the \gene{WT1}-specific calculation in panel (c) of \fref{fig:aml09sensitivity}, we vary its latent diagnosis proportion $\pi_{\mathrm{WT1},\mathrm{diag}}$. We scan $\pi_{\mathrm{WT1},\mathrm{diag}}=0$ and positive values corresponding to at least one latent \gene{WT1} diagnosis cell up to the 0.1\% detection limit~\citep{Morita2020ClonalGenomics}.

\begin{table}[!t]
\centering
\begingroup
\small
\renewcommand{\arraystretch}{1.15}
\begin{tabular}{>{\raggedright\arraybackslash}p{0.20\textwidth}>{\raggedright\arraybackslash}p{0.20\textwidth}>{\raggedright\arraybackslash}p{0.52\textwidth}}
\hline
Symbol & Value used & Description \\
\hline
$C_{\mathrm{tumor,diag}}$ & $6.0\times10^{11}$ cells & Total leukemic cells at diagnosis, estimated from peripheral-blood and bone-marrow blast data. \\
$C_{\mathrm{tumor,rel}}$ & $5.0\times10^{11}$ cells & Total leukemic cells at relapse, estimated from bone-marrow blast percentage and the diagnosis peripheral-blood to bone-marrow blast ratio. \\
$N_{\mathrm{rem}}$ & $3.765\times10^{10}$ cells at the ceiling & Total leukemic cells at remission. The displayed ceiling is $0.05\times7.53\times10^{11}$. \\
$t_{\Delta}$ & $135/365.25=0.370$ yr & Rounded remission-to-relapse interval. \\
$\mu_{ij}$ & gene-level rates & Mutation rates assigned to tree edges. The incoming \gene{NPM1} rate is multiplied by $10^5$ static wild-type cells to produce immigration. The incoming edge rates are listed by subclone in \tref{tab:aml09_clones}. \\
$\beta_v$ & $1\ \mathrm{yr}^{-1}$ & Death rate for every modeled subclone. \\
$F=(f_{ij})$ & posterior median & FiTree fitness matrix inferred from diagnosis samples across 123 AML patients. \\
$\lambda_{v,\mathrm{diag}}$ & $\exp(\varphi_v^0)-1$ & Diagnosis net growth rate implied by the FiTree matrix and $\beta_v=1$. \\
$m_{\mathrm{shared}}$ & 0 to 500, with 52.716 for panel (c) & Shared multiplicative factor for relapse net growth rates. \\
$m_{\mathrm{WT1}}$ & 1 to 8 & Additional \gene{WT1}-specific multiplier in panel (c) of the main results figure. \\
\hline
\end{tabular}
\endgroup
\caption{Data and model inputs for the AML-09 application.}
\label{tab:aml09_assumptions}
\end{table}

\paragraph{Growth parameters}
FiTree's tree-structured population process is a special case of the model in \sref{sec:model}~\citep{Luo2025Bayesian}. In FiTree, all mutant subclones start with zero cells, and the static wild-type source acts through $\nu_{\gene{NPM1}}$. FiTree uses a large-time small-mutation-rate approximation to infer fitness from the pre-treatment diagnosis samples. Here, we retain the posterior-median fitness matrix as a fixed, cohort-derived pre-treatment reference for the diagnosis net growth rates. FiTree's full generative model also contains a tumor-sampling mechanism, which is not part of this special-case relation. The finite-time saddle-point calculation instead starts from the structured, generally nonzero remission state $\bzs_0$.

For a subclone $v$, the genotype $g_v$ contains the genes accumulated along the tree path from the root to $v$. The FiTree diagnosis fitness matrix $F=(f_{ij})$ gives
\[
  \varphi_v^0
  = \max_{i\in g_v} f_{ii}
    + \sum_{\{i,j\}\subset g_v,\ i\neq j} f_{ij},
  \qquad
  \lambda_{v,\mathrm{diag}}=\exp(\varphi_v^0)-1.
\]
Across post-treatment scenarios, only the net growth rate, and therefore the replication rate, is varied. The immigration, death, and mutation rates remain fixed at their baseline values. Within each parameter setting, all rates are held fixed between remission and relapse. The shared factor $m_{\mathrm{shared}}$ scales diagnosis net growth rates as
\[
  \lambda_{v,\mathrm{rel}}=m_{\mathrm{shared}}\lambda_{v,\mathrm{diag}},
  \qquad
  \alpha_{v,\mathrm{rel}}=1+\lambda_{v,\mathrm{rel}}.
\]
For the \gene{WT1}-specific calculation, $N_{\mathrm{rem}}$ is fixed at the 5\% remission ceiling and $m_{\mathrm{shared}}=52.716$ is chosen by matching the deterministic six-type mean of \gene{FLT3} p.D835Y to its estimated relapse count. The \gene{WT1} net growth rate then receives an additional multiplier,
\[
  \lambda_{\mathrm{WT1},\mathrm{rel}}
  = m_{\mathrm{WT1}}m_{\mathrm{shared}}\lambda_{\mathrm{WT1},\mathrm{diag}}.
\]
Subclone-specific proportions and rates are listed in \tref{tab:aml09_clones}.

\paragraph{Probability calculation}
For \gene{NPM1} p.L287fs and \gene{FLT3} p.D835Y, we evaluate the marginal exceedance probability
\[
  P(Z_v(t_{\Delta})\ge \hat z_{v,\mathrm{rel}}
  \mid N_{\mathrm{rem}},m_{\mathrm{shared}})
\]
over $\log_{10}N_{\mathrm{rem}}=0,0.25,\ldots,10.5$, the exact 5\% remission ceiling, and $m_{\mathrm{shared}}=0,1,\ldots,500$. For \gene{WT1} p.P372fs, we fix $N_{\mathrm{rem}}$ and $m_{\mathrm{shared}}$ as above and evaluate
\[
  P(Z_{\mathrm{WT1}}(t_{\Delta})\ge \hat z_{\mathrm{WT1},\mathrm{rel}}
  \mid \pi_{\mathrm{WT1},\mathrm{diag}},m_{\mathrm{WT1}})
\]
over latent diagnosis proportions below 0.1\% and $m_{\mathrm{WT1}}=1,1.05,\ldots,8$.
\begin{table}[H]
\centering
\begingroup
\scriptsize
\renewcommand{\arraystretch}{1.5}
\begin{adjustbox}{width=\textwidth}
\begin{tabular}{>{\raggedright\arraybackslash}p{0.18\textwidth}rrrrrrr}
\hline
Subclone & Diagnosis & Relapse & $\lambda_{v,\mathrm{diag}}$ & $\alpha_{v,\mathrm{diag}}$ & $\beta_v$ & $\mu_{\operatorname{pa}(v),v}$ & $\nu_v$ \\
\hline
\gene{NPM1} p.L287fs & 14.9\% & 5.7\% & 0.2787 & 1.2787 & 1 & $1.350\times10^{-8}$ & $1.350\times10^{-3}$ \\
\gene{FLT3}-ITD & 69.4\% & 0\% & 0.2567 & 1.2567 & 1 & $1.297\times10^{-7}$ & 0 \\
\gene{FLT3} p.D835E & 8.6\% & 0\% & 0.2567 & 1.2567 & 1 & $1.297\times10^{-7}$ & 0 \\
\gene{FLT3} p.D835Y & 1.7\% & 86.5\% & 0.2567 & 1.2567 & 1 & $1.297\times10^{-7}$ & 0 \\
\gene{KRAS} p.G13D & 5.4\% & 0\% & 0.2533 & 1.2533 & 1 & $4.925\times10^{-8}$ & 0 \\
\gene{WT1} p.P372fs & 0\% & 7.8\% & 0.1755 & 1.1755 & 1 & $9.685\times10^{-8}$ & 0 \\
\hline
\end{tabular}
\end{adjustbox}
\endgroup
\caption{Subclone-specific inputs parsed from the AML-09 mutation tree and FiTree fitness matrix. Rates are in $\mathrm{yr}^{-1}$.}
\label{tab:aml09_clones}
\end{table}

\end{document}